\documentclass[journal]{IEEEtran}
\usepackage{cite}

\ifCLASSINFOpdf
\else
\fi
\usepackage{array}
\usepackage{fixltx2e}
\usepackage{url}

\usepackage{amssymb,amsmath,latexsym,amsfonts,mathrsfs}

\usepackage[active]{srcltx}  
\usepackage[cp1251]{inputenc}
\usepackage[english]{babel}
\usepackage{setspace}
\usepackage{natbib}
\usepackage{textcomp}
\usepackage{dsfont}
\usepackage{graphics}
\usepackage{graphicx}
\usepackage{makeidx}
\usepackage{amsbsy}
\usepackage{verbatim}
\usepackage{authblk}
\usepackage{bm}
\usepackage{bbm}
\usepackage{rotating}
\usepackage{lipsum}
\usepackage{adjustbox}
\usepackage[font=scriptsize,labelfont=bf]{caption}
\usepackage{caption}
\usepackage{subfloat}

\makeindex

\bibpunct{(}{)}{,}{a}{,}{,}
 	\newcommand{\bs}[1]{\boldsymbol{#1}}

        \newcommand{\bSigma}{\bs{\Sigma}}
        \newcommand{\bi}{\mathbf{1}}
        \newcommand{\bS}{\hat{\bs{\Sigma}}}
        \newcommand{\bI}{\mathbf{I}}
        \newcommand{\bx}{\mathbf{Z}}
 	
 	\newtheorem{thm}{\bf Theorem}
 	\newtheorem{theorem}[thm]{\bf Theorem}
 	\newtheorem{lemma}[thm]{\bf Lemma}
 	
	\newtheorem{proposition}{\bf Proposition}

 	\newenvironment{proof}{\smallskip\noindent{\it Proof\ }\rm}
 	{\hspace*{\fill} $\Box$ \medskip}

 	\newenvironment{theorem*}[1]{\smallskip\noindent{\bf #1 \ }\rm}
 	{\medskip}

\usepackage{color}

\usepackage{ulem}

\begin{document}
 \title{Tests for the weights of the global minimum variance portfolio in a high-dimensional setting}
 \date{}
 \author[a]{{Taras Bodnar}}
 \author[b]{{Solomiia Dmytriv}}
 \author[c]{Nestor Parolya\thanks{Corresponding Author: Nestor Parolya. E-Mail: N.Parolya@tudelft.nl}}
 \author[b]{{Wolfgang Schmid}\thanks{\textcircled{c} 2019 IEEE.  Personal use of this material is permitted.  Permission from IEEE must be obtained for all other uses, in any current or future media, including reprinting/republishing this material for advertising or promotional purposes, creating new collective works, for resale or redistribution to servers or lists, or reuse of any copyrighted component of this work in other works.}}

 \affil[a]{{\footnotesize Department of Mathematics, Stockholm University, Stockholm, Sweden}}
 \affil[b]{{\footnotesize Department of Statistics, European University Viadrina, Frankfurt(Oder), Germany}}
 \affil[c]{{\footnotesize Delft Institute of Applied Mathematics, Delft University of Technology, The Netherlands}}

\markboth{IEEE Transactions of Signal Processing,~2019~forthcoming...}%
{Bodnar \MakeLowercase{\textit{et al.}}: Tests for the weights of the global minimum variance portfolio in a high-dimensional setting}
 
 \maketitle

 \begin{abstract} In this study, we construct two tests for the weights of the global minimum variance portfolio (GMVP) in a high-dimensional setting, namely, when the number of assets $p$ depends on the sample size $n$ such that $\frac{p}{n}\to c \in (0,1)$ as $n$ tends to infinity. In the case of a singular covariance matrix with rank equal to $q$ we assume that $q/n\to \tilde{c}\in(0, 1)$ as $n\to\infty$. The considered tests are based on the sample estimator and on the shrinkage estimator of the GMVP weights. We derive the asymptotic distributions of the test statistics under the null and alternative hypotheses. Moreover, we provide a simulation study where the power functions and the receiver operating characteristic curves of the proposed tests are compared with other existing approaches. We observe that the test based on the shrinkage estimator performs well  even for values of $c$ close to one.
\end{abstract}


\section{Introduction} Financial markets have developed rapidly in recent years, and the amount of money invested in risky assets has substantially increased. Due to this, an investor must have knowledge of optimal portfolio proportions in order to receive a large expected return and, at the same time, to reduce the level of the risk associated with the investment decision.

Since \citet*{markowitz1952portfolio} presented his mean-variance analysis, many works about optimal portfolio selection have been published. However, investors are faced with some difficulties in the practical implementation of these investing theories since sampling error is present when unknown theoretical quantities are estimated.

 In classical asymptotic analysis, it is almost always assumed that the sample size increases while the size of the portfolio, namely the number of included assets $p$, remains constant (e.g., \citet*{korkie1981}, \citet*{okhrin2006distributional}). Nowadays, this case is often called standard asymptotics (see, \citet*{cam2000asymptotics}). Here, the traditional plug-in estimator of the optimal portfolio, the so-called sample estimator, is consistent and asymptotically normally distributed. However, in many applications, the number of assets in a portfolio is large in comparison to the sample size (i.e., the portfolio dimension $p$ and the sample size $n$ tend to infinity simultaneously) such that $\displaystyle \frac{p}{n}$ tends to the concentration ratio $c>0$. In this case, we are faced with so-called high-dimensional asymptotics or `Kolmogorov' asymptotics (see, \citet*{buhlmann2011statistics}, \citet*{bai2011estimating}, \citet*{cai2011analysis}, \citet*{bodnardetteparolya2019}). Whenever the dimension of the data is large, the classical limit theorems are no longer suitable because the traditional estimators result in a serious departure from the optimal estimators under high-dimensional asymptotics (\citet*{bai2010spectral}). These methods fail to provide consistent estimators of the unknown parameters of the asset returns, that are, the mean vector and the covariance matrix. Generally, the greater the concentration ratio $c$, the worse the sample estimators are. In these cases, new test statistics must be developed, and completely new asymptotic techniques must be applied for their derivations. Several studies deal with high-dimensional asymptotics in portfolio theory using results from random matrix theory (see, \citet*{FrahmRMT} and \citet*{LalouxRMT}). Recently, \citet*{bodnar2014estimation} presented a shrinkage-type estimator for the global minimum variance portfolio (GMVP) weights, and \citet*{bodnar2016okhrin} derived the optimal shrinkage estimator of the mean-variance portfolio.

 Testing the efficiency of a portfolio is a classical problem in finance. What looks good theoretically often suffers from the curse of uncertainty and dimensionality. Nevertheless, some approaches provide effective portfolio choice strategies including the GMVP, which by construction is a mixture of assets that minimizes the portfolio variance/volatility. The success of this strategy violates modern portfolio theory because it takes only the portfolio variance into account. But many empirical studies show that portfolios that focus on minimizing the volatility generate superior out-of-sample results (see, \citet{JagannathanMa2003, Clarke31, Ledoit110, Clarke10} among others). That is why it makes sense to provide a statistical test whether the current portfolio composition is different from the conventional GMVP taking into account both the uncertainty of the asset returns and the large dimensionality of the portfolio.

 The former literature focuses on the case of standard asymptotics or considers exact tests where both $p$ and $n$ are fixed. For example, \citet*{gibbons1989test} provided an exact $F$-test for the efficiency of a given portfolio, and \citet*{britten1999sampling} derived inference procedures on the efficient portfolio weights based on the application of linear regression. More recently, \citet*{bodnar2008test} presented a test for the general linear hypothesis of the portfolio weights in the case of elliptically contoured distributions.
 The contribution of this study is the derivation of statistical techniques for testing the efficiency of a portfolio under high-dimensional asymptotics. Two statistical tests are considered. Whereas the first approach is based on the asymptotic distribution of the test statistic suggested by \citet*{bodnar2008test} in a high-dimensional setting, the second test makes use of the shrinkage estimator of the GMVP weights and provides a powerful alternative to the existing methods. To the best of our knowledge, this analysis is the first time that the shrinkage approach has been applied to statistical test theory.

 It has to be mentioned that there is a direct link between the subject of the paper and classical methods in statistical signal processing. The equivalent of the GMVP portfolio in signal processing literature is the Capon or minimum variance spatial filter (see, \citet*{verdu1998} and \citet*{vantrees2002}). The estimation risk of the high-dimensional minimum variance beamformer has already been studied in \citet*{rubioetal2012} while its constrained versions were discussed in \citet*{LiStoicaWang2004}. The finite sample size effect on minimum variance filter was investigated by \citet*{MestreLaugunas2006}. An improved calibration of the precision matrix, i.e., the central object for constructing the GMVP portfolio, was discussed in \cite{Palomar2013}. For more literature on the applications of the random matrix theory to signal processing and portfolio optimization see, \cite{PalomarBook2016} and references therein.

The testing procedure we propose can be used not only for testing on the GMV portfolio but also for the inference on the shrinkage intensity, i.e., the level of shrinkage one needs to decrease the estimation risk of the GMVP. Our test is based on the shrinkage technique for GMVP weights and, thus, setting different shrinkage targets leads to different tests, which could be of independent interest for financial analysts. As an example, one could construct a test whether the GMVP portfolio is stochastically dominating a naive (equally weighted) portfolio, which has attracted much attention of financial scientists during the last decade (see, \citet{DeMiguel2009a, DeMiguel2009b}).

The paper is structured as follows. In Section II, we discuss the main results on distributional properties for optimal portfolio weights presented by \citet*{okhrin2006distributional}. In Section III.A the high-dimensional version of the test based on the test statistics given in \citet*{bodnar2008test} is proposed, while a new test based on the shrinkage estimator for the GMVP weights is derived in Section III.B. The asymptotic distributions of the test statistics under both the null hypothesis and the alternative hypothesis are obtained, and the corresponding power functions of both tests are presented. In Section III.C, new test procedures for the GMVP weights are proposed under a high-dimensional setting when the covariance matrix is singular. In Section IV, the power functions and the receiver operating characteristic curves of the proposed tests are compared with each other for different values of $c\in (0, 1)$. In our comparison study, a test of \citet*{glombeck} is considered as well. We conclude in Section V. All proofs are given in the Appendix.

\section{Estimation of Optimal Portfolio Weights}
We consider a financial market consisting of $p$ risky assets. Let $\textbf{X}_t$ denote the $p$-dimensional vector of the returns on risky assets at time $t$. Suppose that $E(\textbf{X}_t)=\boldsymbol{\mu}$ and $Cov(\textbf{X}_t)=\mathbf{\Sigma}$. The covariance matrix $\mathbf{\Sigma}$ is assumed to be positive definite.

Let us consider a single period investor who invests in the GMVP, one of the most commonly used portfolios (see, for example, \citet*{memmel2006estimating}, \citet*{frahm2010dominating}, \citet*{okhrin2006distributional}, \citet*{bodnar2008test}, \citet*{glombeck}, and others). This portfolio exhibits the smallest attainable portfolio variance $\mathbf{w}'\mathbf{\Sigma} \mathbf{w}$ under the constraint $\textbf{w}'\mathbf{1}=1$, where $\mathbf{1}=(1,\ldots,1)'$ denotes the $p$-dimensional vector of ones and $\mathbf{w}$ stands for the vector of portfolio weights. The weights of GMVP are given by
	\begin{equation} \label{GMVP weights}
\mathbf{w}_{GMVP}=\frac{\mathbf{\Sigma}^{-1}\mathbf{1}}{\mathbf{1}'\mathbf{\Sigma}^{-1}\mathbf{1}}.
\end{equation}

The global minimum variance portfolio is of fundamental interest in applications involving array signal processing. In the array processing literature it is the so-called minimum variance distortionless response (MVDR) spatial filter or beamformer defined as $\mathbf{w}_{MVDR}=\frac{\mathbf{\Sigma}^{-1}\textbf{s}}{\textbf{s}^H\mathbf{\Sigma}^{-1}\mathbf{s}}$ (see, e.g., \citet*{vantrees2002}, Chapter 6). The vector $\mathbf{s}\in\mathbbm{C}^p$ is the scalar signature vector associated with some waveform $s\in\mathbbm{C}$. Thus, the tests for the global minimum variance portfolio developed in this paper could directly be used for minimum variance beamformer just by a simple modification.

The practical implementation of the mean-variance framework in the spirit of \citet*{markowitz1952portfolio} relies on estimating the first two moments of the asset returns.
Because we do not know the true covariance matrix, it is usually replaced by its sample estimator, which is based on a sample of $n>p$ historical asset returns  $\textbf{X}_{1},\ldots, \textbf{X}_{n} $ given by
{\small
\begin{equation}
	\mathbf{\hat{\Sigma}}_{n}= \frac{1}{n-1}\sum_{j=1}^{n}\left(\textbf{X}_{j}-\mathbf{\bar{X}}_{n}\right)\left(\textbf{X}_{j}-\mathbf{\bar{X}}_{n}\right)'\, \textrm{with  } \, \mathbf{\bar{X}}_{n}=\frac{1}{n}\sum_{v=1}^{n}\textbf{X}_{v}.
	\end{equation}
}
	Replacing $\mathbf{\Sigma}$ in (\ref{GMVP weights}) by the sample estimator $\mathbf{\hat{\Sigma}}_{n}$, we obtain an estimator of the GMVP weights expressed as
	\begin{equation}
	\label{weights}
		\mathbf{\hat{w}}_{n}=\frac{\mathbf{\hat{\Sigma}}_{n}^{-1}\mathbf{1}}{\mathbf{1}'\mathbf{\hat{\Sigma}}_{n}^{-1}\mathbf{1}}.
	\end{equation}
Note that the estimator of the GMVP weights is exclusively a function of the estimator $\mathbf{\hat{\Sigma}}_{n} $ of the covariance matrix.
	
Assuming that the asset returns $\{\textbf{X}_{t}\}$ follow a stationary Gaussian process with mean $\boldsymbol{\mu}$ and covariance matrix $\mathbf{\Sigma}$, \citet*{okhrin2006distributional} proved that the vector of estimated optimal portfolio weights is asymptotically normal. Under the additional assumption of independence, they derived the exact distribution of $\mathbf{\hat{w}}_{n}$. \citet*{okhrin2006distributional} showed that the distribution of arbitrary $p-1$ components of $\mathbf{\hat{w}}_{n}$ is a $(p-1)$- dimensional $t$-distribution with $n-p+1$ degrees of freedom and
{
\begin{eqnarray*}\label{basic}
&&  E(\mathbf{\hat{w}}_{n})= \mathbf{w}_{GMVP},\\ && Cov(\mathbf{\hat{w}}_{n})=\mathbf{\Omega}= \frac{1}{n-p-1}\frac{\mathbf{Q}}{\mathbf{1}'\mathbf{\Sigma}^{-1}\mathbf{1}},\\
&&	\mathbf{Q}=\mathbf{\Sigma}^{-1}-\frac{\mathbf{\Sigma}^{-1}\mathbf{1}\mathbf{1}'\mathbf{\Sigma}^{-1}}{\mathbf{1}'\mathbf{\Sigma}^{-1}\mathbf{1}}\,.
\end{eqnarray*}
}
Consequently, if $\mathbf{\hat{w}}^{*}_{n}$ and $\mathbf{w}_{GMVP}^*$ are obtained by deleting the last element of $\mathbf{\hat{w}}_{n}$ and $\mathbf{w}_{GMVP}$ and if $\mathbf{\Omega}^{*}$ and $\mathbf{Q}^{*}$ consist of the first $(p-1)\times(p-1)$ elements of $\mathbf{\Omega}$ and $\mathbf{Q}$, then $\mathbf{\hat{w}}^*_{n}$ has a $(p-1)$-variate t-distribution with $n-p+1$ degrees of freedom and  parameters $\mathbf{w}^*_{GMVP}$ and $ \displaystyle \frac{1}{n-p+1}\frac{\mathbf{Q}^*}{\mathbf{1}'\mathbf{\Sigma}^{-1}\mathbf{1}}$. This distribution is denoted by $ \displaystyle \mathbf{\hat{w}}^*_{n}\sim t_{p-1}(n-p+1,\mathbf{w}^*_{GMVP} ,\frac{n-p-1}{n-p+1} \mathbf{\Omega}^*)$, since $ \displaystyle \frac{n-p-1}{n-p+1}    \mathbf{\Omega}^*=\frac{1}{n-p+1}\frac{\mathbf{Q}^*}{\mathbf{1}'\mathbf{\Sigma}^{-1}\mathbf{1}}$.

\section{Test Theory for the GMVP in High Dimensions}
	
At each time point, an investor is interested to know whether the portfolio he is holding coincides with the true GMVP or has to be reconstructed.  For that reason, we consider the following testing problem:
	 \begin{equation}\label{hypotheses}
     H_{0}:\mathbf{w}_{GMVP}=\mathbf{r}\qquad \textrm{against} \qquad H_{1}:\mathbf{w}_{GMVP}\not= \mathbf{r},
	 \end{equation}
where $\mathbf{r}$ with $\mathbf{r}' \mathbf{1}=1 $ is a known vector of, for example, the weights of the holding portfolio. Thus, this problem analyses whether the true GMVP weights are equal to some given values.

\citet*{bodnar2008test} analysed a general linear hypothesis for the GMVP portfolio weights and introduced an exact test assuming that the asset returns are independent and elliptically contoured distributed. Moreover, they derived the exact distribution of the test statistic under the null hypothesis and the alternative hypothesis.

The main focus of this study is high-dimensional portfolios. We want to consider the testing problem (\ref{hypotheses}) in a high-dimensional environment, that is, assuming that $\displaystyle \frac{p}{n}\to c \in (0,1)$ as $n \to \infty$. Note that, in this case, $H_0$ and $H_1$ depend on $n$ as well. Thus, it would be more precise to write $H_{0,n}: \textbf{w}_{GMVP,n}^* = \textbf{r}^*_n$ and $H_{1,n}: \textbf{w}_{GMVP,n}^* \neq \textbf{r}^*_n$. In the following, we will ignore this fact in order to simplify our notation. Moreover, it turns out that the sample covariance matrix is no longer a good estimator of the covariance matrix (see, \citet*{bai2010spectral, bai2011estimating, yao_zheng_bai_2015}). Indeed, the latter references reveal that if $p/n\to c\in (0, 1)$ and the covariance matrix is $\bSigma=\bI$ then the empirical spectral distribution of the eigenvalues of the sample covariance matrix $\mathbf{\hat{\Sigma}}_{n}$ is supported on $\left( (1-\sqrt{c})^2, (1+\sqrt{c})^2 \right)$. As a result, the larger $p/n$, the more  the eigenvalues spread out. It implies in terms of the $L_2$ norm that $\mathbf{\hat{\Sigma}}_{n}$ is not consistent.

For that reason, it is unclear how well the test of \citet*{bodnar2008test} behaves in that context. First, we study its behaviour under the high-dimensional asymptotics, and, after that, we propose an alternative test that makes use of the shrinkage estimator for the portfolio weights (cf. \citet*{bodnar2014estimation}).

In recent years, several studies have dealt with estimators of unknown portfolio parameters under high-dimensional asymptotics with applications to portfolio theory. \citet*{glombeck} formulated tests for the portfolio weights, variances of the excess returns, and Sharpe ratios of the GMVP for $c \in (0,1)$. \citet*{bodnar2014estimation} and \citet*{bodnar2016okhrin} derived the shrinkage estimators for the GMVP and for the mean-variance portfolio, respectively, under the Kolmogorov asymptotics for $c \in (0,\infty)$.

\subsection{A Test Based on the Mahalanobis Distance}

\citet*{bodnar2008test} proposed a test for a general linear hypothesis of the weights of the global minimum variance portfolio. Here, we are interested in the special case (\ref{hypotheses}). For this case, the test statistic is given by

\begin{equation}\label{Bodnar_Schmid}
T_{n}=\frac{n-p}{p-1}(\mathbf{1}'\hat{\mathbf{\Sigma}}^{-1}\mathbf{1})(\mathbf{\hat{w}}^{*}_{n}-\mathbf{r}^{*})'(\mathbf{\hat{Q}}_n^*)^{-1}(\mathbf{\hat{w}}^{*}_{n}-\mathbf{r}^{*}),
\end{equation}
 where $\hat{\mathbf{Q}}_n^*$ consists of the first $(p-1)\times(p-1)$ elements of $\hat{\mathbf{Q}}_n=\hat{\mathbf{\Sigma}}_n^{-1}-\displaystyle\hat{\mathbf{\Sigma}}_n^{-1}\mathbf{1}
 \mathbf{1}'\hat{\mathbf{\Sigma}}_n^{-1}/\mathbf{1}'\hat{\mathbf{\Sigma}}_n^{-1}\mathbf{1}$ and the number of assets $p$ in the portfolio is fixed. It was shown that $T_n$ has a central $F$-distribution with $p-1$ and $n-p$ degrees of freedom under the null hypothesis, i.e., $T_n\sim F_{p-1,n-p}$. Moreover, the density of $T_n$ under the alternative hypothesis $H_1$ is equal to
{\footnotesize
\begin{eqnarray}\label{lambda_exact}
  f_{T_n}(x) & = & f_{p-1,n-p}(x) \; (1 + \lambda)^{-(n-1)/2}  \nonumber\\
             &  \times &  _2F_1\left( \frac{n-1}{2}, \frac{n-1}{2}, \frac{p-1}{2}; \frac{(p-1)x}{n-p+(p-1)x} \frac{\lambda}{1+\lambda} \right),
\end{eqnarray}
}
where
\begin{equation}\label{lambda}
\lambda =  \mathbf{1}' \mathbf{\Sigma}^{-1} \mathbf{1} (\mathbf{w}^*_{GMVP} - \mathbf{r}^*)'
(\mathbf{Q}^*)^{-1} (\mathbf{w}^*_{GMVP} - \mathbf{r}^*)
\end{equation}
and $_2F_1$ stands for the hypergeometric function (see, \citet*{abramowitz+stegun}, chap. 15), that is,
\[ _2F_1(a,b,c; x) = \frac{\Gamma(c)}{\Gamma(a) \Gamma(b)} \sum_{i=0}^\infty \frac{\Gamma(a+i) \Gamma(b+i)}{\Gamma(c+i)} \frac{z^i}{i!} . \]

Thus, the exact power function of the test is given by
\begin{equation}\label{G}
  G(\lambda,p,n)=1-\int_{f_{1-\alpha;p-1,n-p}}^{\infty}f_{T_n}(x) \mbox{d}x\,,
\end{equation}
where $f_{1-\alpha;p-1,n-p}$ denotes the $(1-\alpha)$ quantile from the central $F$-distribution with $p-1$ and $n-p$ degrees of freedom. Note that this result is also valid for matrix-variate elliptically contoured distributions (see, \citet*{bodnar2008test}). On the other hand, several computational difficulties appear when the power function of the test is calculated for large values of $p$ and $n$, since doing so involves a hypergeometric function whose computation is very challenging for large values of  $p$ and $n$. In order to deal with this problem, we derive the asymptotic distribution of $T_n$ in a high-dimensional setting. This result is given in Theorem \ref{th1}. The proof is in the Appendix. Since $\lambda$ depends on $p$ (i.e., on $n$) through $\mathbf{\Sigma}$, we write $\lambda_n$ in the rest of the paper.

\begin{theorem}\label{th1}
	Let $p\equiv p(n)$ and $c_n = \frac{p}{n} \to c \in (0,1)$. Assume that $\{\mathbf{X}_t\}$ is a sequence of independent and normally distributed $p$-dimensional random vectors with mean $\boldsymbol{\mu}$ and covariance matrix $\mathbf{\Sigma}$, which is assumed to be positive definite.
	 Let
 	\begin{equation*}
		C^2_n= 	2+2\frac{\lambda_n^2}{c}+4\frac{\lambda_n}{c}+2\frac{c}{1-c}\left(1+\frac{\lambda_n}{c}\right)^2.
		\end{equation*}
	Then, it holds that
\begin{equation*}
\sqrt{p-1}\left(\frac{T_n-1-\lambda_n\frac{n-1}{p-1}}{C_n}\right) \stackrel{d}{\to} \mathcal{N}\left(0, 1\right)
\end{equation*}
for $p/n\to c\in (0,1)$ as $n\to \infty$. Under the null hypothesis, $\sqrt{p-1} \; ( T_n - 1 ) \stackrel{d}{\to} \mathcal{N}\left(0, {2}/{(1-c)}\right)$ for $p/n\to c\in (0,1)$ as $n\to \infty$.
\end{theorem}

The results of Theorem \ref{th1} lead to an asymptotic expression of the power function given by
\begin{eqnarray}
&&P\left(  \frac{\sqrt{p-1}\left(T_n-1\right)}{\sqrt{2 /(1-c)}} > z_{1-\alpha}\right) \nonumber \\
 & = &
1 - P\Bigg(  \frac{\sqrt{p-1}\left(T_n-1-\lambda_n\frac{n-1}{p-1}\right)}{C_n} \nonumber\\
&\le&  \frac{\sqrt{\frac{2}{(1-c)}} z_{1-\alpha} - \frac{\sqrt{p-1}\lambda_n(n-1)}{p-1}}{C_n} \Bigg) \nonumber \\
& \approx & 1 - \Phi\left(\frac{\sqrt{2/(1-c)} z_{1-\alpha} - \sqrt{p-1}\frac{\lambda_n}{c}}{C_n} \right) \label{as_power_Tn} ,
\end{eqnarray}
where $z_{1-\alpha}$ is the $(1-\alpha)$-quantile of the standard normal distribution.

In Figure \ref{fig:fig1}, we plot the power function \eqref{as_power_Tn} as a function of $\lambda_n$ for several values of $c$ and $n$ (solid line). In addition, the empirical power of the test is shown for the same values of $c$ and $n$ (dashed line) and is equal to the relative number of rejections of the null hypothesis obtained via a simulation study. It is remarkable that, following the proof of Lemma \ref{lem0}, the considered simulation study can be considerably simplified. Instead of generating a $p \times n$ random matrix of asset returns in each simulation run, we simulate four independent random variables from standard univariate distributions and then compute the statistic $T_n$ for the given value of $\lambda_n$ following the stochastic representation \eqref{stoch_pres_Tn} in the Appendix. Namely, the simulation study is performed in the following way:
\begin{itemize}
\item[(i)] Generate four independent random variables $\omega_1^{(b)}\sim \mathcal{N}(0,1)$, $\xi_2^{(b)}\sim \chi^2_{n-p}$, $\xi_3^{(b)}\sim \chi^2_{n-1}$, and $\xi_4^{(b)}\sim \chi^2_{p-2}$
\item[(ii)] For fixed $\lambda_n$, compute
\[T_n^{(b)} \stackrel{d}{=} \frac{n-p}{p-1} \frac{(\sqrt{\lambda_n \xi_3^{(b)}}+\omega_1^{(b)})^2+\xi_4^{(b)}}{\xi_2^{(b)}}\]
\item[(iii)] Repeat steps (i) and (ii) for $b=1,...,B$, where $B$ is the number of independent repetitions and calculate the empirical power by
\begin{equation}
\hat{P}=\frac{1}{B} \sum_{b=1}^{B} \mathds{1}_{(z_{1-\alpha},+\infty)}\left( \frac{\sqrt{p-1}\left(T_n^{(b)}-1\right)}{\sqrt{2 /(1-c)}}\right),
\end{equation}
where $\mathds{1}_{\mathcal{A}}(.)$ is the indicator function of the set $\mathcal{A}$.
\end{itemize}
In Figure \ref{fig:fig1}, we observe a good performance of the asymptotic approximation of the power function. This approximation works almost perfectly for both small and large values of $c$.
	\vspace{-0.2cm}
	
		\begin{figure}[th!]
		\centering
		\begin{tabular}{cc}
			\includegraphics[width=0.45\linewidth]{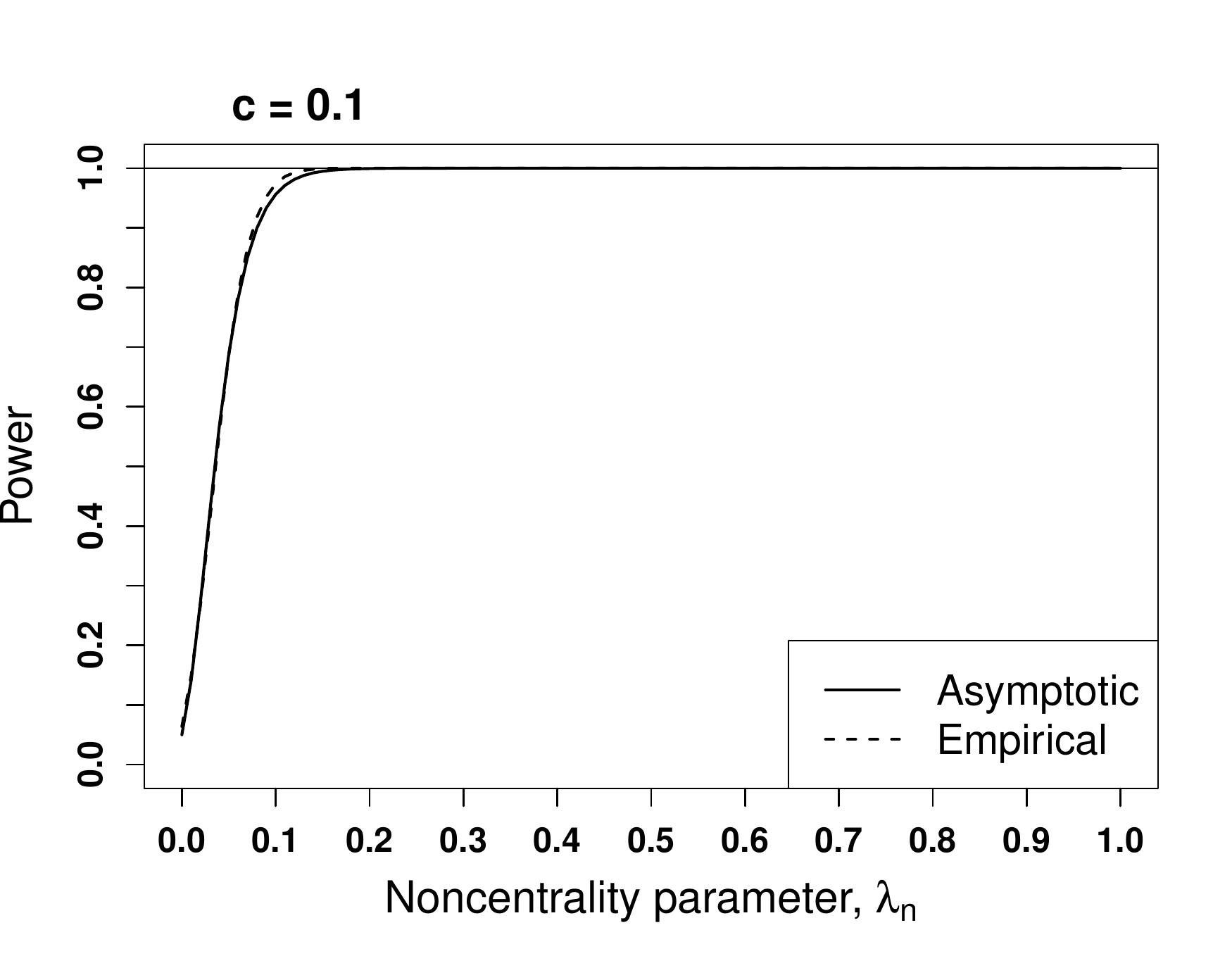}	\includegraphics[width=0.45\linewidth]{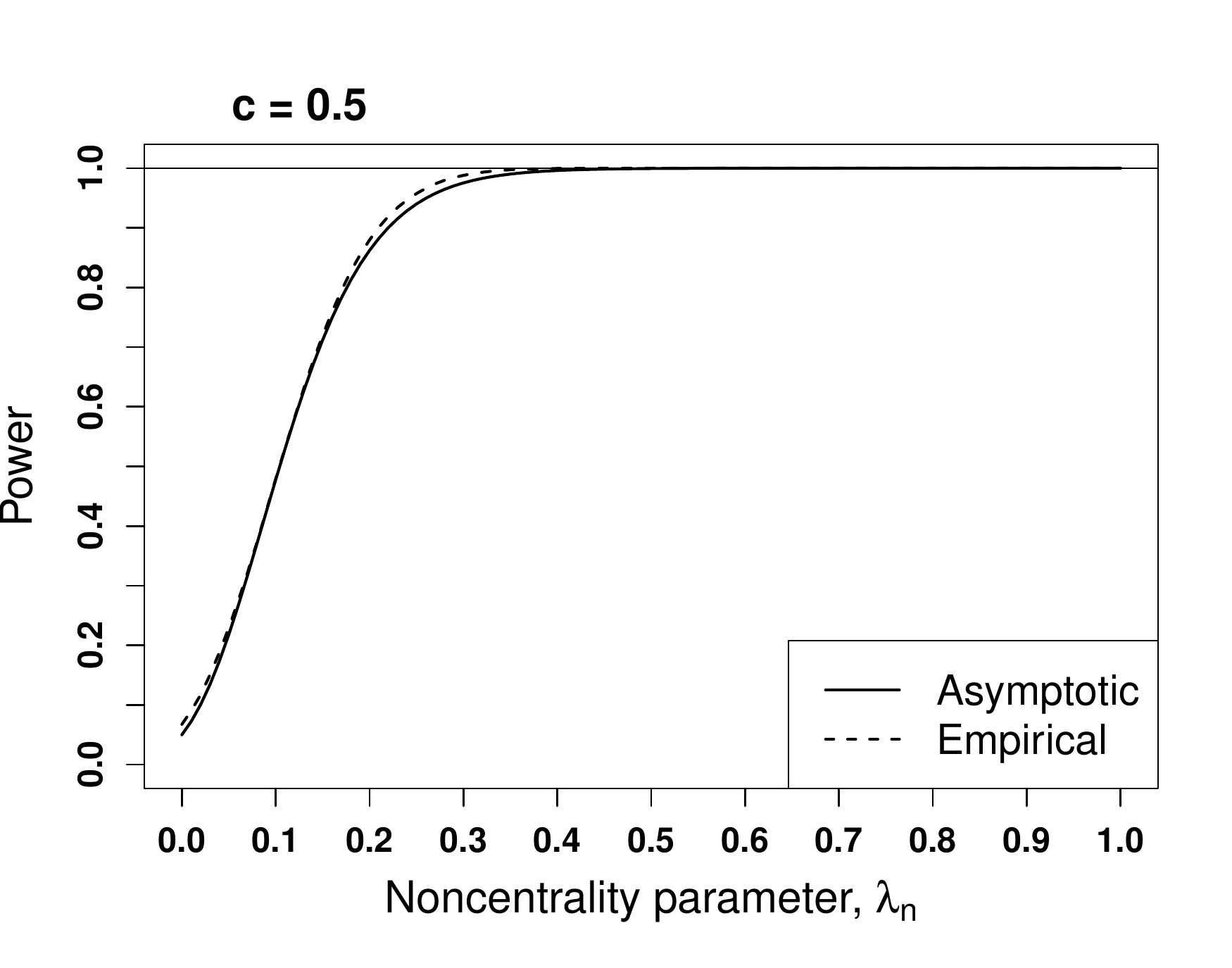}\\
		
			\includegraphics[width=0.45\linewidth]{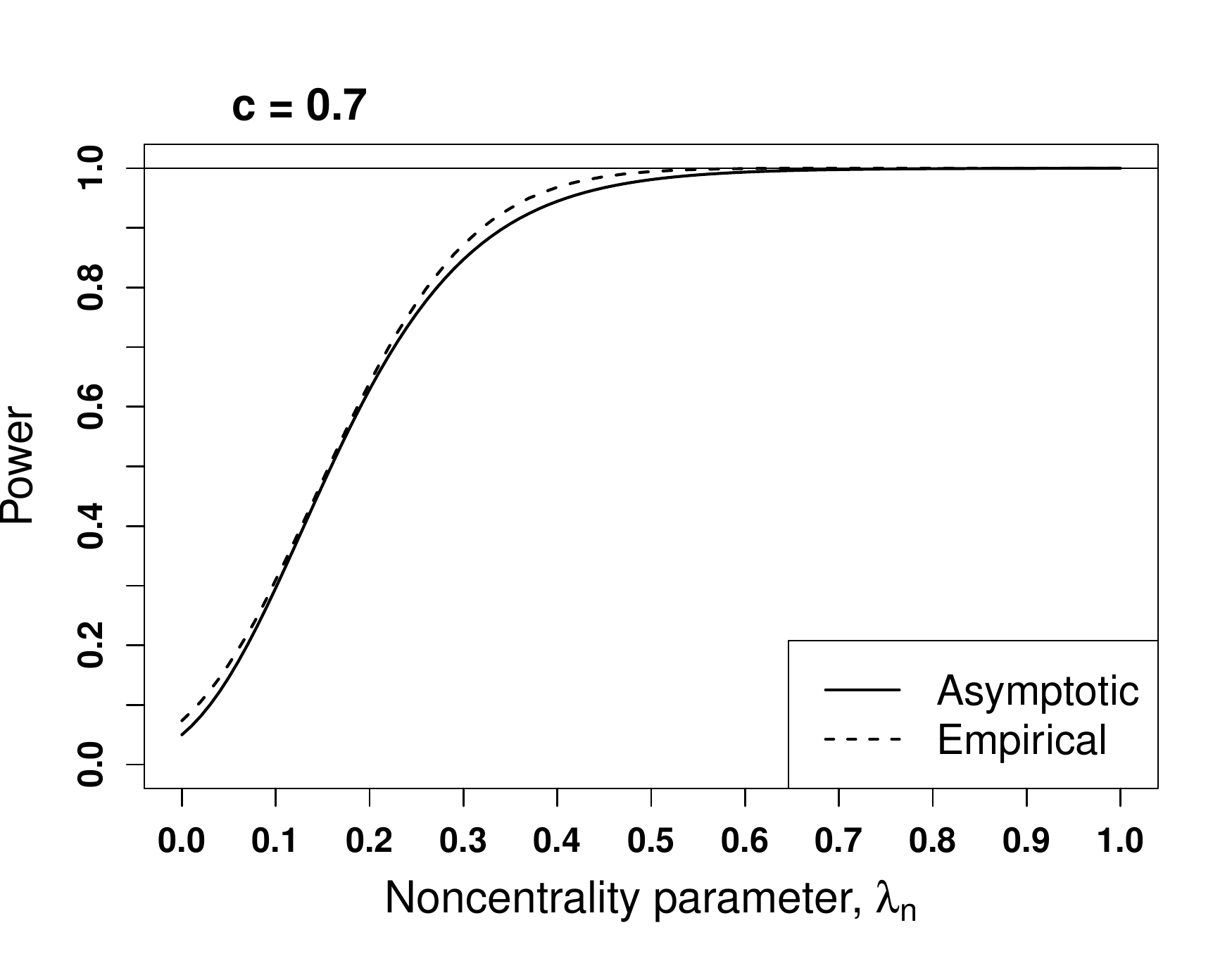}	\includegraphics[width=0.45\linewidth]{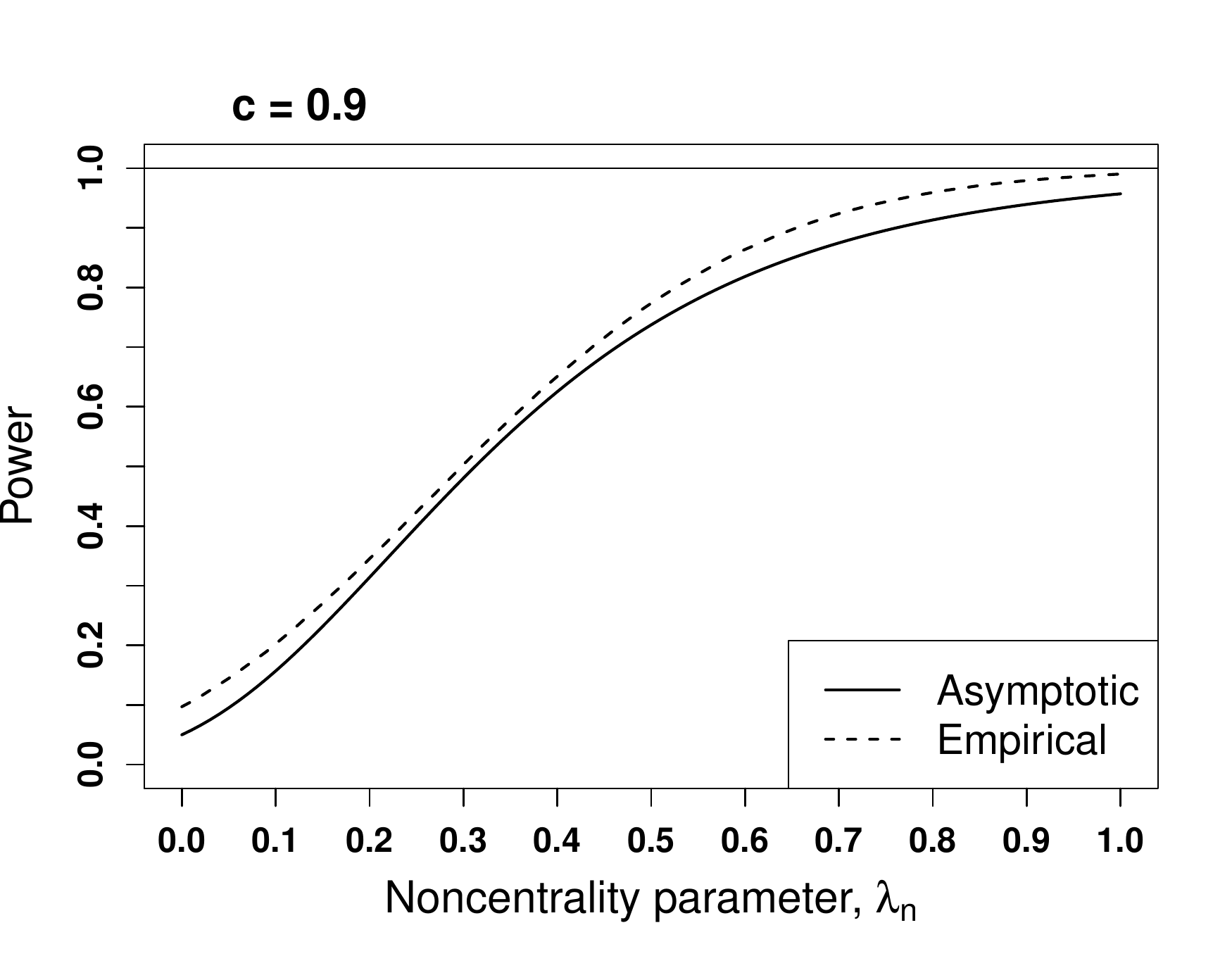}
			
		\end{tabular}
	\caption{Asymptotic power function (solid line) vs. empirical power function (dashed line) for the test problem in \eqref{hypotheses} as functions of $\lambda_n$ for various values of $c \in \{0.1, 0.5, 0.7, 0.9\}$ and $n=500$. The nominal significance level of the test (the probability of a type I error) is $\alpha=5\%$.}
\label{fig:fig1}	
\end{figure}	
	
\subsection{Test Based on a Shrinkage Estimator}

In most cases, the unknown parameters of the asset return distribution are replaced by their sample counterparts when an optimal portfolio is constructed. In recent years, however, other types of estimators, such as shrinkage estimators, have been discussed as well (see, \citet*{Okhrin2007} and \citet*{bodnar2014estimation}). The shrinkage methodology was introduced by \citet*{stein1956}. His results were extended by \citet*{efron1976} to the case in which the covariance matrix is unknown. The shrinkage methodology can be applied to the expected asset returns (e.g., \citet*{jorion1986bayes}) and the covariance matrix (\citet*{BodnarGuptaParolya2014, BodnarGuptaParolya2016}). Both of these applications appear
to be very successful in reducing damaging influences on the portfolio selection. A shrinkage estimator was applied directly to the portfolio weights by \citet*{golosnoy2007multivariate} and \citet*{okhrin2008estimation}. They showed that the shrinkage estimators of the portfolio weights lead to a decrease in the variance of the portfolio weights and to an increase in utility.

\citet*{bodnar2014estimation} proposed a new shrinkage estimator for the weights of the GMVP that turns out to provide better results in the high-dimensional case than the existing estimators do. This estimator is based on a convex combination of the sample estimator of the GMVP weights and an arbitrary constant vector expressed as

\begin{equation}
\mathbf{\hat{w}}_{n;GSE}=\alpha_n \frac{\mathbf{\hat{\Sigma}}_{n}^{-1}\mathbf{1}}{\mathbf{1}'\mathbf{\hat{\Sigma}}_{n}^{-1}\mathbf{1}}+(1-\alpha_n)\mathbf{b}_n \qquad \mbox{with} \qquad \mathbf{b}'_n\mathbf{1}=1.
\end{equation}
Here, the index GSE stands for `general shrinkage estimator'. It is assumed that $\mathbf{b}_n \in \mathbb{R}^p$ is a vector of constants such that $\mathbf{b}_n'\mathbf{\Sigma}\mathbf{b}_n$ is uniformly bounded. \citet*{bodnar2014estimation} proposed determining the optimal shrinkage intensity $\alpha_n$ for a given target portfolio $\mathbf{b}_n$ such that the out-of-sample risk is minimal, that is,
\begin{equation}
L=(\mathbf{\hat{w}}_{n;GSE}-\mathbf{w}_{GMVP})'\mathbf{\Sigma}(\mathbf{\hat{w}}_{n;GSE}-\mathbf{w}_{GMVP})
\end{equation}
is minimized with respect to $\alpha_n$. This result leads to

\begin{equation}
\displaystyle
\hat{\alpha}_n=\frac{\left(\mathbf{b}_n - \mathbf{\hat{w}}_n \right)'\mathbf{\Sigma}\,\mathbf{b}_n }{\left(\mathbf{b}_n - \mathbf{\hat{w}}_n \right)'\mathbf{\Sigma}\left(\mathbf{b}_n -  \mathbf{\hat{w}}_n \right)}.
\end{equation}

The authors showed that the optimal shrinkage intensity $\hat{\alpha}_n$ is almost surely asymptotically equivalent to a non-random quantity $\tilde{\alpha}_n \in \left[ 0,1 \right] $ when $\frac{p}{n}\to c \in (0,1)$ as $n\to \infty $, which is given by
\begin{equation}\label{talp}
\tilde{\alpha}_n =\frac{(1-c) R_{\mathbf{b}_n}}{c+(1-c) R_{\mathbf{b}_n}} ,
\end{equation}
where
\begin{equation}\label{Rbn}
 \displaystyle R_{\mathbf{b}_n}=\frac{\sigma_{\mathbf{b}_n}^2-\sigma^2_{n}}{\sigma^2_{n}} = \mathbf{1}' \mathbf{\Sigma}^{-1} \mathbf{1} \mathbf{b}_n' \mathbf{\Sigma} \mathbf{b}_n - 1
 \end{equation}
is the relative loss of the target portfolio $\mathbf{b}_n$, $\displaystyle \sigma^2_{\mathbf{b}_n}=\mathbf{b}'_n\mathbf{\Sigma}\mathbf{b}_n$ is the variance of the target portfolio, and $\displaystyle \sigma^2_{n}=1/\mathbf{1}'\mathbf{\Sigma}^{-1}\mathbf{1}$ is the variance of the GMVP.
This result provides an estimator of the optimal shrinkage intensity given by
\begin{equation} \label{estim.intens}
\hat{\tilde{\alpha}}_n=\frac{(1-\frac{p}{n}) \hat{R}_{\mathbf{b}_n}}{\frac{p}{n}+(1-\frac{p}{n}) \hat{R}_{\mathbf{b}_n}}, \mbox{   } \hat{R}_{\mathbf{b}_n}=(1-\frac{p}{n})\mathbf{b}_n'\mathbf{\hat{\Sigma}}_{n}\mathbf{b}_n \mathbf{1}'\mathbf{\hat{\Sigma}}_{n}^{-1}\mathbf{1}-1.
\end{equation}
Using the estimated shrinkage intensity $\hat{\tilde{\alpha}}_n$, the corresponding portfolio weights are given by
\begin{equation}
\mathbf{\hat{w}}_{n;ESI}=\hat{\tilde{\alpha}}_n \mathbf{\hat{w}}_{n} +(1-\hat{\tilde{\alpha}}_n)\mathbf{b}_n .
\end{equation}

\citet*{bodnar2014estimation}
proved that the ratio $\dfrac{\hat{\tilde{\alpha}}_n}{\tilde{\alpha}_n}\to1$ if $\frac{p}{n} \to c \in (0,1)$ as $n\to \infty$. In Theorem \ref{th2}, we show that the estimated intensity is asymptotically normally distributed. The proof of Theorem \ref{th2} is given in the Appendix. 

\vspace{0.3cm}

\begin{theorem}\label{th2}
	Let $p\equiv p(n)$ and $c_n = \frac{p}{n} \to c \in (0,1)$. Assume that $\{\mathbf{X}_t\}$
	is a sequence of independent and normally distributed $p$-dimensional random vectors with mean $\boldsymbol{\mu}$ and covariance matrix $\mathbf{\Sigma}$,
	which is assumed to be positive definite. Then

\begin{equation}\label{talp_asym}
\sqrt{n}\frac{ \hat{\tilde{\alpha}}_n -A_n}{B_n} \stackrel{d}{\to} \mathcal{N}(0,1) ~~\text{for $p/n\to c\in (0,1)$ as $n\to \infty$,}
\end{equation}
 where {\small
	 \begin{eqnarray*}
	 A_n & = & \frac{(1-c_n) R_{\mathbf{b}_n}}{c_n+(1-c_n) R_{\mathbf{b}_n}},  \\
         B_n^2 & = & 2 \frac{c_n^2 (1-c_n)(2-c_n) (R_{\mathbf{b}_n}+1)}{((c_n + R_{\mathbf{b}_n}(1 - c_n))^4}\left(R_{\mathbf{b}_n}+\frac{c_n}{2-c_n}\right).
	 \end{eqnarray*}
}
\end{theorem}

Next, we introduce a test based on the estimated shrinkage intensity. The motivation is based on the following equivalences (see, \eqref{talp} and \eqref{Rbn}):
\begin{equation*}
\tilde{\alpha}_n = 0\, \iff \,R_{\mathbf{b}_n}=0\, \iff \,\sigma_{\mathbf{b}_n}^2=\sigma^2_{n}  \, .
\end{equation*}
This result means that $\tilde{\alpha}_n = 0$ if and only if the variance of the portfolio based on $\mathbf{b}_n$ is equal to the variance of the GMVP. This finding in turn means that
 $\displaystyle \mathbf{b}'_n\mathbf{\Sigma}\mathbf{b}_n=1/\mathbf{1}'\mathbf{\Sigma}^{-1}\mathbf{1}= \min\limits_{\substack{\bf{w}:\,\bf{w'1}}=1}\mathbf{w}'\mathbf{\Sigma} \mathbf{w} =  \mathbf{w}_{GMVP}'\mathbf{\Sigma} \,\mathbf{w}_{GMVP}$.
 Since the GMVP weights are uniquely determined, this result is valid if and only if $\mathbf{b}_n= \mathbf{w}_{GMVP}$. Choosing $\mathbf{b}_n= \mathbf{r}$, it holds that

\begin{equation*}
\mathbf{w}_{GMVP} = \mathbf{r}\, \iff \, \tilde{\alpha}_n = 0 .
\end{equation*}
Thus, it is possible to obtain a test for the structure of the GMVP using the shrinkage intensity with the hypothesis given by
\begin{equation}\label{hypotheses_shr}
H_{0}:\tilde{\alpha}_n=0\qquad \textrm{against} \qquad H_{1}:\tilde{\alpha}_n>0.
\end{equation}
Note that $\hat{\tilde{\alpha}} = \hat{\tilde{\alpha}}(\mathbf{b}_n)$. Let $S_n = \sqrt{n} \; \hat{\tilde{\alpha}}(\mathbf{b}_n = \mathbf{r})$. For testing (\ref{hypotheses_shr}), we use the test statistic $S_n$.

From Theorem \ref{th2} we get
\begin{equation*}
\frac{S_n-\sqrt{n}A_n}{B_n} \stackrel{d}{\to} \mathcal{N}(0,1)~~ \text{for $p/n\to c\in (0,1)$ as $n\to \infty$},
\end{equation*}
where $A_n$ and $B_n$ are given in the statement of \textit{Theorem \ref{th2}}. Moreover, under the null hypothesis, $R_{\mathbf{b}_n}=0$ and, thus, $S_n \stackrel{d}{\to} \mathcal{N}\left(0, 2(1-c)/c\right)$ for $p/n\to c\in (0,1)$ as $n\to \infty$. This result gives us a promising new approach for detecting deviations of the true portfolio weights from the given quantities. Using Theorem \ref{th2}, we are able to make a statement about the power function of this test. Since $A_n$ and $B_n$ depend on $\mathbf{b}_n$, we only have to replace this quantity with $\mathbf{r}$. It holds that
\begin{eqnarray}
&&P\left(  \frac{S_n}{\sqrt{2 \frac{1-c}{c}}} > z_{1-\alpha}\right) \nonumber\\
&=&1 - P\Bigg(  \frac{S_n - A_n(\mathbf{b}_n = \mathbf{r})}{B_n(\mathbf{b}_n = \mathbf{r})} \nonumber\\
&\le&  \frac{\sqrt{2 \frac{1-c}{c}} z_{1-\alpha} - A_n(\mathbf{b}_n = \mathbf{r})}{B_n(\mathbf{b}_n = \mathbf{r})} \Bigg) \nonumber\\
& \approx & 1 - \Phi\left(\frac{\sqrt{2 \frac{1-c}{c}} z_{1-\alpha} - A_n(\mathbf{b}_n = \mathbf{r})}{B_n(\mathbf{b}_n = \mathbf{r})} \right) \label{power} .
\end{eqnarray}
Note that the approximation given in (\ref{power}) is purely a function of
$R_{\mathbf{b}_n = \mathbf{r}}$. This property is a main difference from the test discussed in Section III.A, where the power function is a function of $\lambda_n$. These properties are very useful to analyse the performances of both tests and simplify the power analysis.

 \begin{figure}[!h]
	\centering
	\vspace{-0.5cm}
	 \begin{tabular}{c}
	 	\includegraphics[scale=0.48]{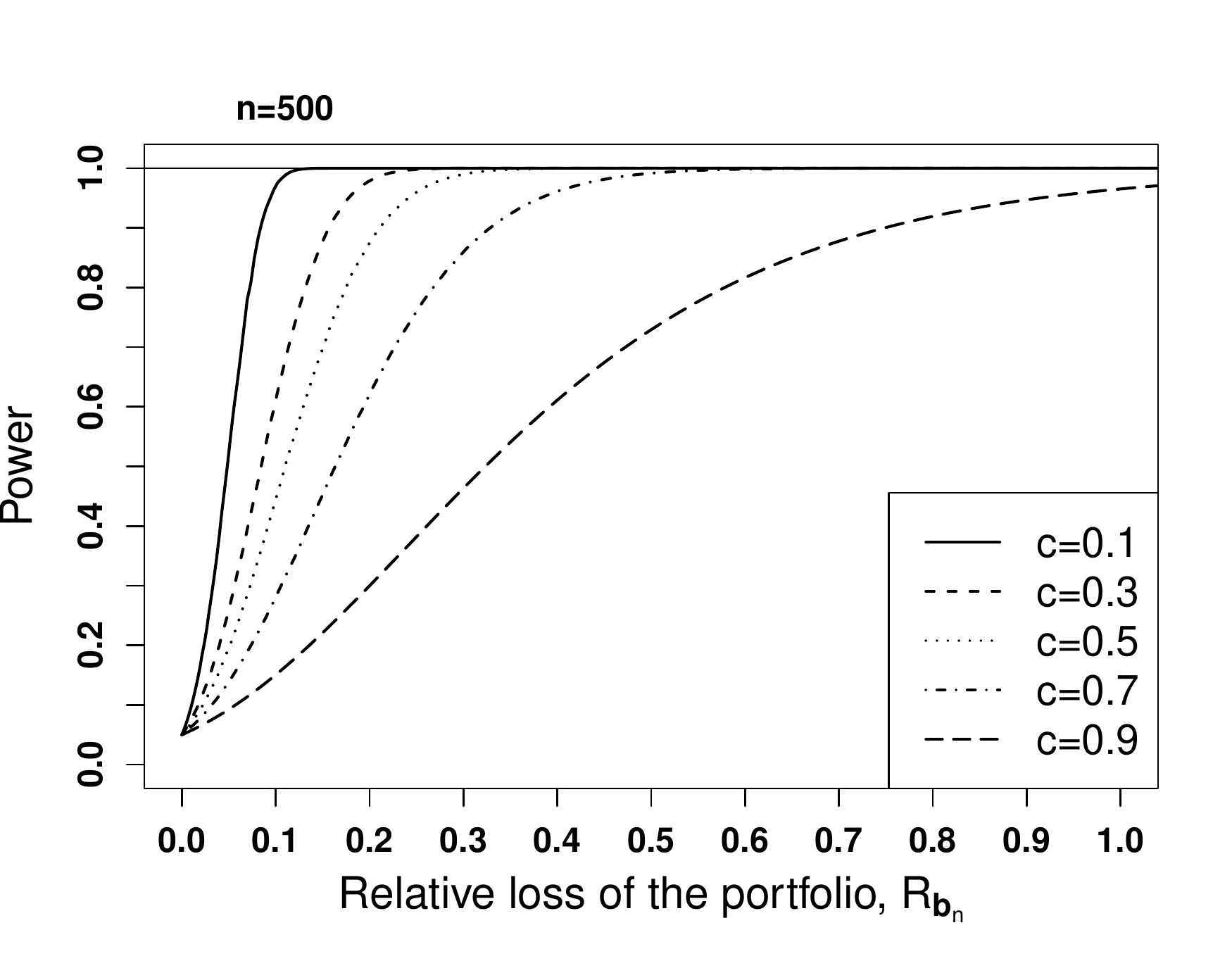}\\   	
	 	\includegraphics[scale=0.48]{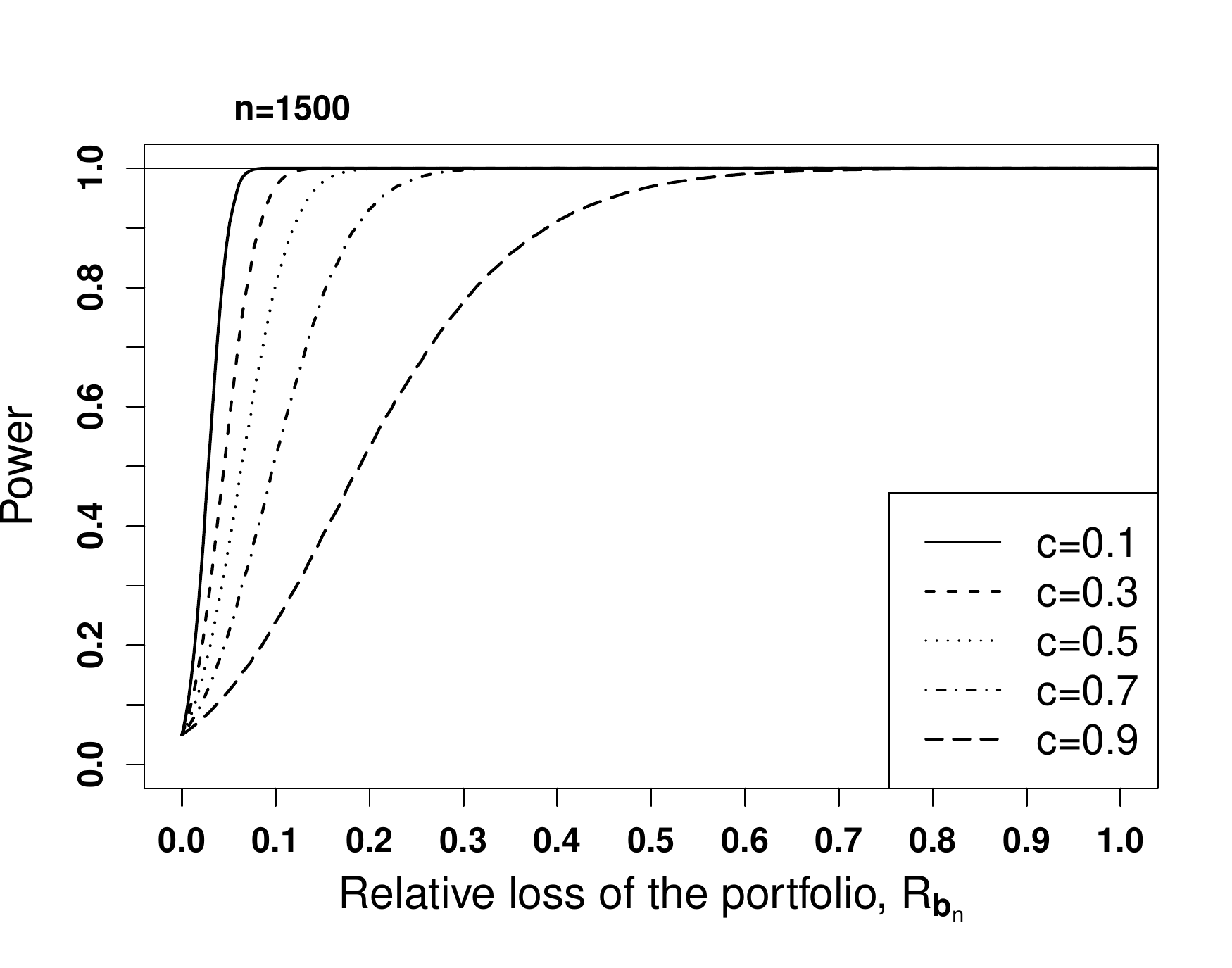}
	 \end{tabular}
	\caption{Asymptotic power function for the test problem in (\ref{hypotheses_shr}) as a function of $R_{\mathbf{b}_n}$ for various values of $c \in \{0.1, 0.3, 0.5,0.7, 0.9\}$. The number of observations is $n=\{500, 1500\}$. The nominal significance level of the test (the probability of a type I error) is $\alpha=5\%$.} \label{fig:fig2}
\end{figure}

In \textit{Figure \ref{fig:fig2}}, the power of the test is shown as a function of $R_{\mathbf{b}_n}$ and $n$. It can be seen that the test performs better for smaller values of $c$. With increasing values of $c$, the power of the test decreases.
	We determine the power function for two different sample numbers, $n=500$ and $n=1500$. As expected, the test shows a better performance for larger values of $n$, since $A_n(\mathbf{b}_n = \mathbf{r})$ increases, the numerator of the expression in the cumulative distribution function in (\ref{power}) becomes negative, and the whole expression tends to one.

\subsection{Case of a Singular Covariance Matrix $\mathbf{\Sigma}$}

We extend the results of Section III.A and Section III.B to the case of a singular covariance matrix with $\text{rank}(\mathbf{\Sigma})=q<p$. Here, we consider two types of singularity: (i) in the population covariance matrix $\mathbf{\Sigma}$ and (ii) in addition, in the sample covariance matrix $\hat{\mathbf{\Sigma}}_n$ by allowing the sample size $n$ to be smaller than the dimension $p$. Throughout this section, we refer to $q$ as the actual dimension of the data generating process and, consequently, derive the results under the high-dimensional asymptotic regime $q/n \to \tilde{c} \in (0,1)$ as $n \to \infty$.

In the case $q<p$, the sample covariance matrix $\hat{\mathbf{\Sigma}}_n$ is singular and its inverse does not exist. As a result, the Moore-Penrose inverse of $\hat{\mathbf{\Sigma}}_n$, which we denote by $\hat{\mathbf{\Sigma}}_n^{+}$, is used to estimate the weights of the GMVP expressed as
\begin{equation}\label{weights_MP}
		\mathbf{\hat{\tilde{w}}}_{n}=\frac{\mathbf{\hat{\Sigma}}_{n}^{+}\mathbf{1}}{\mathbf{1}'\mathbf{\hat{\Sigma}}_{n}^{+}\mathbf{1}}.
\end{equation}
In a similar way, the true GMVP weights are obtained and they are given by
\begin{equation*}
		\mathbf{\tilde{w}}_{GMVP}=\frac{\mathbf{\Sigma}^{+}\mathbf{1}}{\mathbf{1}'\mathbf{\Sigma}^{+}\mathbf{1}}.
\end{equation*}
The Moore-Penrose inverse of the covariance matrix has already been used in portfolio theory by \citet*{PappasKiriakopoulosKaimakamis2010,bodnar2016singular} among others, while \citet*{bodnar2016spectral} derived several distributional properties of the Moore-Penrose inverse of the sample covariance matrix.

Next, we consider linear combinations of both the true GMVP weights and their estimator given by
\begin{equation*}
		\mathbf{\tilde{w}}_{GMVP}^*=\frac{\mathbf{L}\mathbf{\Sigma}^{+}\mathbf{1}}{\mathbf{1}'\mathbf{\Sigma}^{+}\mathbf{1}}
\quad \text{and} \quad
		\mathbf{\hat{\tilde{w}}}_{n}^*=\frac{\mathbf{L}\mathbf{\hat{\Sigma}}_{n}^{+}\mathbf{1}}{\mathbf{1}'\mathbf{\hat{\Sigma}}_{n}^{+}\mathbf{1}},
\end{equation*}
where $\mathbf{L}$ is a $k \times p$ matrix of constants with $k\le q$ and $rank(L)=k$. In particular, if $\mathbf{L}=[\mathbf{I}_k \mathbf{O}_{k,p-k}]$ with the $k$-dimensional identity $\mathbf{I}_k$ and the $k\times (p-k)$ zero matrix $\mathbf{O}_{k,p-k}$, then $\mathbf{\tilde{w}}_{GMVP}^*$ is the vector of the first components of $\mathbf{\tilde{w}}_{GMVP}$.

In order to verify the structure of the GMVP, we first extend the test based on the Mahalanobis distance to the test problem given by
 \begin{equation}\label{hypotheses_sin}
     H_{0}:\mathbf{\tilde{w}}_{GMVP}^*=\mathbf{\tilde{r}}^*\qquad \textrm{against} \qquad H_{1}:\mathbf{\tilde{w}}_{GMVP}^*\not= \mathbf{\tilde{r}}^*
	 \end{equation}
for some $k$-dimensional vector $\mathbf{\tilde{r}}^*$ and the test statistic
\begin{equation}\label{tT_n}
\tilde{T}_{n}=\frac{n-q}{k}(\mathbf{1}'\hat{\mathbf{\Sigma}}^{+}_n\mathbf{1})(\mathbf{\hat{\tilde{w}}}^{*}_{n}-\mathbf{\tilde{r}}^{*})'
(\mathbf{\hat{\tilde{Q}}}_n^*)^{-1}(\mathbf{\hat{\tilde{w}}}^{*}_{n}-\mathbf{\tilde{r}}^{*}),
\end{equation}
where $\mathbf{\hat{\tilde{Q}}}_n^*=\mathbf{L}\hat{\mathbf{\Sigma}}_n^{+}\mathbf{L}'-\displaystyle\frac{\mathbf{L}\hat{\mathbf{\Sigma}}_n^{+}\mathbf{1}
 \mathbf{1}'\hat{\mathbf{\Sigma}}_n^{+}\mathbf{L}'}{\mathbf{1}'\hat{\mathbf{\Sigma}}_n^{+}\mathbf{1}}$.
This test statistic was considered in \citet*{bodnar2017test}, who derived its finite-sample distribution for both small portfolio dimension and sample size.

In Theorem \ref{th3}, we extend these results by deriving the asymptotic distribution of $\tilde{T}_{n}$ under the high-dimensional asymptotic regime with $q/n \to \tilde{c} \in (0,1)$ and $k/n \to \tilde{b} \in [0,1)$ as $n \to \infty$. To this end, we also note that only a part of the GMVP weights are tested in \eqref{hypotheses_sin}. In order to test the structure of the whole portfolio, we have to repeat the test \eqref{hypotheses_sin} for several subvectors of $\mathbf{\tilde{w}}_{GMVP}$ and to adjust the significance level of each test by applying, for example, the Bonferroni correction.

\begin{theorem}\label{th3}
Assume that $\{\mathbf{X}_t\}$ is a sequence of independent and singular normally distributed $p$-dimensional random vectors with mean $\boldsymbol{\mu}$ and singular covariance matrix $\mathbf{\Sigma}$ with $rank(\mathbf{\Sigma})=q$. Let $q\equiv q(n)$ and $\tilde{c}_n = \frac{q}{n} \to \tilde{c} \in (0,1)$ and let $k<q$ such that $\tilde{b}_n = \frac{k}{n} \to \tilde{b} \in (0,1)$. We define
 	\begin{eqnarray*}
\tilde{C}^2_n&=&2+2\frac{(1-\tilde{c}+\tilde{b})\tilde{\lambda}_n^2}{\tilde{b}}+4\frac{(1-\tilde{c}+\tilde{b})\tilde{\lambda}_n}{\tilde{b}}\\
&+&2\frac{\tilde{b}}{1-\tilde{c}}\left(1+\frac{(1-\tilde{c}+\tilde{b})\tilde{\lambda}_n}{\tilde{b}}\right)^2
		\end{eqnarray*}
with
\begin{equation*}
\tilde{\lambda}_n=\left(\mathbf{1}^\prime \mathbf{\Sigma}^{+}\mathbf{1}\right)
\left(\tilde{\mathbf{w}}^*_{GMVP}-\tilde{\mathbf{r}}^*\right)^\prime (\tilde{\mathbf{Q}}^*)^{-1} \left(\tilde{\mathbf{w}}^*_{GMVP}-\tilde{\mathbf{r}}^*\right).
\end{equation*}

Then, it holds that
\begin{equation*}
\sqrt{k}\left(\frac{\tilde{T}_n-1-\tilde{\lambda}_n\frac{n-q+k}{k}}{\tilde{C}_n}\right) \stackrel{d}{\to} \mathcal{N}\left(0, 1\right)
\end{equation*}
for $q/n\to \tilde{c}\in (0,1)$ and $k/n\to \tilde{b}\in (0,1)$ as $n\to \infty$. Under the null hypothesis, $\sqrt{q-1} \; ( \tilde{T}_n - 1 ) \stackrel{d}{\to} \mathcal{N}\left(0, {2(1-\tilde{c}+\tilde{b})}/{(1-\tilde{c})}\right)$ for $q/n\to \tilde{c}\in (0,1)$ and $k/n\to \tilde{b}\in (0,1)$ as $n\to \infty$.
\end{theorem}

The results of Theorem \ref{th3} are very useful to derive the power function of the suggested test. Similarly to the case of a non-singular covariance matrix, it is given by
\begin{eqnarray*}
&&P\left(  \frac{\sqrt{q-1}\left(\tilde{T}_n-1\right)}{\sqrt{2 /(1-\tilde{c})}} > z_{1-\alpha}\right) \nonumber \\
 & = &
1 - P\Bigg(  \frac{\sqrt{q-1}\left(\tilde{T}_n-1-\tilde{\lambda}_n\frac{n-1}{q-1}\right)}{\tilde{C}_n} \nonumber\\
&\le&  \frac{\sqrt{\frac{2}{(1-\tilde{c})}} z_{1-\alpha} - \frac{\sqrt{q-1}\tilde{\lambda}_n(n-1)}{q-1}}{\tilde{C}_n} \Bigg) \nonumber \\
& \approx & 1 - \Phi\left(\frac{\sqrt{2/(1-\tilde{c})} z_{1-\alpha} - \sqrt{q-1}\frac{\tilde{\lambda}_n}{\tilde{\tilde{c}}}}{\tilde{C}_n} \right).
\end{eqnarray*}

Next, we present the test based on a shrinkage estimator for the singular covariance matrix $\mathbf{\Sigma}$. Similarly as in the case of a nonsingular covariance matrix we get the shrinkage intensity given by
\begin{equation}
\displaystyle
\hat{\alpha}^+_n=\frac{\left(\mathbf{b}_n - \mathbf{\hat{\tilde{w}}}_{n} \right)'\mathbf{\Sigma}\,\mathbf{b}_n }{\left(\mathbf{b}_n - \mathbf{\hat{\tilde{w}}}_{n}\right)'\mathbf{\Sigma}\left(\mathbf{b}_n -  \mathbf{\hat{\tilde{w}}}_{n} \right)}\,,
\end{equation}
where $\mathbf{\hat{\tilde{w}}}_{n}$ are given in \eqref{weights_MP}. The following proposition holds.
\begin{proposition}\label{alpha+} Assume that $\{\mathbf{X}_t\}$
	is a sequence of independent and singular normally distributed $p$-dimensional random vectors with mean $\boldsymbol{\mu}$ and singular covariance matrix $\mathbf{\Sigma}$ with a rank $q$ and assume $0< M_l \le 1/\bi'\bSigma^{+}\bi \le \mathbf{b}_n'\bSigma\mathbf{b}_n \le M_u <\infty $ for all $n$.
The optimal shrinkage intensity $\hat{\alpha}^+_n$ is almost surely asymptotically equivalent to a non-random quantity $\tilde{\alpha}^+_n \in \left[ 0,1 \right] $ when $q/n\to \tilde{c} \in (0,1)$ as $n\to \infty $, which is given by
\begin{equation}\label{talp+}
\tilde{\alpha}^+_n =\frac{(1-\tilde{c}) R^+_{\mathbf{b}_n}}{\tilde{c}+(1-\tilde{c}) R^+_{\mathbf{b}_n}} ,
\end{equation}
where
\begin{equation}\label{Rbn+}
 \displaystyle R^+_{\mathbf{b}_n}= \mathbf{1}' \mathbf{\Sigma}^{+} \mathbf{1} \mathbf{b}_n' \mathbf{\Sigma} \mathbf{b}_n - 1.
 \end{equation}
\end{proposition}
Proposition \ref{alpha+} is complementary to \citet[Theorem 2.1]{bodnar2014estimation} and covers additionally the case of a nonsingular matrix $\mathbf{\Sigma}$. Going carefully through the proof of Proposition \ref{alpha+} we can easily deduce the consistent estimator of $\tilde{\alpha}_n^+$ given by
\begin{eqnarray}
  \label{bona_alpha+}
  \widehat{\tilde{\alpha}}^+_n =\frac{(1-q/n) \hat{R}^+_{\mathbf{b}_n}}{q/n+(1-q/n) \hat{R}^+_{\mathbf{b}_n}}
\end{eqnarray}
with
\begin{equation}\label{hatRbn+}
 \hat{R}^+_{\mathbf{b}_n}= (1-q/n)\mathbf{1}' \hat{\mathbf{\Sigma}}_n^{+} \mathbf{1} \mathbf{b}_n' \hat{\mathbf{\Sigma}}_n \mathbf{b}_n - 1.
 \end{equation}
Now we are ready to state the central limit theorem for $\widehat{\tilde{\alpha}}^+_n$, which is a straightforward consequence of the proofs of Proposition \ref{alpha+} and Theorem \ref{th2}.
\begin{thm}
\label{th4}
	Let $q\equiv q(n)$ and $\tilde{c}_n = \frac{q}{n} \to \tilde{c} \in (0,1)$. Assume that $\{\mathbf{X}_t\}$
	is a sequence of independent and singular normally distributed $p$-dimensional random vectors with mean $\boldsymbol{\mu}$ and singular covariance matrix $\mathbf{\Sigma}$ with a rank $q$. Then
\begin{equation}\label{talp_asym+}
\sqrt{n}\frac{ \hat{\tilde{\alpha}}^+_n -A^+_n}{B^+_n} \stackrel{d}{\to} \mathcal{N}(0,1) ~~\text{as $n\to \infty$,}
\end{equation}
 where{\small
	 \begin{eqnarray*}
	 A^+_n & = & \frac{(1-\tilde{c}_n) R^+_{\mathbf{b}_n}}{\tilde{c}_n+(1-\tilde{c}_n) R^+_{\mathbf{b}_n}},  \\
           B_n^{2\;+} & = & 2 \frac{\tilde{c}_n^2 (1-\tilde{c}_n)(2-\tilde{c}_n) (1+R^+_{\mathbf{b}_n})}{((\tilde{c}_n + R^+_{\mathbf{b}_n}(1 - \tilde{c}_n))^4}\left(R^+_{\mathbf{b}_n}+\frac{\tilde{c}_n}{2-\tilde{c}_n}\right).
	 \end{eqnarray*}}
\end{thm}
Next, we are ready to introduce a test based on the estimated shrinkage intensity for testing the hypotheses
\begin{equation}\label{hypotheses_shr+}
H_{0}:\tilde{\alpha}^+_n=0\qquad \textrm{against} \qquad H_{1}:\tilde{\alpha}^+_n>0
\end{equation}
which are equivalent to
\begin{eqnarray*}
\label{hypotheses_sin+}
     H_{0}:\mathbf{\tilde{w}}_{GMVP}=\mathbf{\tilde{r}}\qquad \textrm{against} \qquad H_{1}:\mathbf{\tilde{w}}_{GMVP}\not= \mathbf{\tilde{r}}.
\end{eqnarray*}
Similarly as in the case of a nonsingular covariance matrix, we use the test statistic $S^+_n= \sqrt{n} \; \hat{\tilde{\alpha}}^+(\mathbf{b}_n = \mathbf{r})$ for testing (\ref{hypotheses_shr+}).
From Theorem \ref{th4} we get
\begin{equation*}
\frac{S^+_n-\sqrt{n}A^+_n}{B^+_n} \stackrel{d}{\to} \mathcal{N}(0,1)~~ \text{for $q/n\to \tilde{c}\in (0,1)$ as $n\to \infty$},
\end{equation*}
where $A^+_n$ and $B^+_n$ are given in the statement of Theorem \ref{th4}. Under the null hypothesis, $S^+_n \stackrel{d}{\to} \mathcal{N}\left(0, 2(1-\tilde{c})/\tilde{c}\right)$ for $q/n\to \tilde{c}\in (0,1)$ as $n\to \infty$. The power function can be constructed in a similar manner as in the case of a nonsingular matrix $\mathbf{\Sigma}$.

This result extends our previous findings and suggests that we may still use the test based on the optimal shrinkage intensity in the case of a singular population covariance matrix with the only difference that instead of $p/n\to c\in (0, 1)$ we demand $q/n\to \tilde{c} \in (0, 1)$ as $n\to \infty$ and instead of the usual inverse we can safely use the Moore-Penrose inverse of the sample covariance matrix $\hat{\mathbf{\Sigma}}_n$. Moreover, the test based on the shrinkage intensity needs no multiple testing scheme, which indicates its huge advantage over the test based on the Mahalanobis distance.

\section{Comparison Study}

The aim of this section is to compare several tests for the weights of the GMVP.

In the preceding two subsections, we considered two tests for the weights of the GMVP. For the test based on the empirical portfolio weights, the exact distribution of the test statistic is known. In Section III.A, the asymptotic power function of the test proposed by \citet*{bodnar2008test} is derived in a high-dimensional setting. In Section III.B, a new test is proposed, and its asymptotic power function, which purely depends on $R_{\mathbf{b}_n=\mathbf{r}}$, is determined. The fact that both tests depend on different quantities
complicates the comparison of both tests. Note that
  \begin{eqnarray*}
  &&  R_{\mathbf{b}_n=\mathbf{r}} = \mathbf{1}' \mathbf{\Sigma}^{-1} \mathbf{1} \; \mathbf{r}' \mathbf{\Sigma} \mathbf{r} - 1 \\
    &=& \lambda_n \; \frac{\mathbf{r}' \mathbf{\Sigma} \mathbf{r}}{(\mathbf{w}^*_{GMVP} - \mathbf{r}^*)' ( \mathbf{Q}^* )^{-1} (\mathbf{w}^*_{GMVP} - \mathbf{r}^*)} \; - \; 1  .
  \end{eqnarray*}

Here, both tests are compared with each other via simulations. Additionally, we include the test presented by \citet[Theorem 10]{glombeck} in our comparison study as well as tests derived for a singular covariance matrix in Section III.C.

\subsection{Design of the Comparison Study}

Let $\mathbf{\Sigma}$ be a $p\times p$ positive semi-definite covariance matrix of asset returns, $n$ the number of samples, and $p\equiv p(n)$.
		The structure of the covariance matrix is chosen in the following way: one-ninth of the non-zero eigenvalues are set equal to $2$, four-ninths are set equal to $5$, and the rest are set equal to $10$. A similar structure of the spectrum of the populaion covariance matrix is present in \citet*{ledoit2012nonlinear}. In doing so, we can ensure that the eigenvalues are not very dispersed, and if $p$ increases, then the spectrum of the covariance matrix does not change its behaviour. Then, the covariance matrix is determined as follows
		\begin{equation*}
			\mathbf{\Sigma}=\Theta\mathbf{\Lambda}\Theta^{'},
		\end{equation*}
where $\mathbf{\Lambda}$ is the diagonal matrix whose diagonal elements are the predefined eigenvalues and $\Theta$ is the $p\times p$ matrix of eigenvectors obtained from the spectral decomposition of a standard Wishart-distributed random matrix.

We consider the following scenario for modelling the changes. Under the alternative hypothesis, the covariance matrix is defined by
\begin{equation}\label{scenario1}
\mathbf{\Sigma}_{1}=\Theta\Delta\mathbf{\Lambda}\Delta\Theta^{'},
\end{equation}
where
\begin{equation}
\Delta=\left( \begin{array}{c|c}
D_m & \textbf{0} \\
\hline
\textbf{0} & I_{p-m}
\end{array}\right) ,
\end{equation}
with $D_m=diag(a)$ and $a=1+0.1\kappa$,\, $\kappa \in \{1,2,\ldots, 15\}$, $m \in \{0.2p,0.8p\}$ when $\bSigma$ is non-singular and $m \in \{0.2q,0.8q\}$ when $\bSigma$ is singular. The matrix $\Delta$ determines the deviations from the null hypothesis due to changes in the eigenvalues of the covariance matrix $\mathbf{\Sigma}$. Other specifications of the covariance matrix $\mathbf{\Sigma}$ under the alternative hypothesis might be considered as well.

\subsection{Comparison of the Tests}

\begin{figure}[th!]
		\centering
		\begin{tabular}{cc}
			\includegraphics[width=0.45\linewidth]{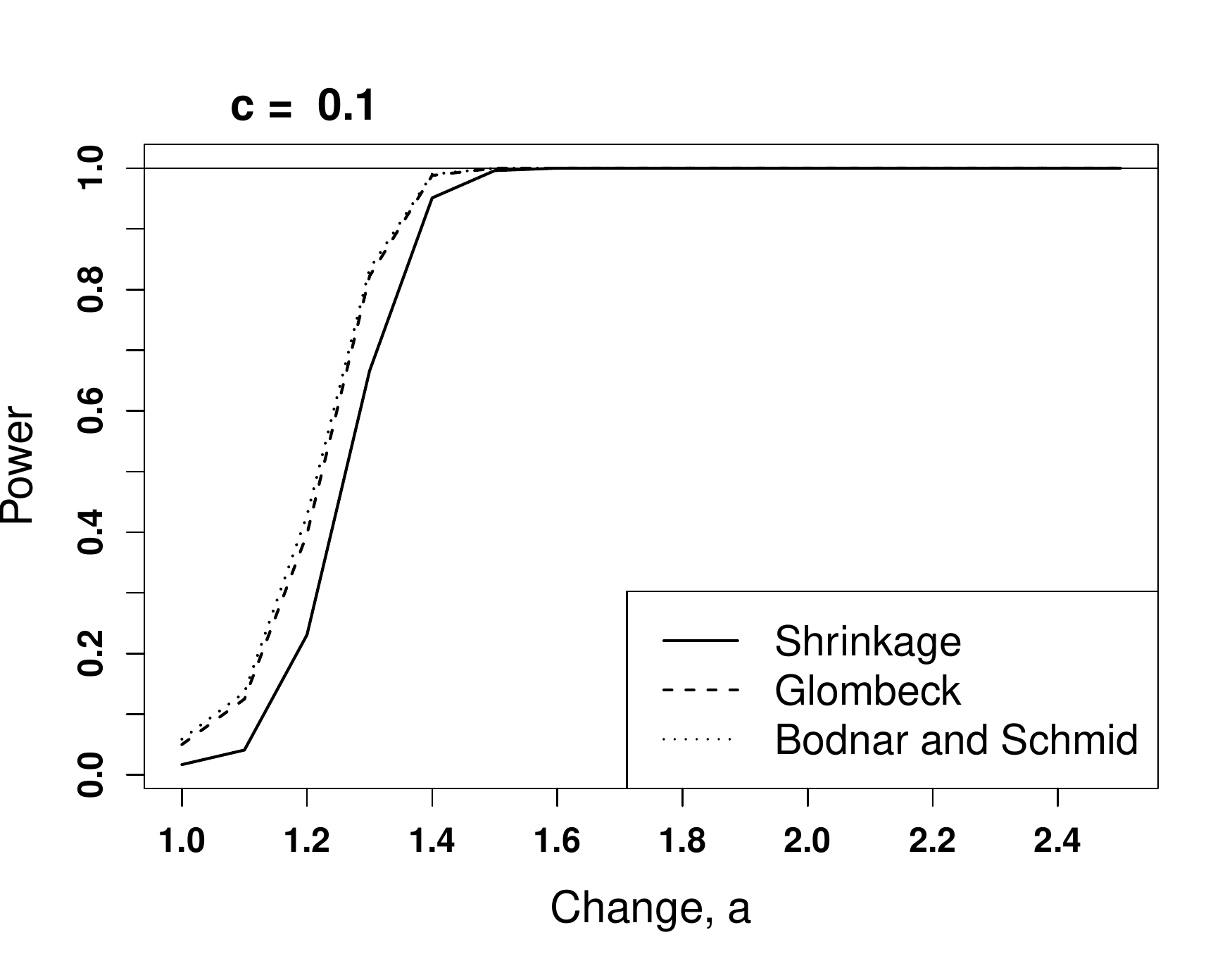} &	\includegraphics[width=0.45\linewidth]{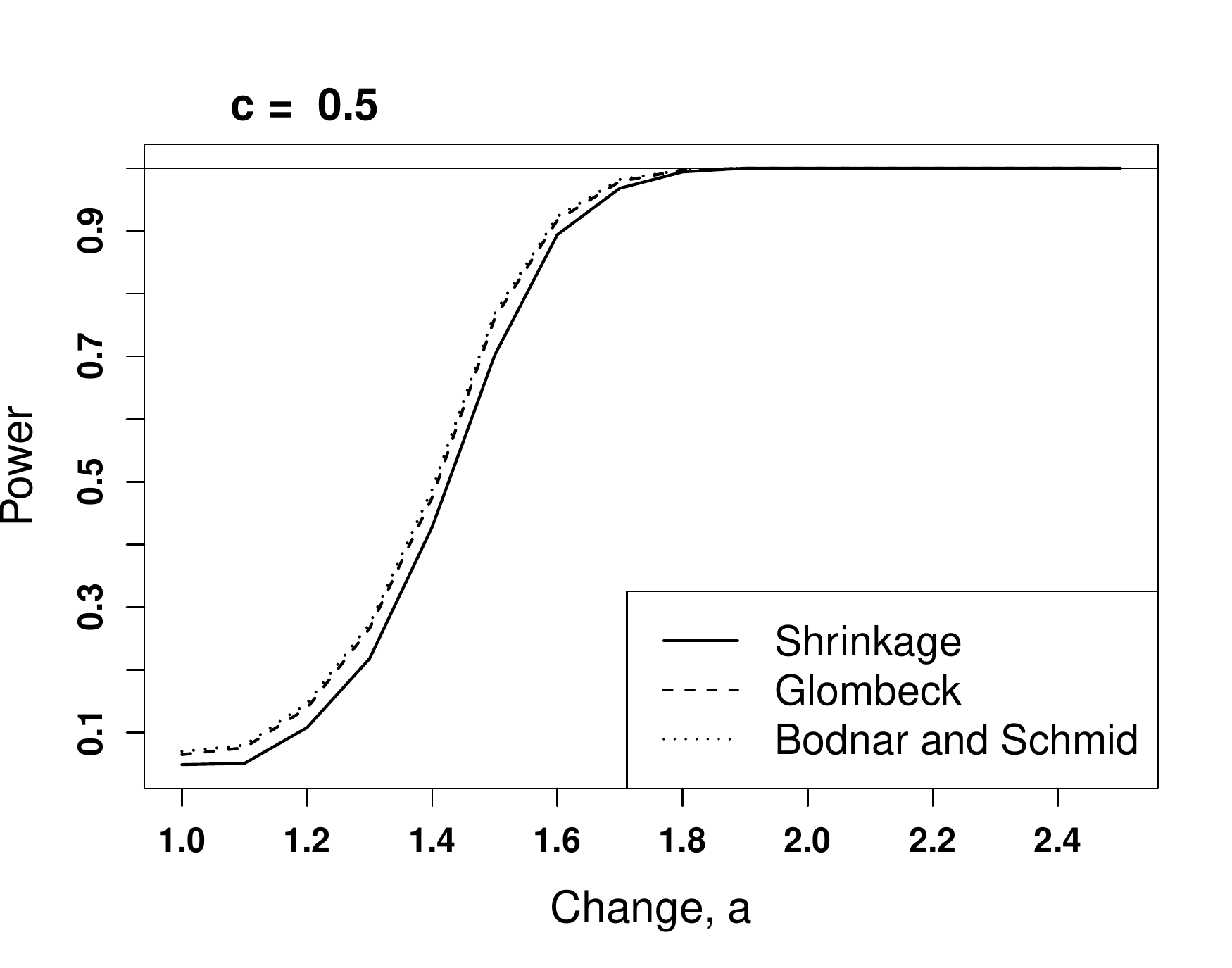}\\
		
			\includegraphics[width=0.45\linewidth]{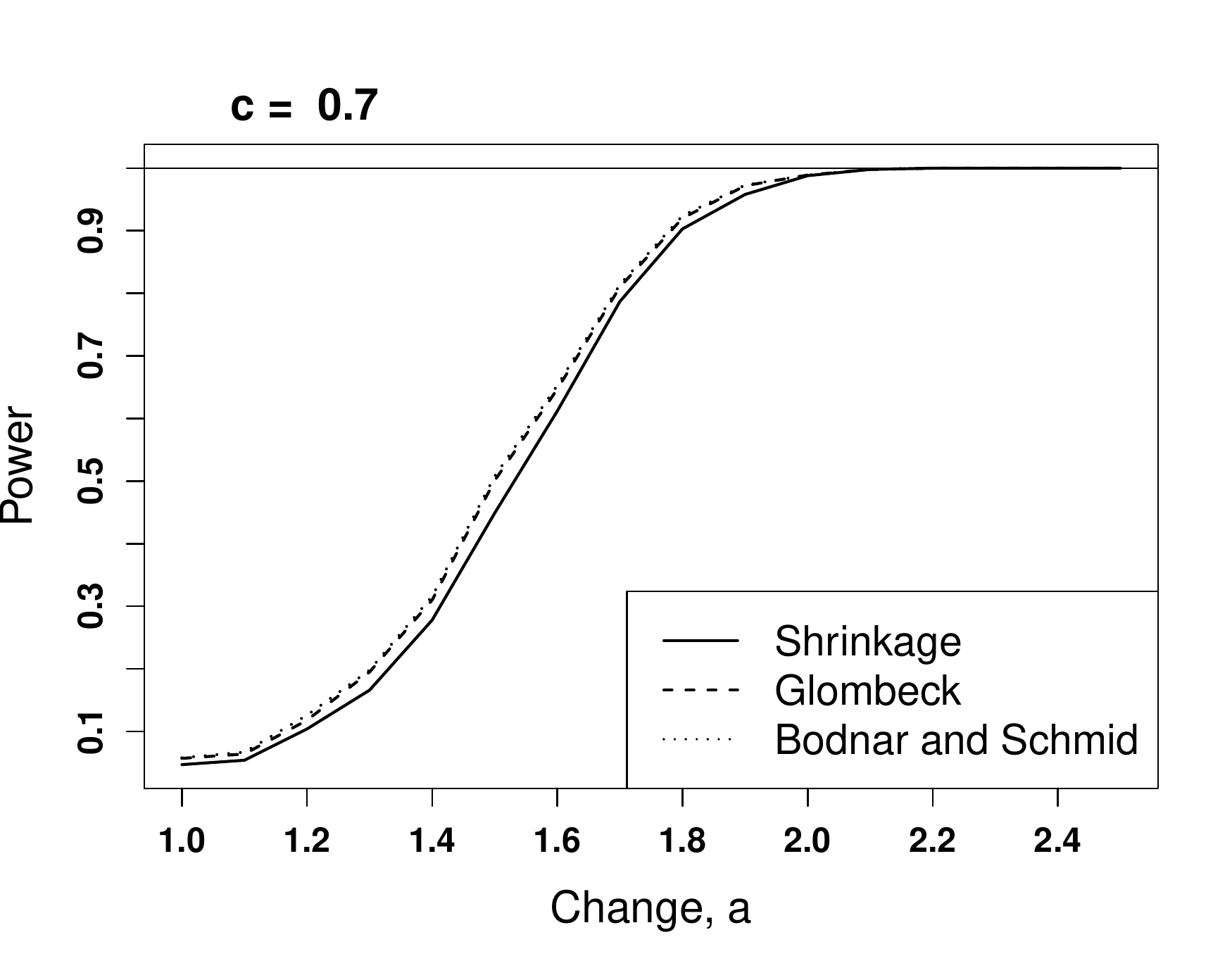} &	\includegraphics[width=0.45\linewidth]{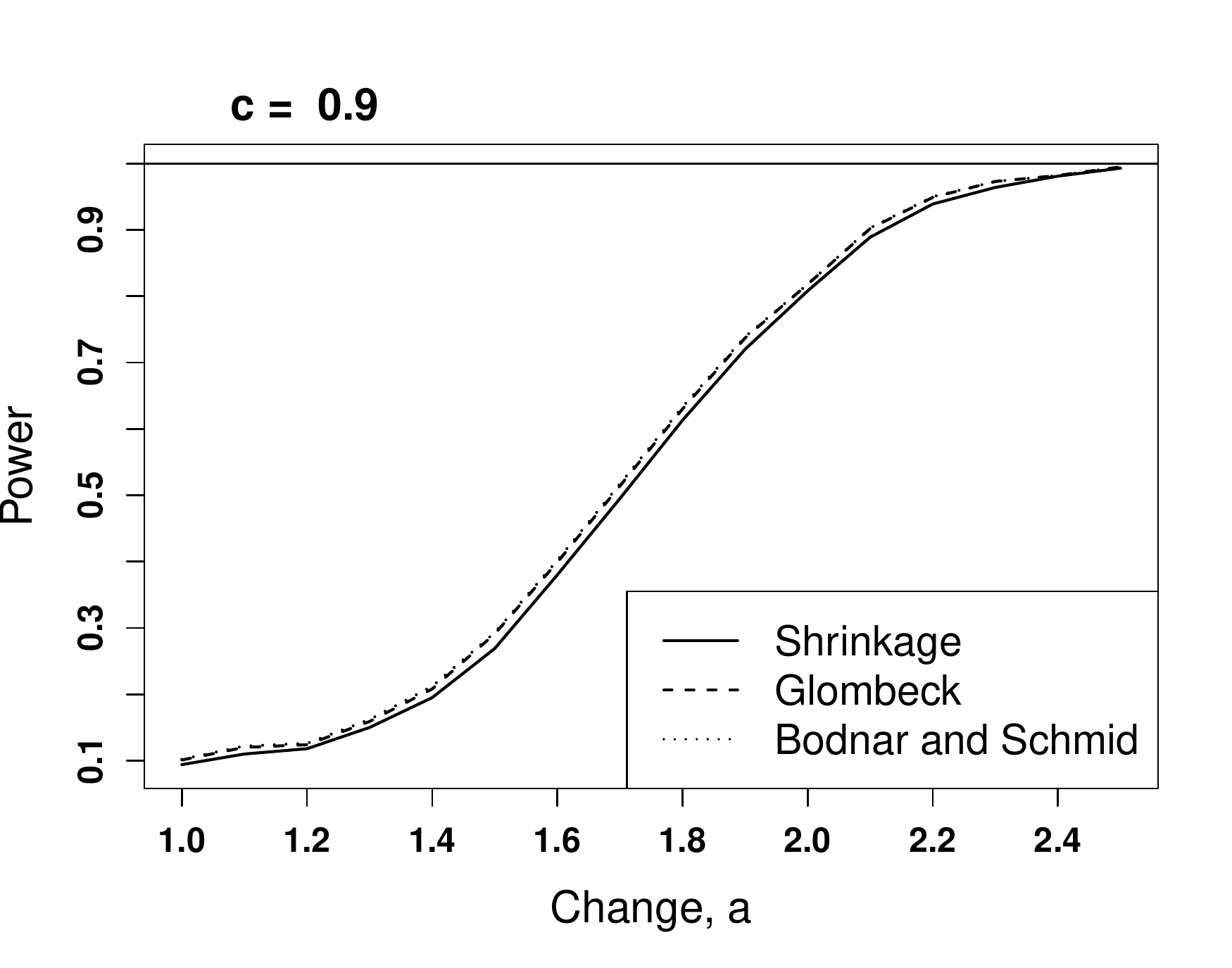}
			
		\end{tabular}
	\caption{Empirical power functions of the three tests for different values of $c$, $20\%$ changes on the main diagonal according to scenario given in (\ref{scenario1}) and $n=500$.}
\label{fig:fig3}	
\end{figure}

\begin{figure}[th!]
		\centering
		\begin{tabular}{cc}
			\includegraphics[width=0.45\linewidth]{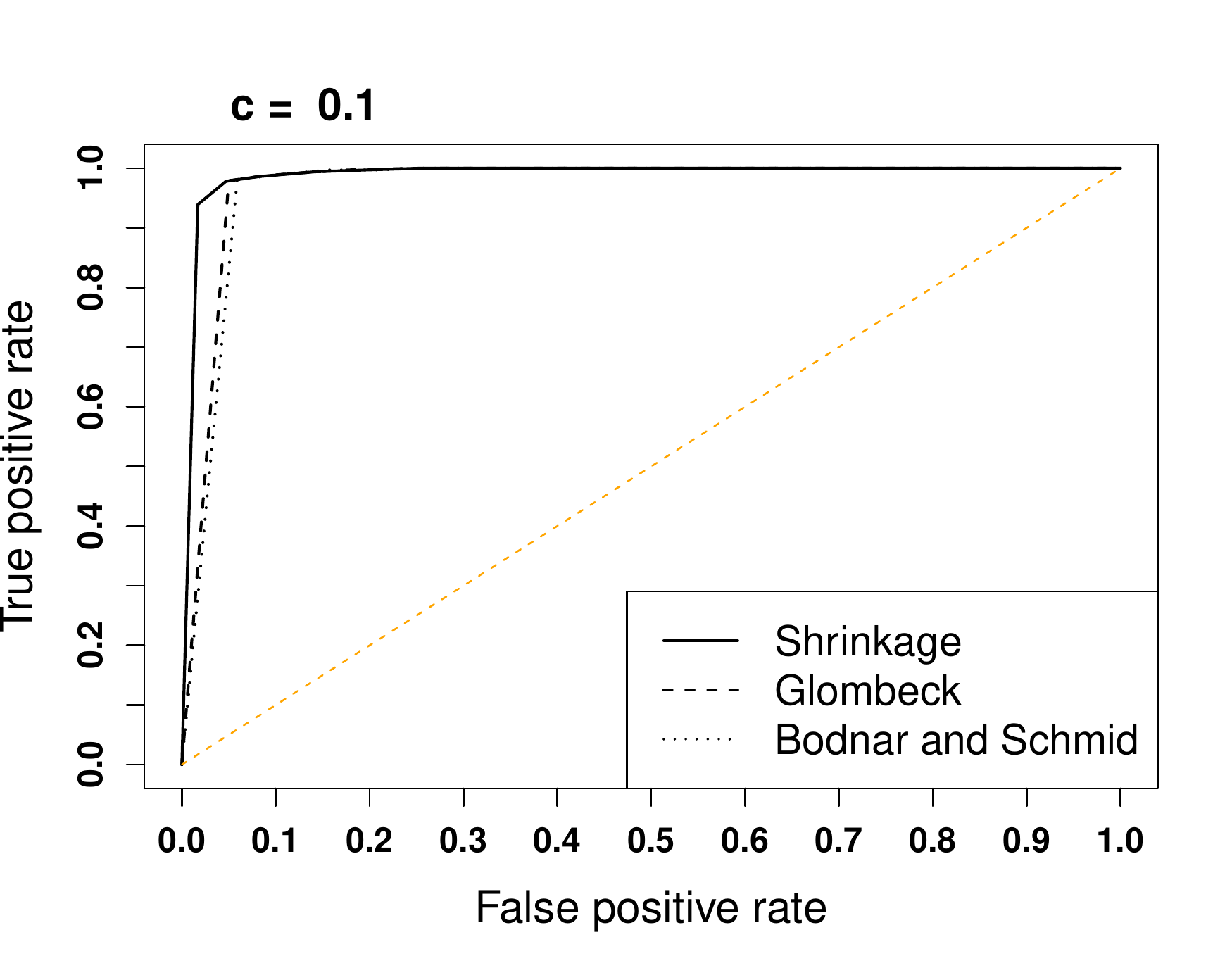} &	\includegraphics[width=0.45\linewidth]{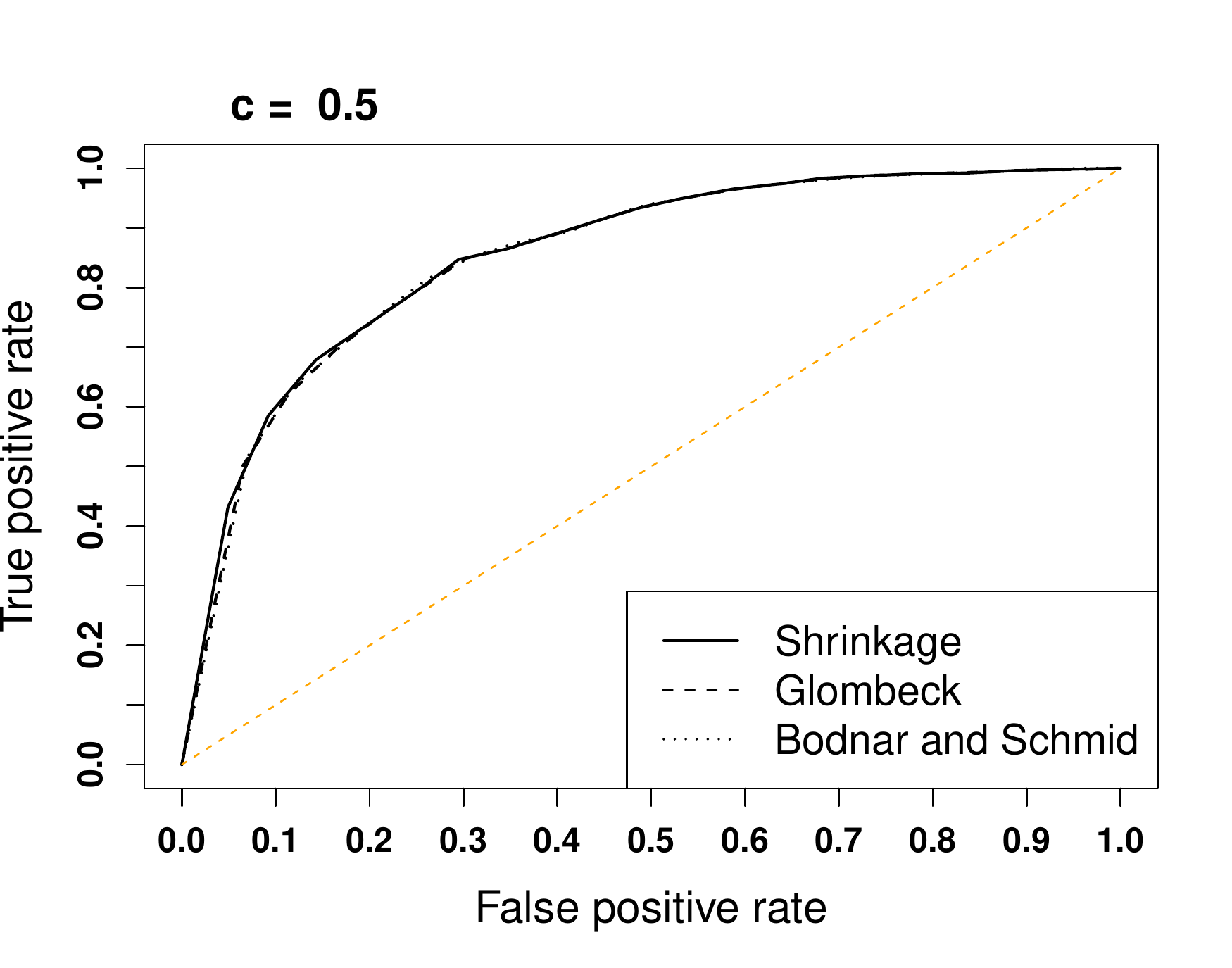}\\
		
			\includegraphics[width=0.45\linewidth]{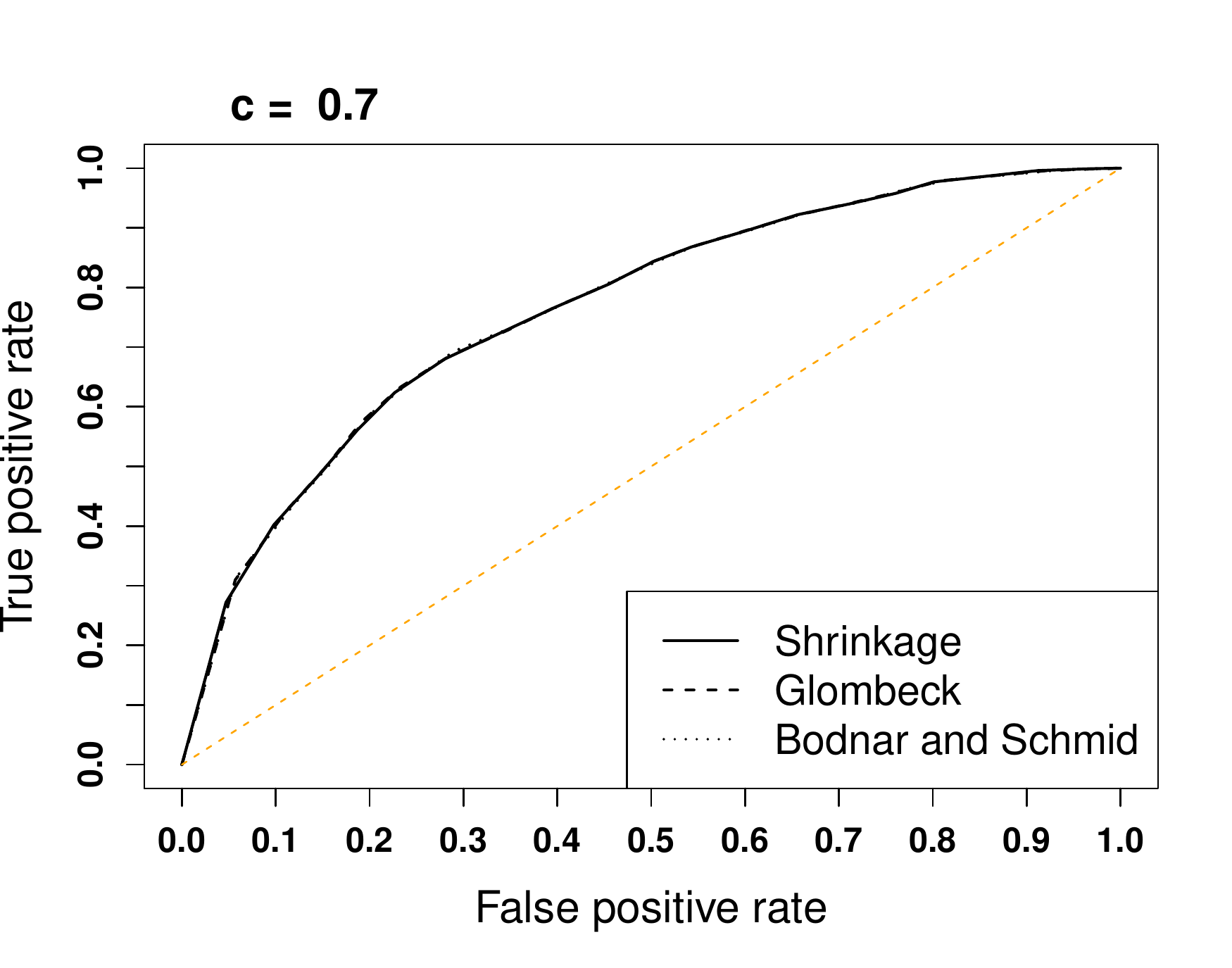} &	\includegraphics[width=0.45\linewidth]{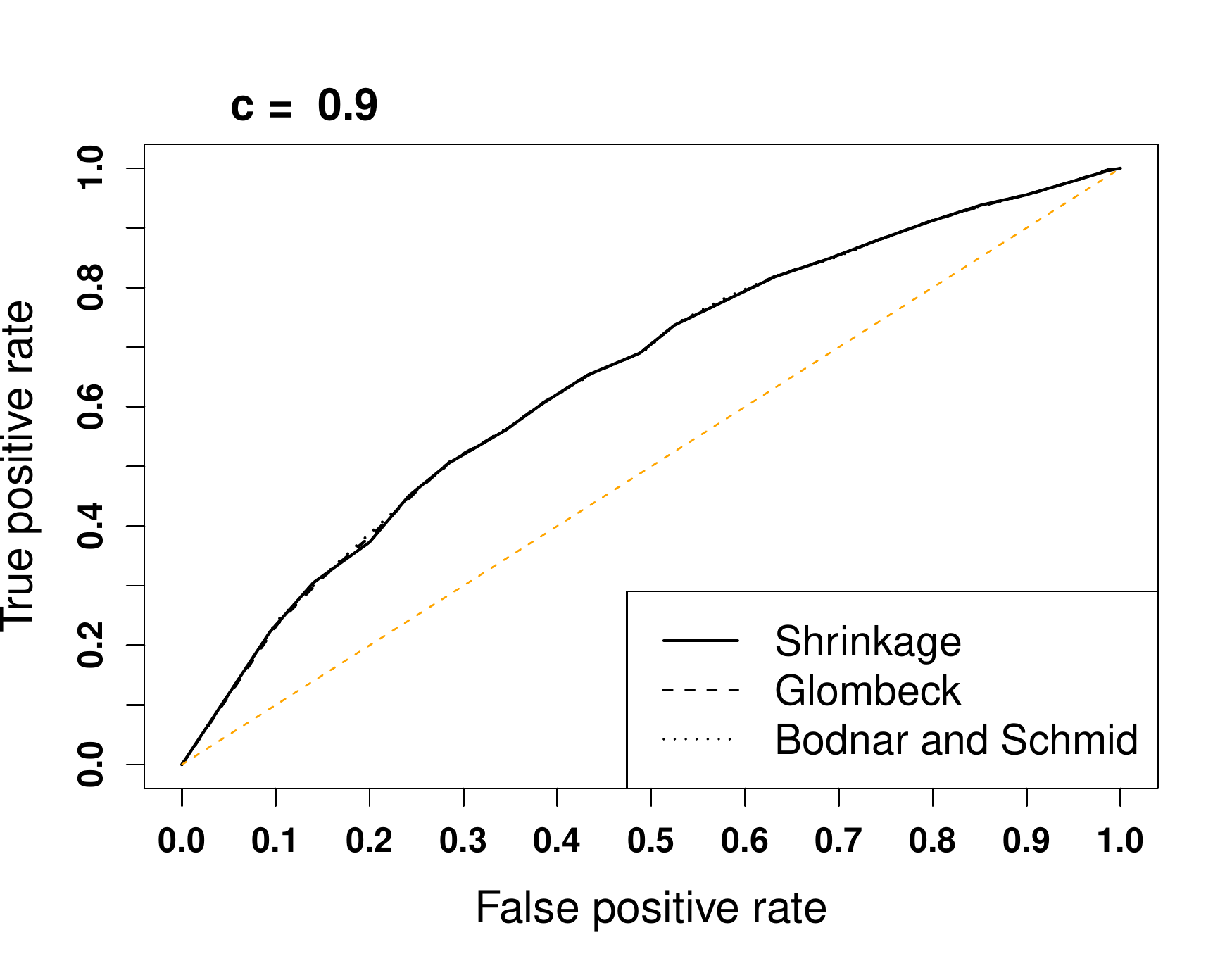}
			
		\end{tabular}
	\caption{ROC of the three tests for different values of $c$, $20\%$ changes on the main diagonal according to scenario given in (\ref{scenario1}) and $n=500$.}
\label{fig:fig4}	
\end{figure}	

\begin{figure}[th!]
		\centering
		\begin{tabular}{cc}
			\includegraphics[width=0.45\linewidth]{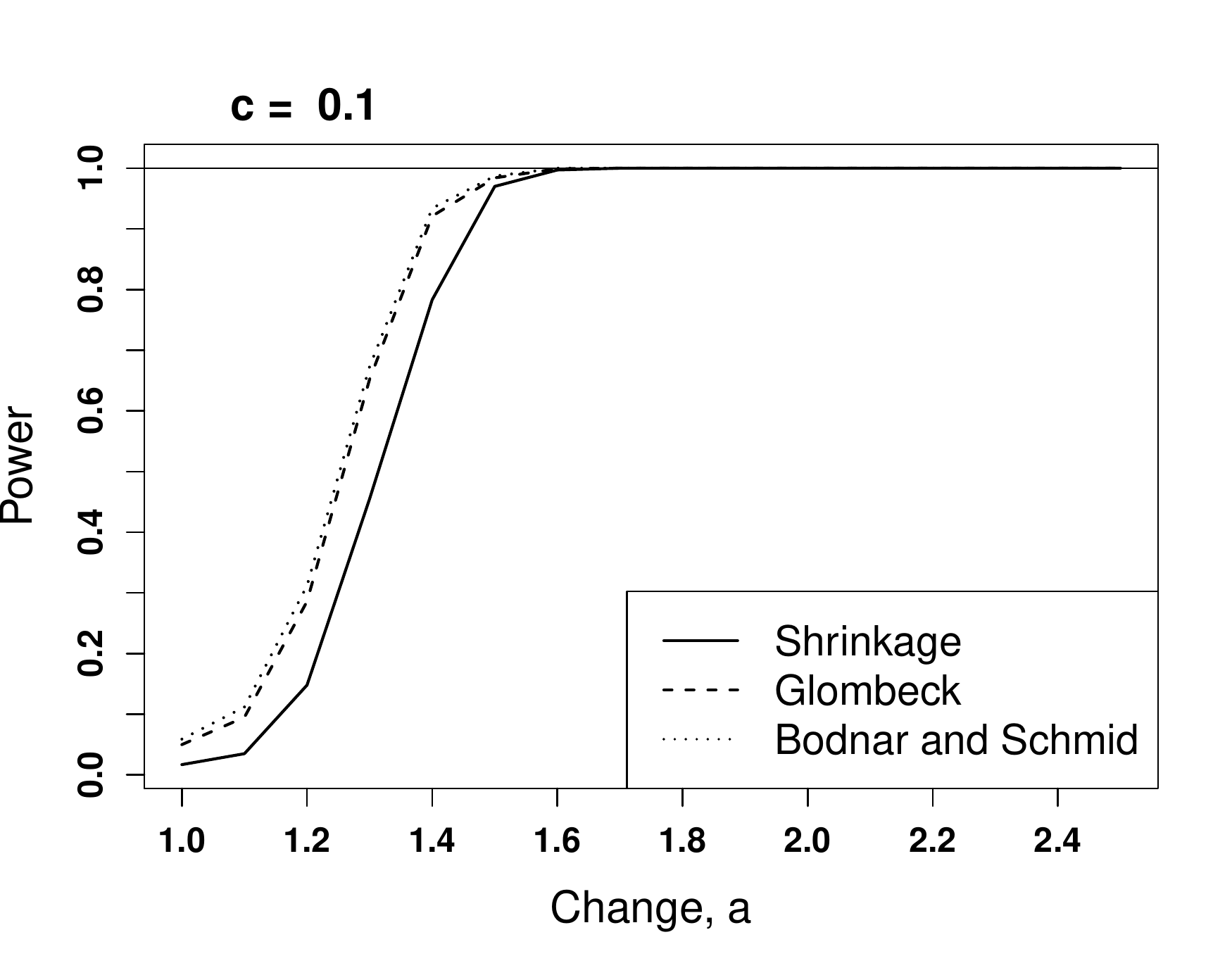} &	\includegraphics[width=0.45\linewidth]{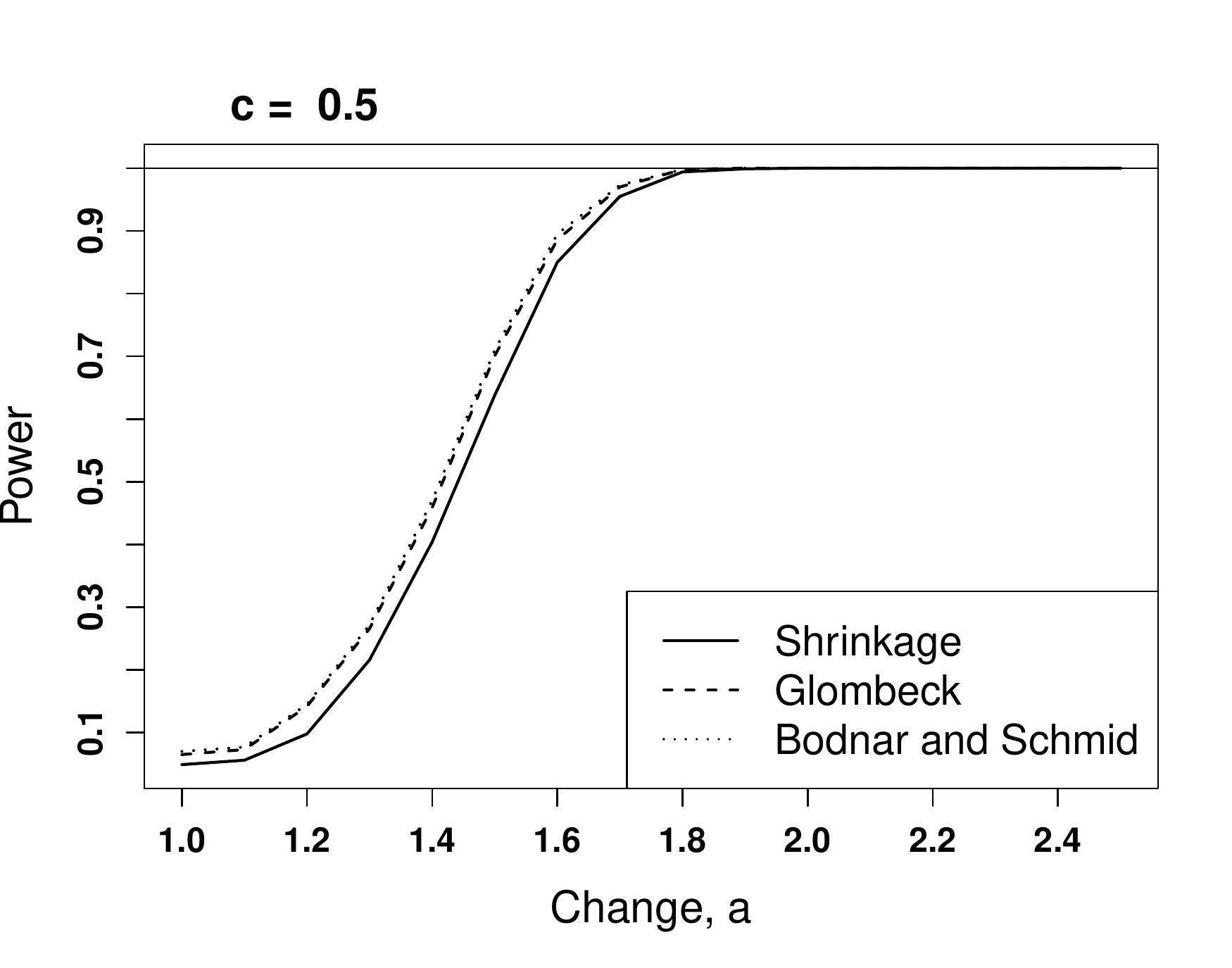}\\
		
			\includegraphics[width=0.45\linewidth]{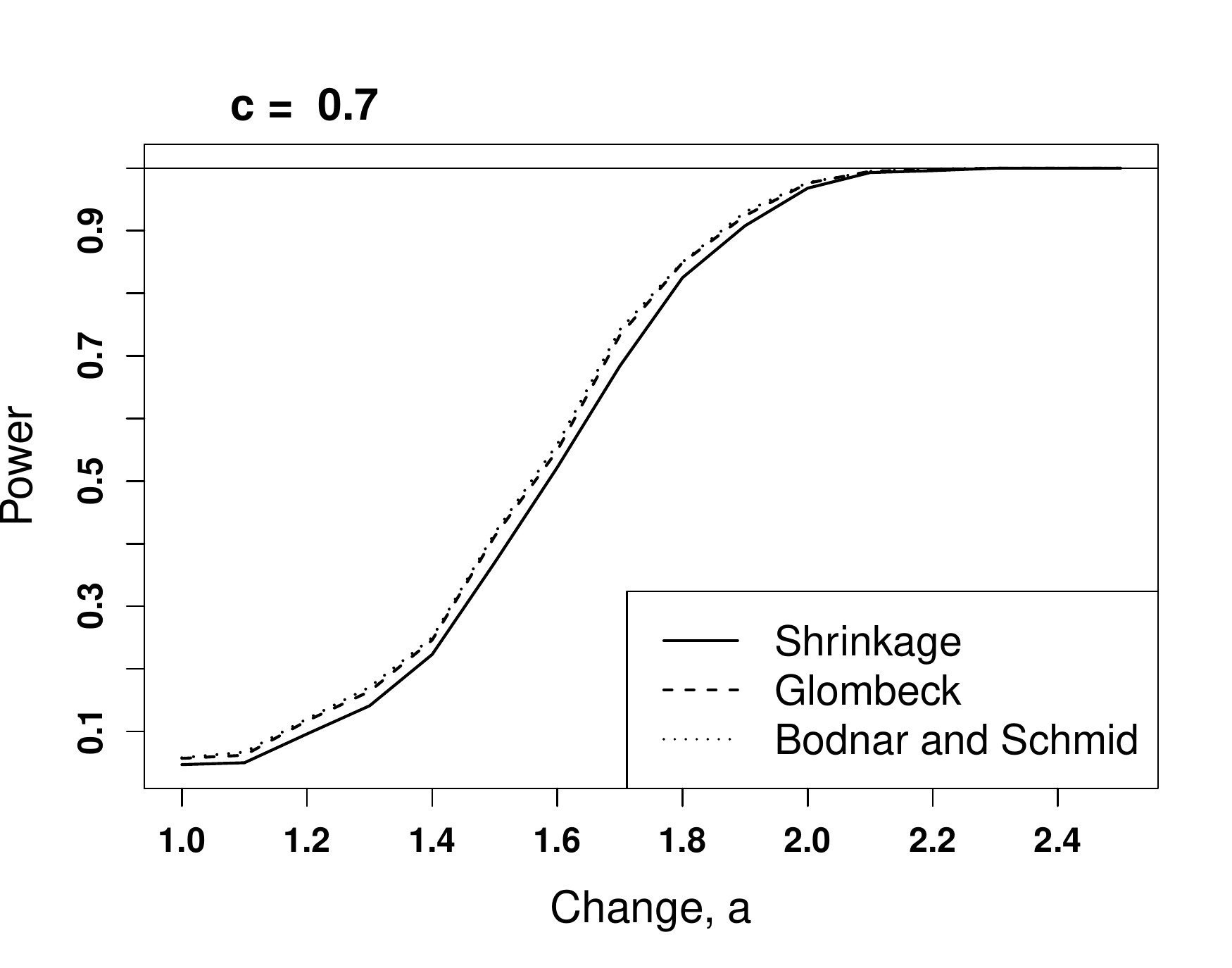} &	\includegraphics[width=0.45\linewidth]{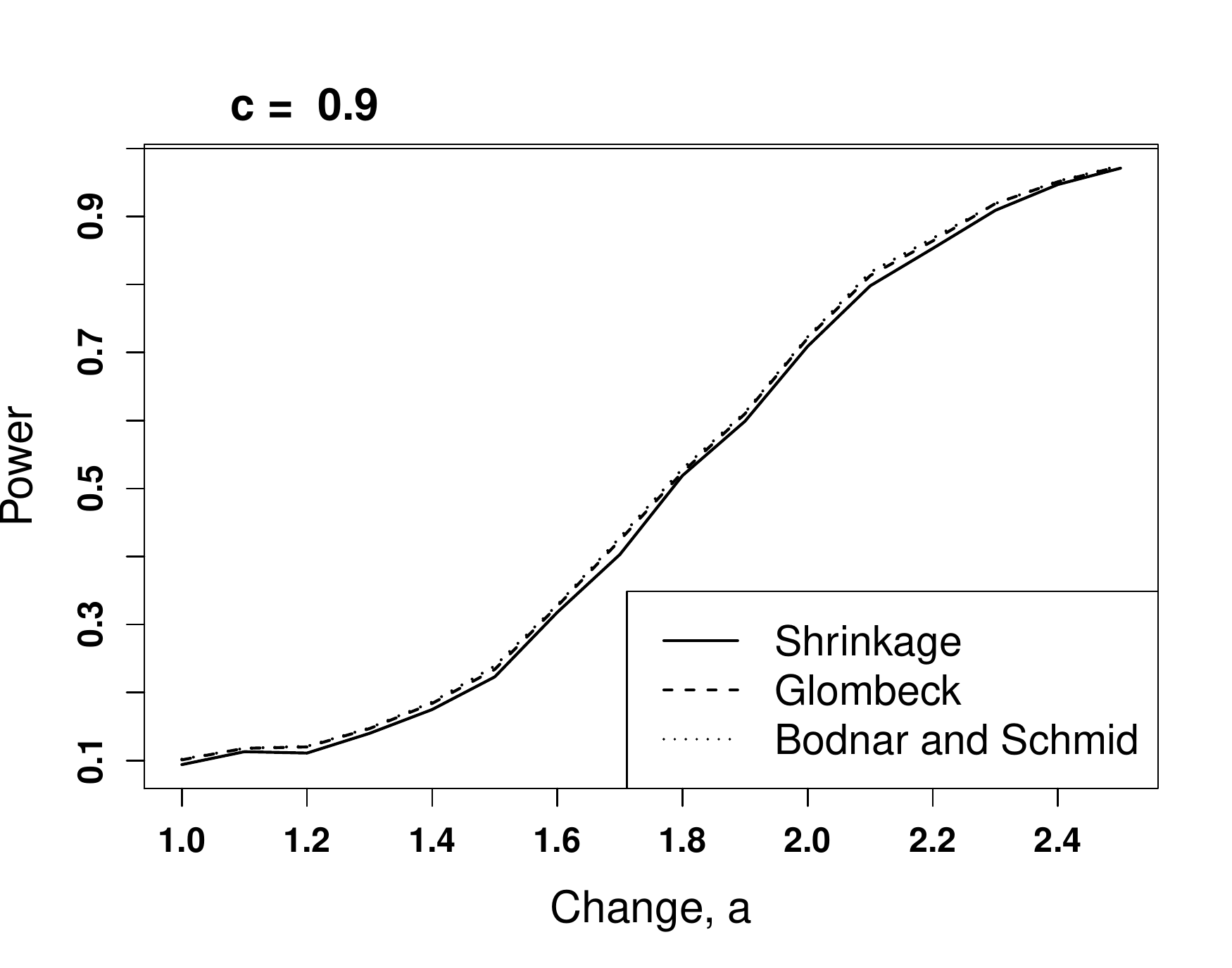}
			
		\end{tabular}
	\caption{Empirical power functions of the three tests for different values of $c$, $50\%$ changes on the main diagonal according to scenario given in (\ref{scenario1}) and $n=500$.}
\label{fig:fig5}	
\end{figure}	

\begin{figure}[th!]
		\centering
		\begin{tabular}{cc}
			\includegraphics[width=0.45\linewidth]{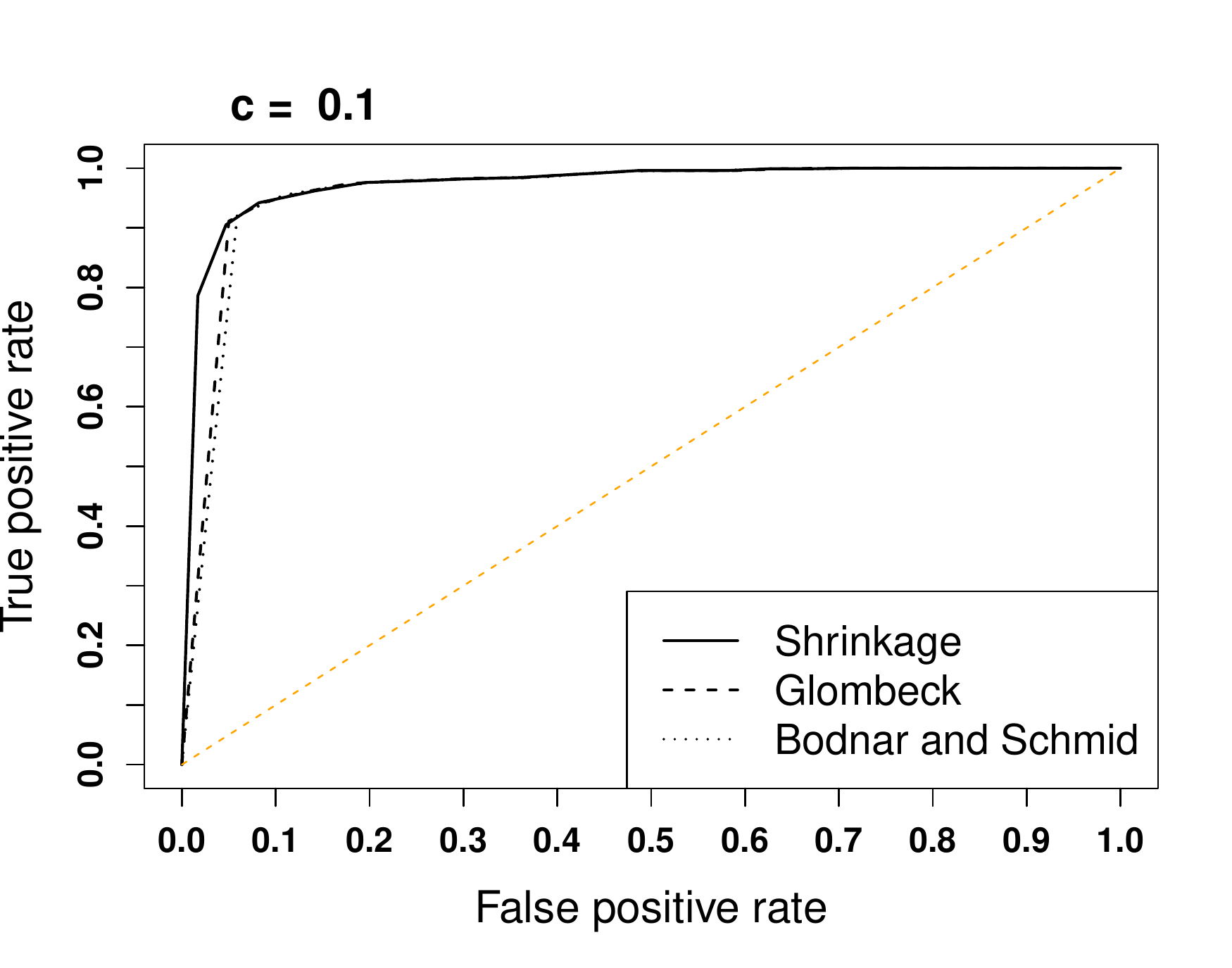} &	\includegraphics[width=0.45\linewidth]{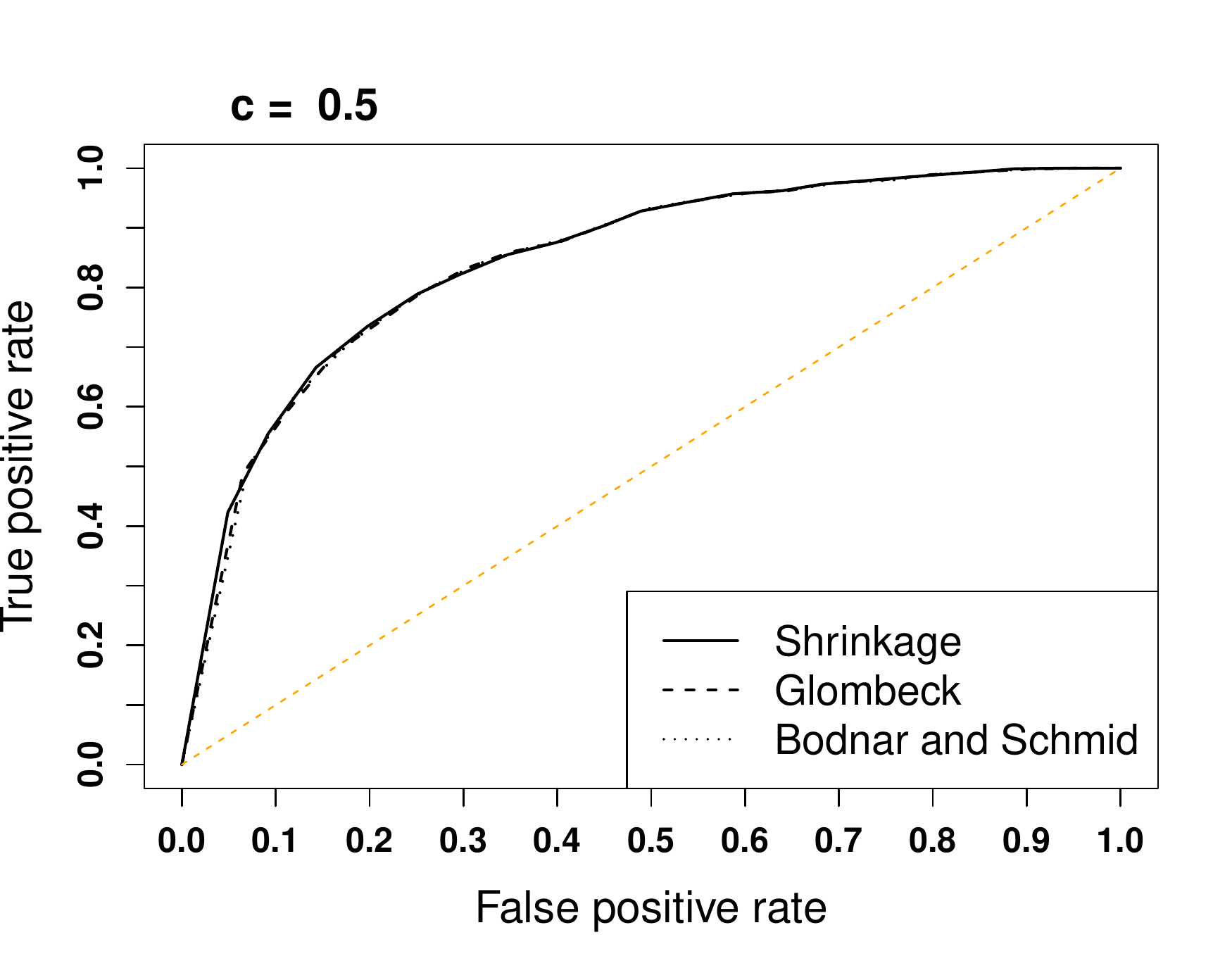}\\
		
			\includegraphics[width=0.45\linewidth]{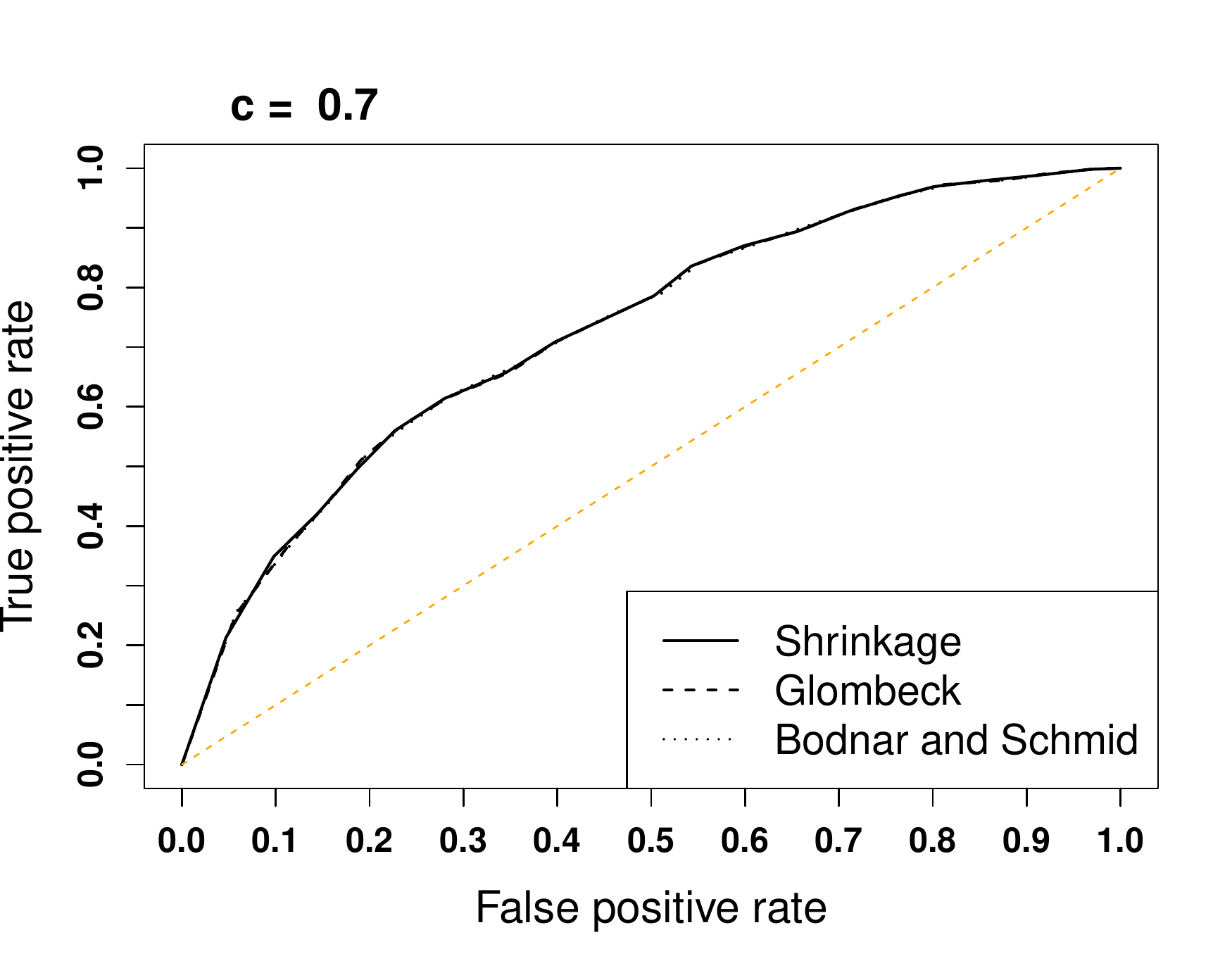} &	\includegraphics[width=0.45\linewidth]{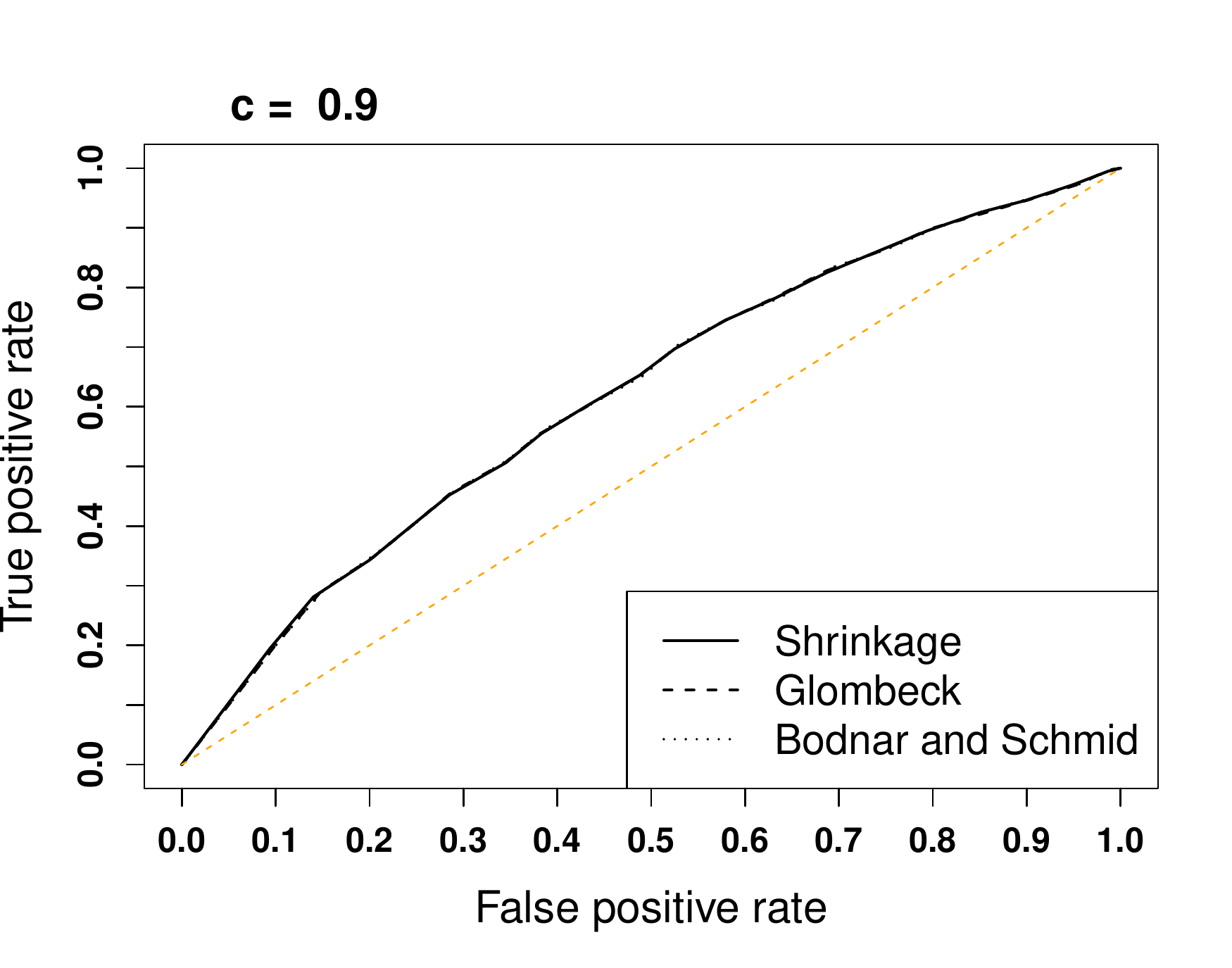}
			
		\end{tabular}
	\caption{ROC of the three tests for different values of $c$, $50\%$ changes on the main diagonal according to scenario given in (\ref{scenario1}) and $n=500$.}
\label{fig:fig6}	
\end{figure}	

In this section, we present the results of a simulation study to compare the power functions and the ROC (Receiver Operating Characteristic) curves of three tests in the case of a non-singular covariance matrix, of five tests when $\bSigma$ is singular and $p<n$, and of two tests when $\bSigma$ is singular and $p>n$. Our simulation study is based on $10^5$ independent realizations of $\Delta$. The significance level $\alpha$ is chosen to be $5\%$ in the figures showing the power functions and $a=1.4$ in the figures with the ROC curves. We set $n=500$, choose $c \in \{0.1,0.5,0.7,0.9\}$ when $\bSigma$ is non-singular, and use $\tilde{c} \in \{0.2,0.6\}$ in the singular case. Furthermore, we consider $p\in \{450, 600\}$ in the singular case.

In order to illustrate the performance of the tests based on the shrinkage approach, the test based on the statistic of \citet*{bodnar2008test}, and the test proposed by \citet*{glombeck} for the non-singular covariance matrix, the empirical power functions for the general hypothesis are evaluated for $m=0.2p$ (Figure \ref{fig:fig3}) and $m=0.5p$ (Figure \ref{fig:fig5}) while the ROC curves are presented in Figure \ref{fig:fig4} ($m=0.2p$) and Figure \ref{fig:fig6} ($m=0.5p$).

In Figure \ref{fig:fig3}, where $20\%$ of the eigenvalues of the covariance matrix are contaminated, we observe a slow increase of the power functions for $c=0.9$ and a better behaviour for smaller values of $c$. In the case $c=0.9$, there is no significant difference in the performance of the tests. For all considered values of $c$, the power curves of Glombeck's test and the test of \citet*{bodnar2008test} are very close to each other and they lie slightly above the power curve of the test based on the shrinkage approach. Some larger deviations are present in the case $c=0.1$. While in terms of the power the tests of \citet*{glombeck} and of \citet*{bodnar2008test} outperform the test based on the shrinkage approach, the opposite conclusion is drawn when the tests are compared by using their ROC curves. Here, we observe that the new approach performs better than the other two competitors. These two different performance results can be explained by the observation that the test based on the shrinkage approach tends to be in general undersized for small values of $c$ which are not of great importance for the proposed high-dimensional approach. Finally, we observe a similar behavior of the tests in terms of both the power functions and the ROC curves in Figures \ref{fig:fig5} and \ref{fig:fig6} for $m=0.5p$.

\begin{figure}[th!]
		\centering
		\begin{tabular}{cc}
			\includegraphics[width=0.45\linewidth]{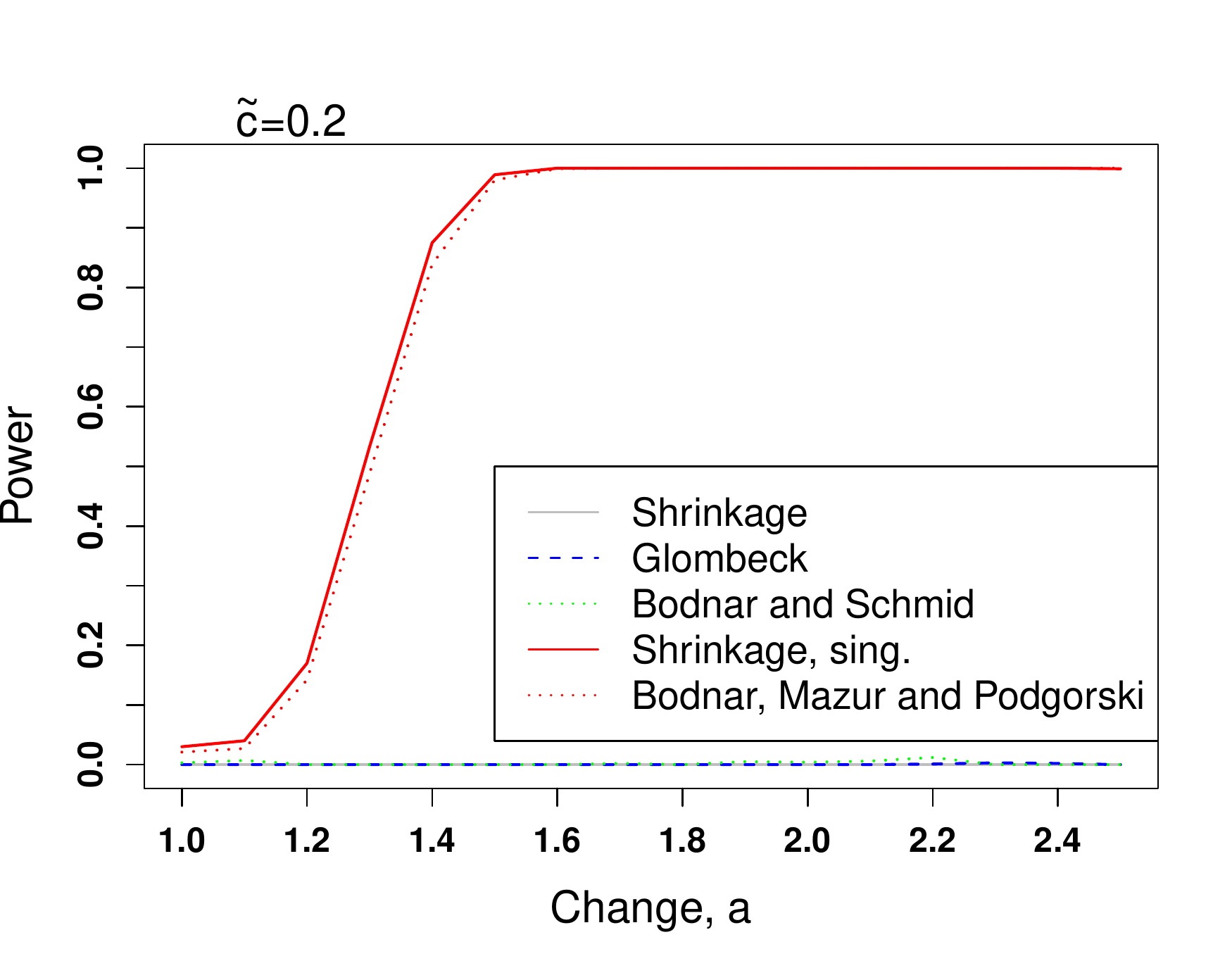} &	\includegraphics[width=0.45\linewidth]{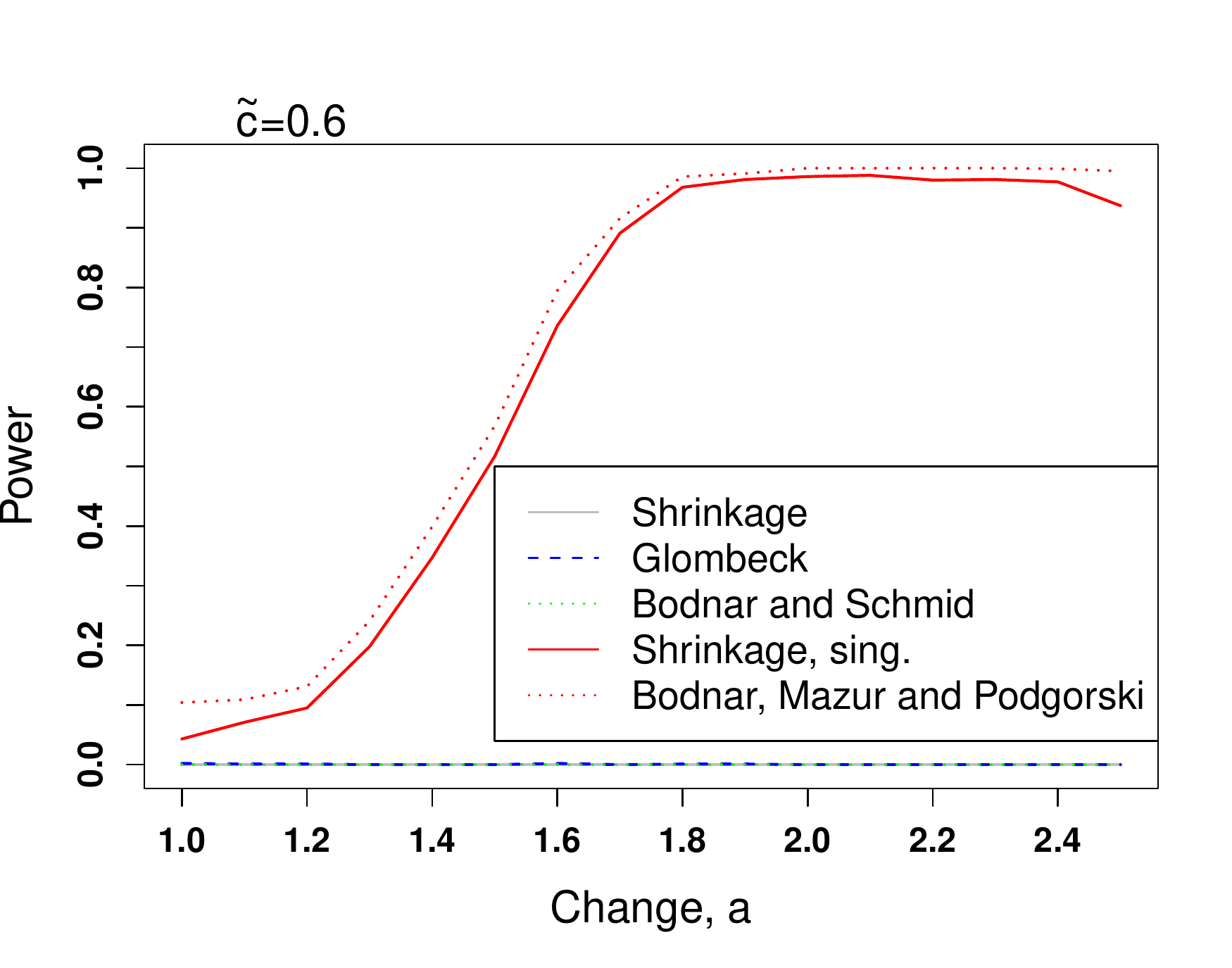}\\
		
			\includegraphics[width=0.45\linewidth]{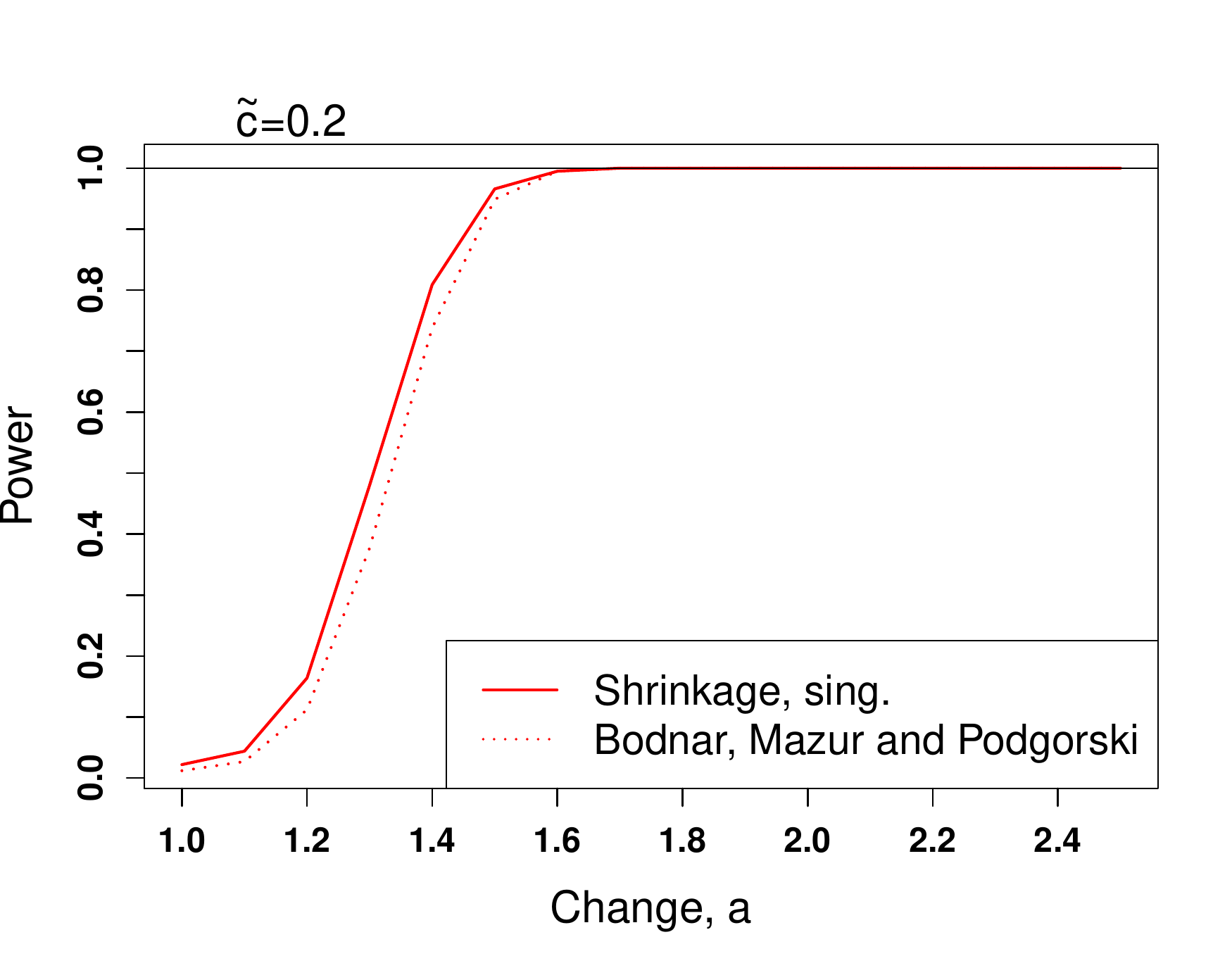} &	\includegraphics[width=0.45\linewidth]{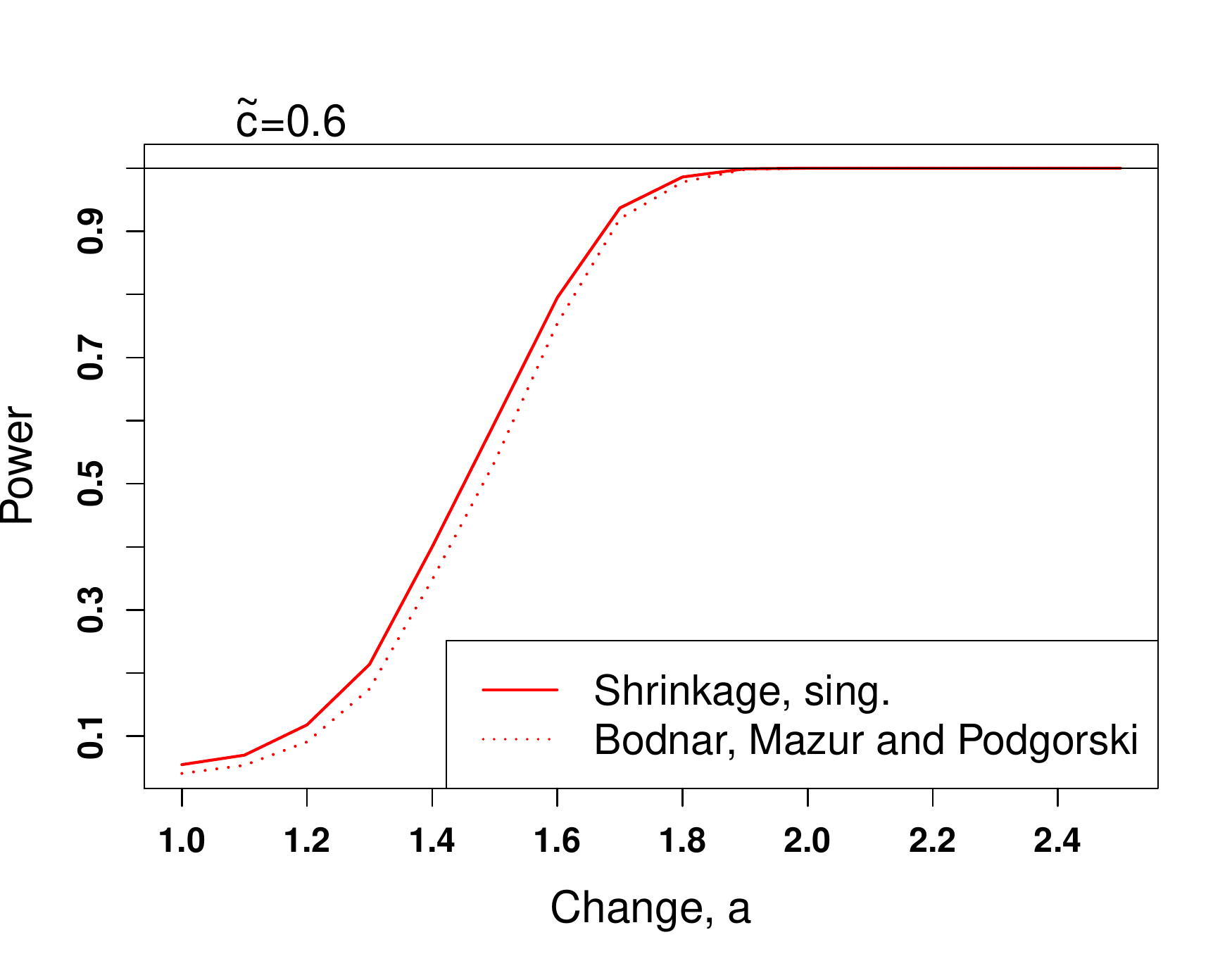}
			
		\end{tabular}
	\caption{Empirical power functions of the three tests derived under a non-singular covariance matrix and of the two tests developed for the singular covariance matrix for different values of $\tilde{c}$, $20\%$ changes on the main diagonal according to scenario given in (\ref{scenario1}), $n=500$, $p=450$ (upper figures) and $p=600$ (lower figures).}
\label{fig:fig7}	
\end{figure}	

\begin{figure}[th!]
		\centering
		\begin{tabular}{cc}
			\includegraphics[width=0.45\linewidth]{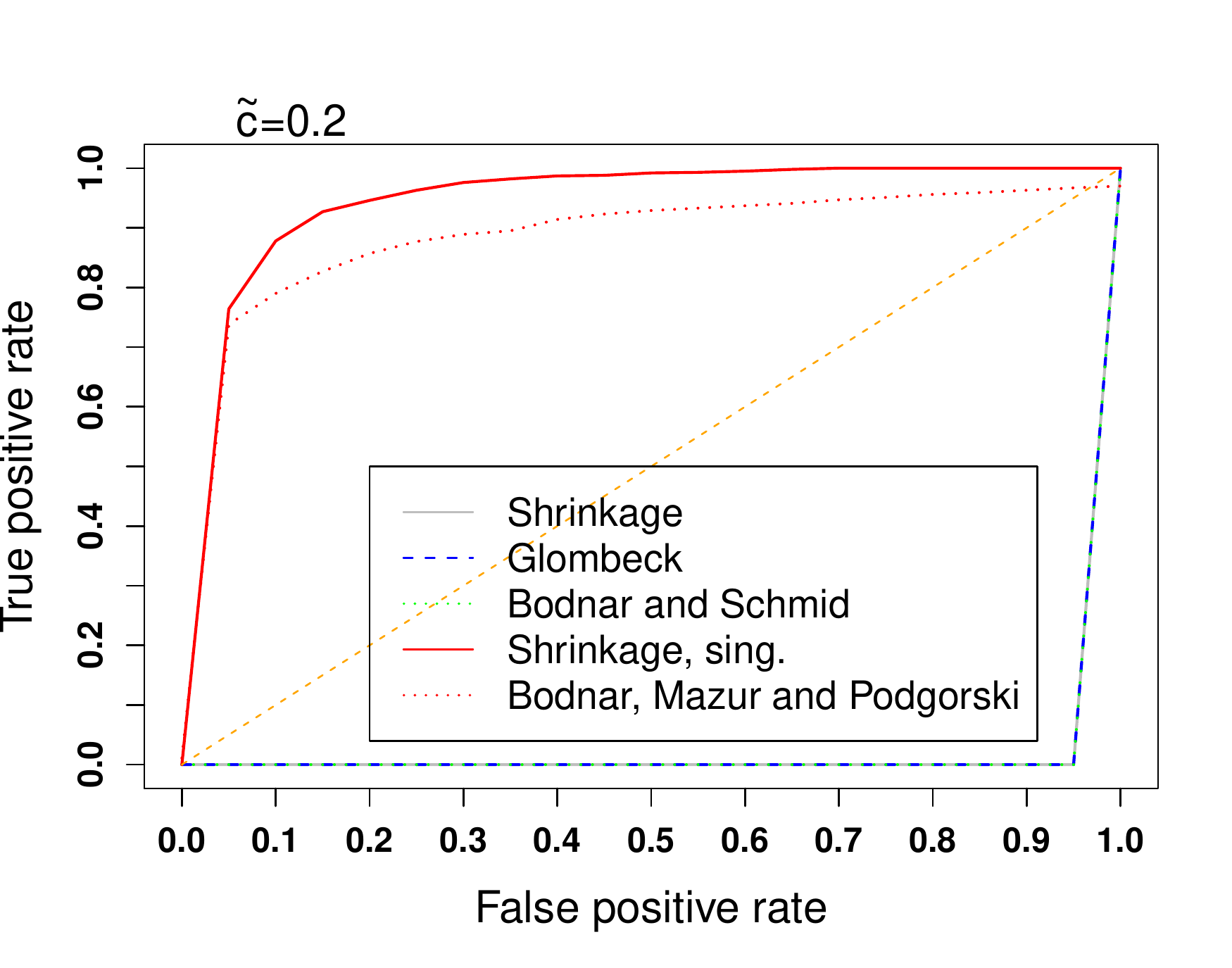} &	\includegraphics[width=0.45\linewidth]{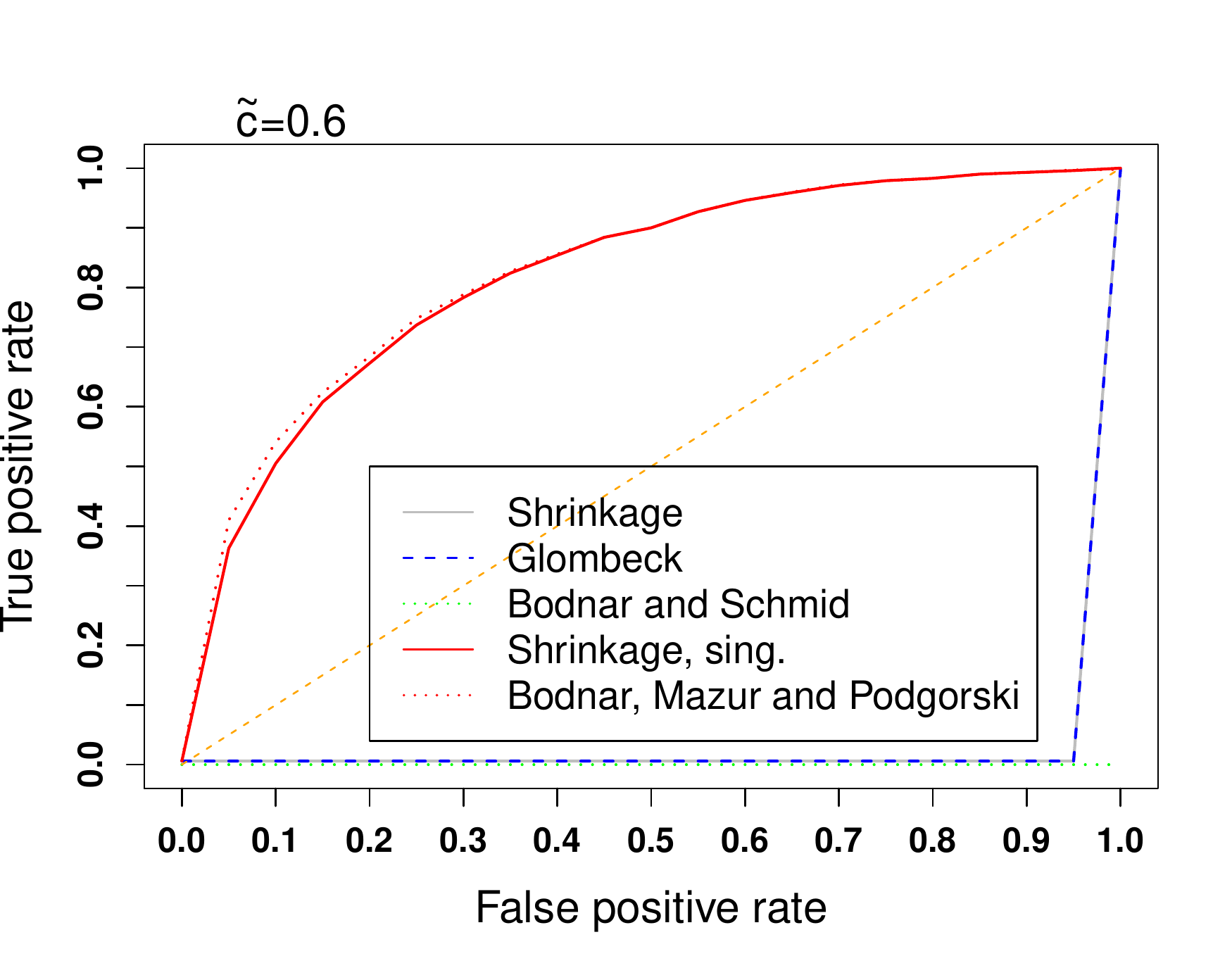}\\
		
			\includegraphics[width=0.45\linewidth]{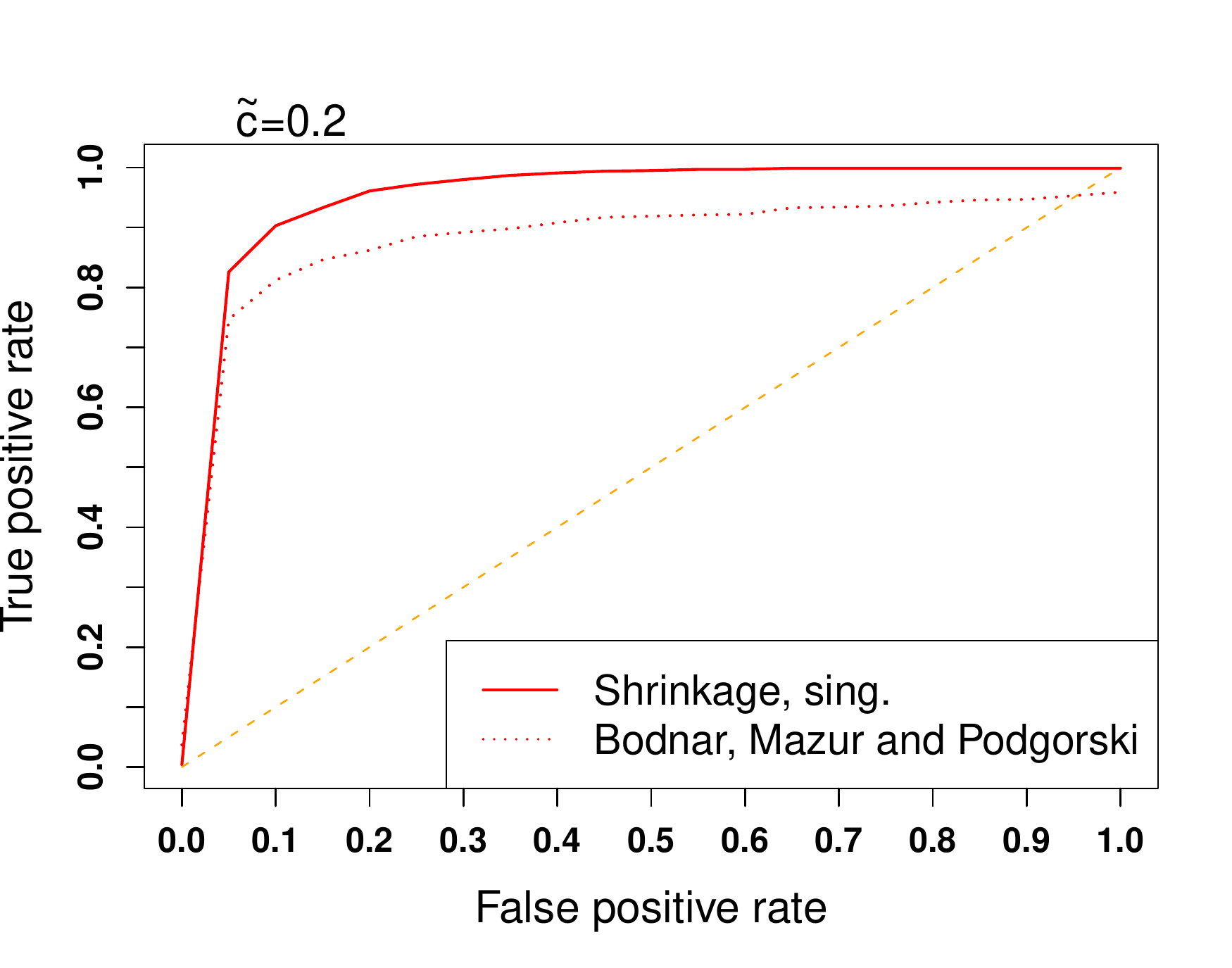} &	\includegraphics[width=0.45\linewidth]{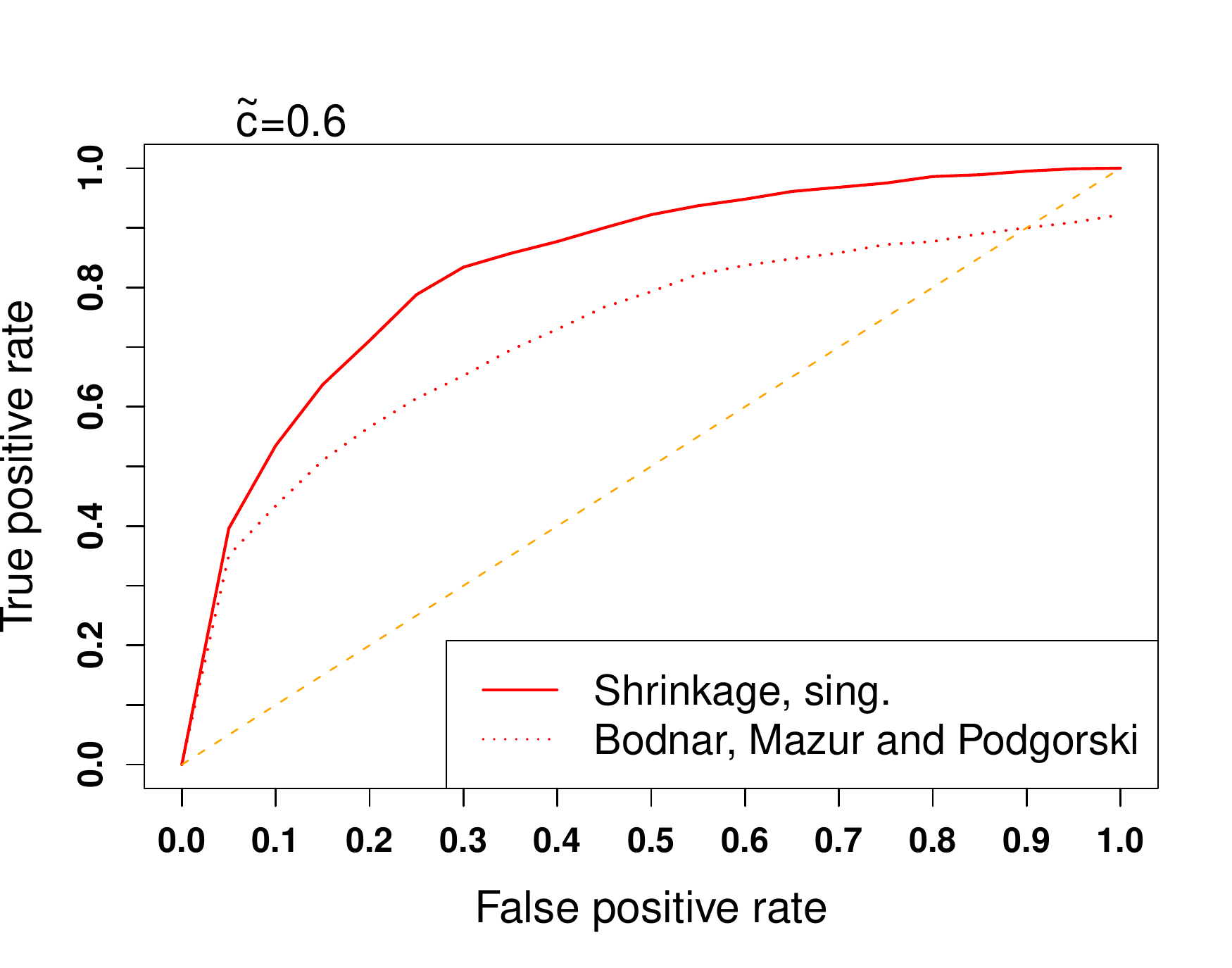}
			
		\end{tabular}
	\caption{ROC of the three tests derived under a non-singular covariance matrix and of the two tests developed for the singular covariance matrix for different values of $\tilde{c}$, $20\%$ changes on the main diagonal according to scenario given in (\ref{scenario1}), $n=500$, $p=450$ (upper figures) and $p=600$ (lower figures).}
\label{fig:fig8}	
\end{figure}

\begin{figure}[th!]
		\centering
		\begin{tabular}{cc}
			\includegraphics[width=0.45\linewidth]{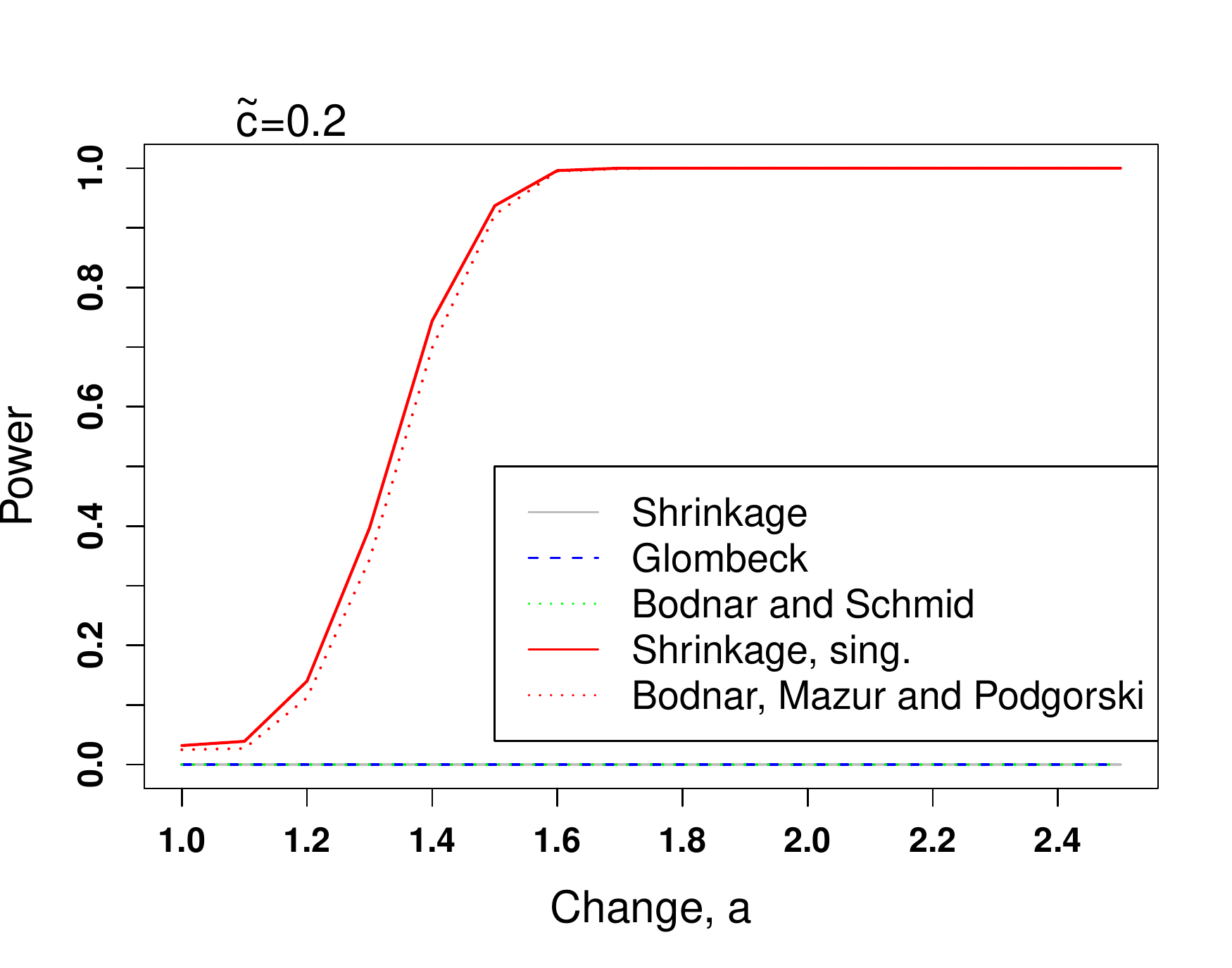} &	\includegraphics[width=0.45\linewidth]{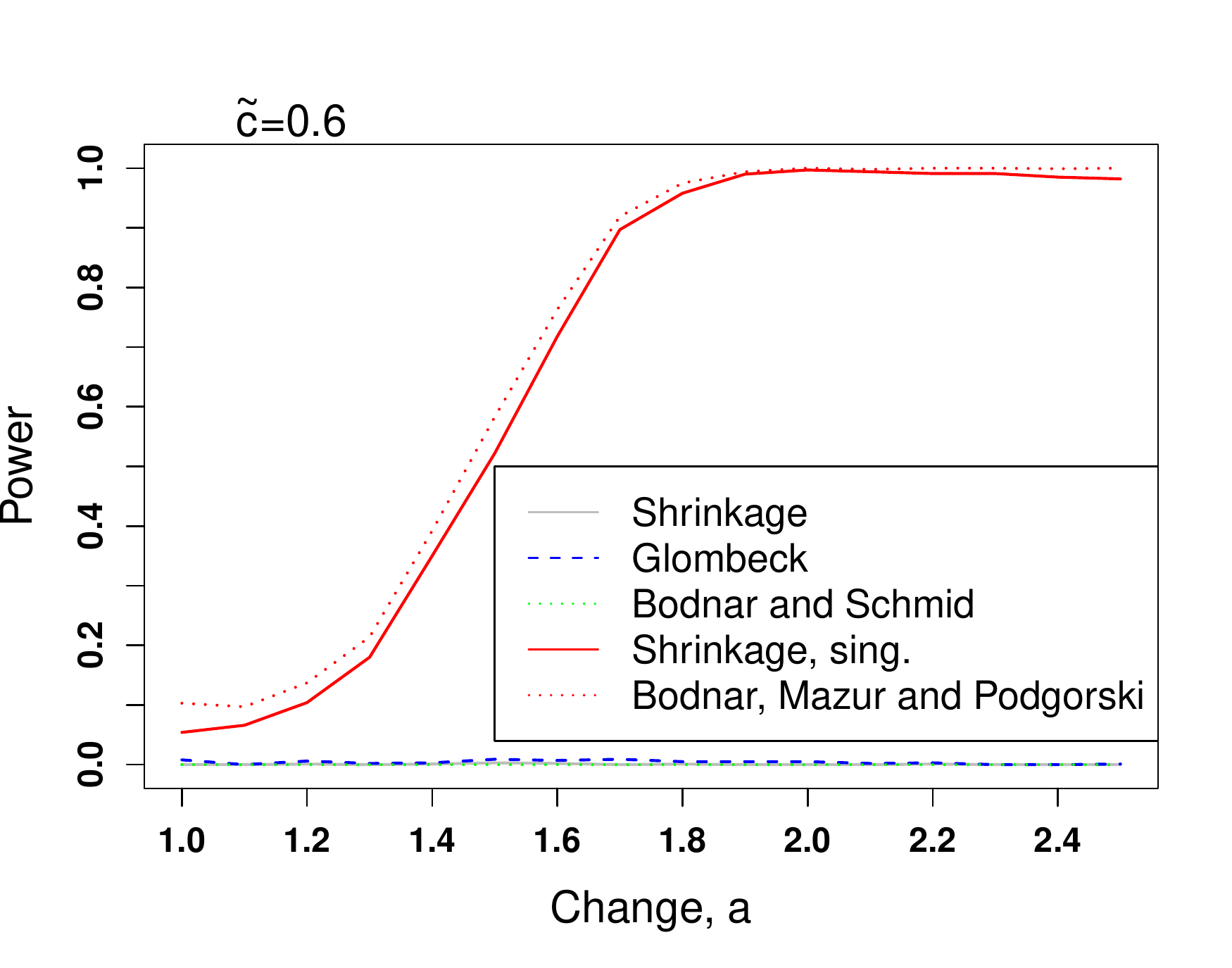}\\
		
			\includegraphics[width=0.45\linewidth]{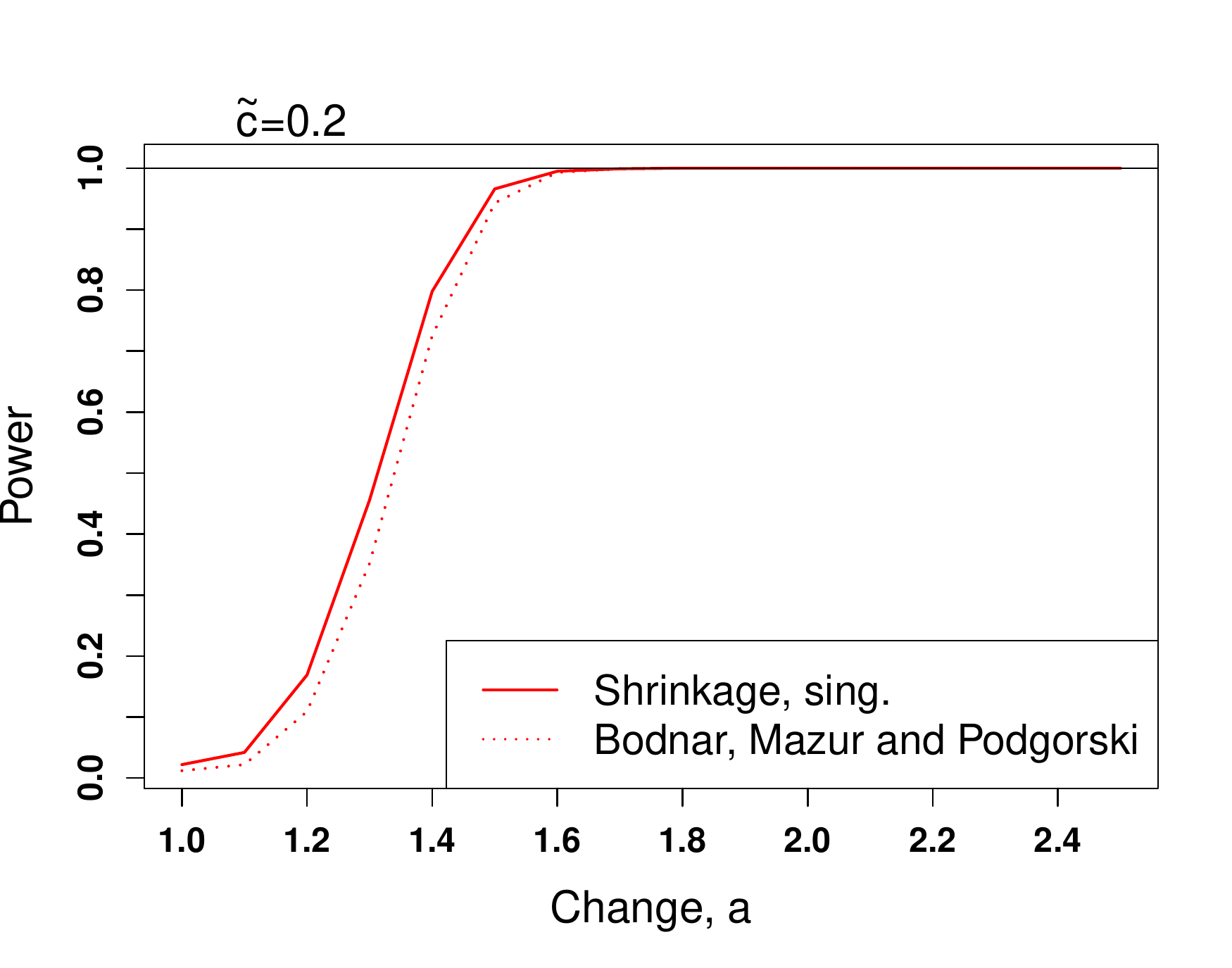} &	\includegraphics[width=0.45\linewidth]{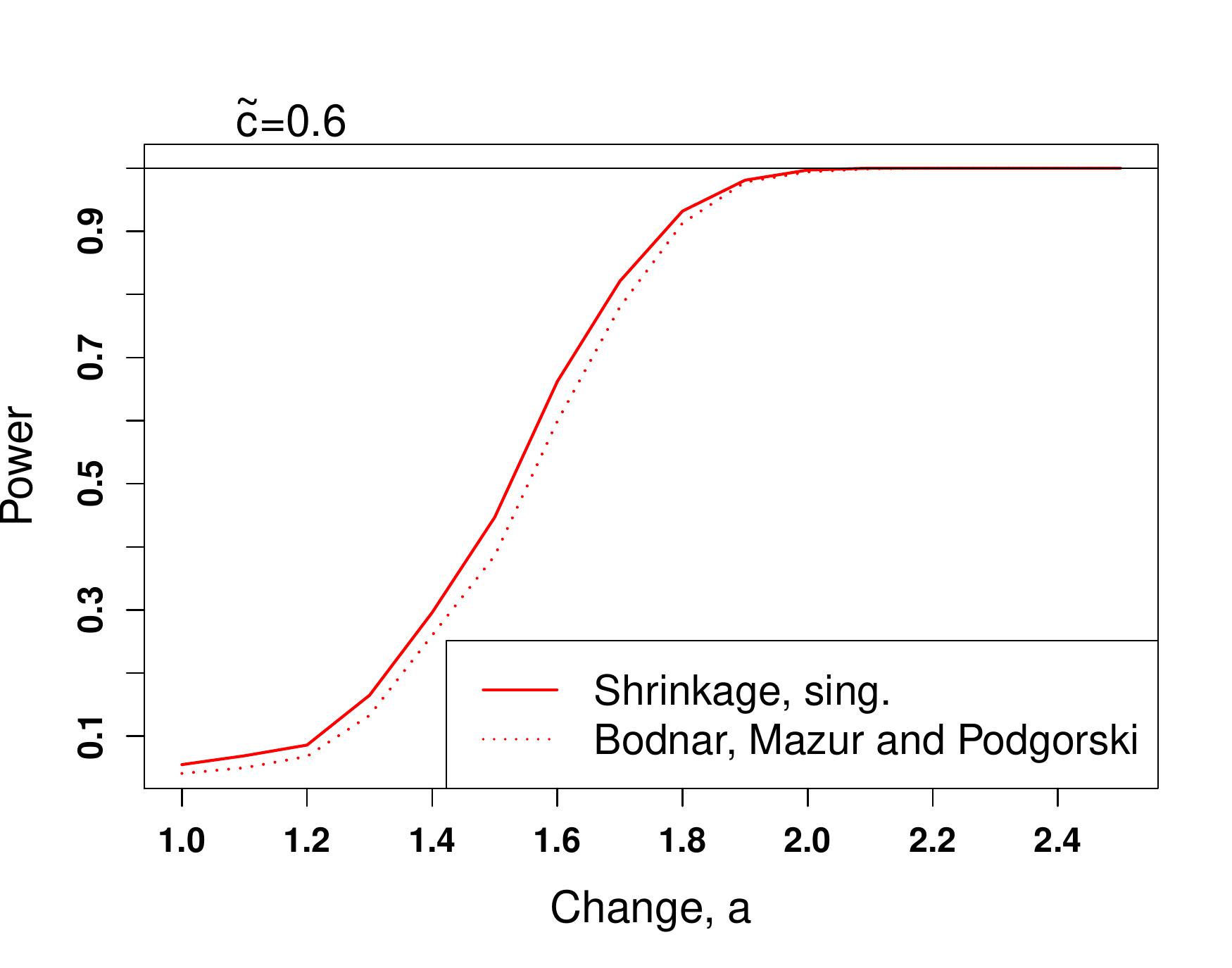}
			
		\end{tabular}
	\caption{Empirical power functions of the three tests derived under a non-singular covariance matrix and of the two tests developed for the singular covariance matrix for different values of $\tilde{c}$, $50\%$ changes on the main diagonal according to scenario given in (\ref{scenario1}), $n=500$, $p=450$ (upper figures) and $p=600$ (lower figures).}
\label{fig:fig9}	
\end{figure}	

\begin{figure}[th!]
		\centering
		\begin{tabular}{cc}
			\includegraphics[width=0.45\linewidth]{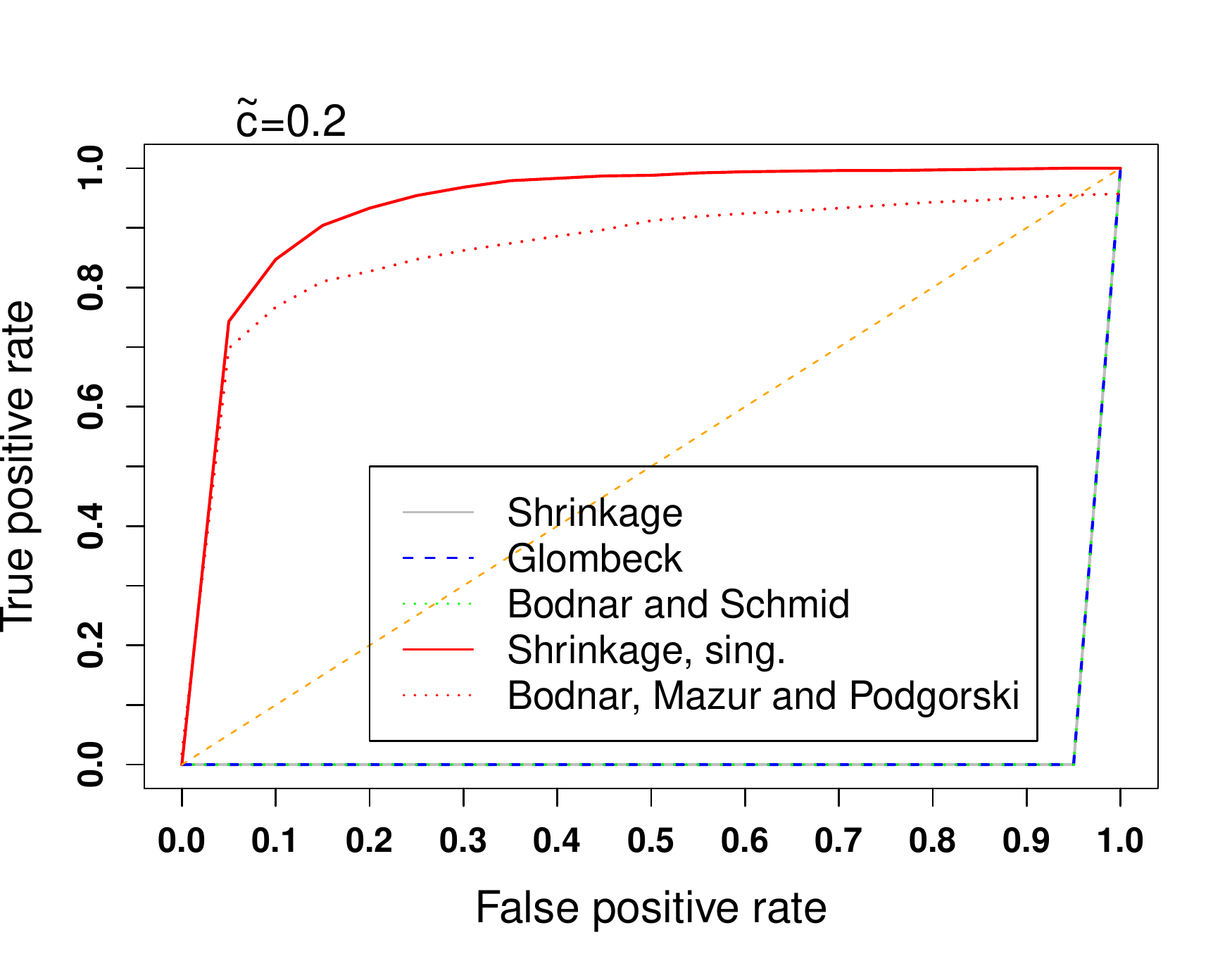} &	\includegraphics[width=0.45\linewidth]{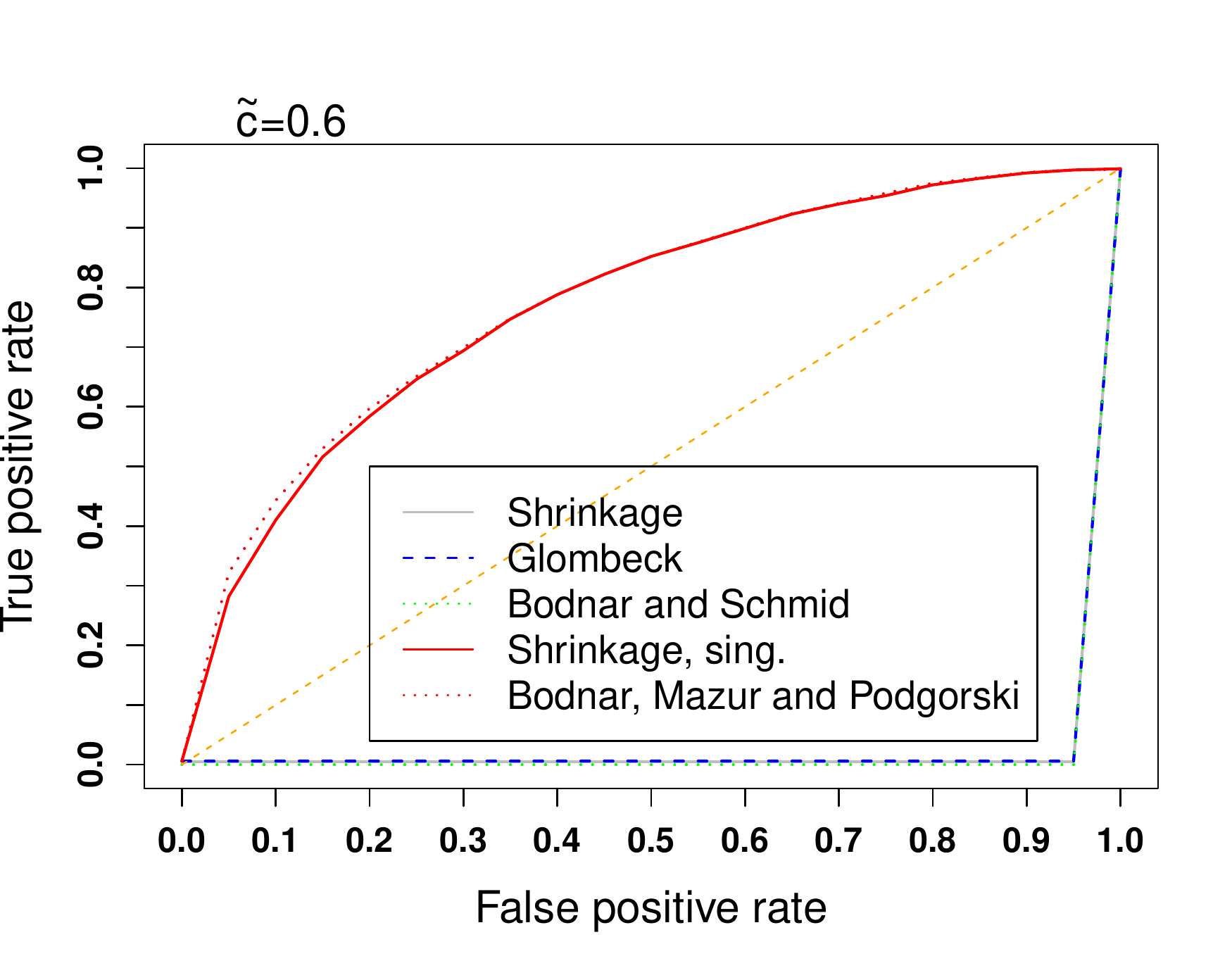}\\
		
			\includegraphics[width=0.45\linewidth]{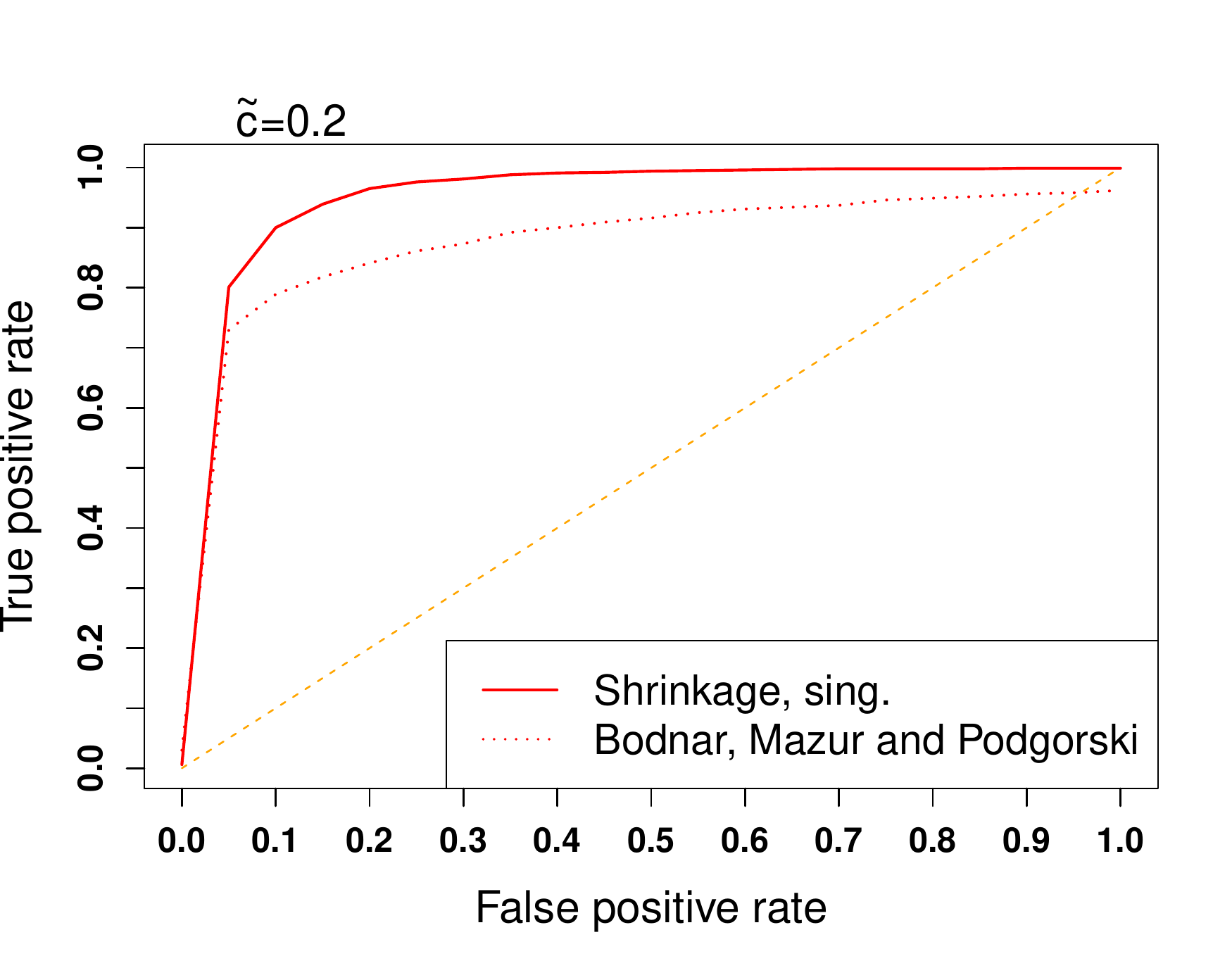} &	\includegraphics[width=0.45\linewidth]{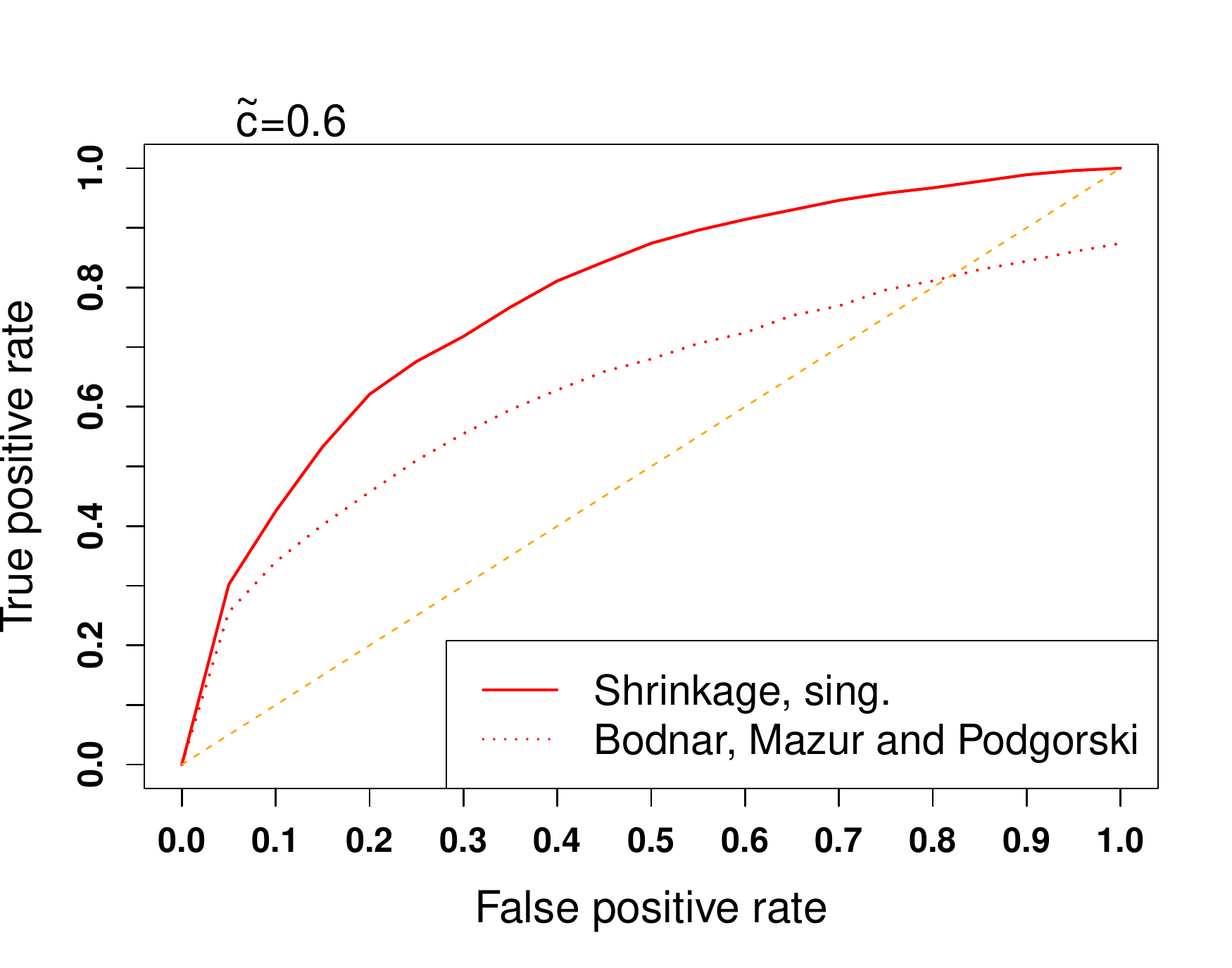}
			
		\end{tabular}
	\caption{ROC of the three tests derived under a non-singular covariance matrix and of the two tests developed for the singular covariance matrix for different values of $\tilde{c}$, $50\%$ changes on the main diagonal according to scenario given in (\ref{scenario1}), $n=500$, $p=450$ (upper figures) and $p=600$ (lower figures).}
\label{fig:fig10}	
\end{figure}	

In Figures \ref{fig:fig7} to \ref{fig:fig10}, we present the results in the case of the singular covariance matrix $\bSigma$ with two possible values for the ranks, namely $rank(\bSigma)\in \{100,300\}$ which corresponds to $\tilde{c} \in \{0.2,0.6\}$. It is remarkable that all three tests, which do not take into account the singularity of the covariance matrix, perform very bad. Both the power functions and the ROC curves are very close to zero in all considered cases. This is due to the fact that under the null hypothesis the computed asymptotic variances for all test statistics are considerably large since the singularity of the covariance matrix was ignored in their derivations. In contrast, the tests of Section III.C, which take into account this singularity in their derivations, provide improvements in both the expressions of the resulting test statistics and in their asymptotic distributions.

Further, we note a very good performance of the test based on the shrinkage approach that takes the singularity of the covariance matrix into account. It outperforms other approaches in almost all considered cases independently of the choice of the performance criterion. Only for $\tilde{c}=0.6$ and $p=450$, the test of \citet*{bodnar2016singular} shows a slightly better power function but this is due to the fact that its type I error is larger. Finally, in terms of the ROC curve the test of \citet*{bodnar2016singular} has not a good performance for moderate and large values of the false positive rate. This result is expected since, the test of \citet*{bodnar2016singular} is a multiple test whose critical values are obtained by employing the Bonferroni correction which appears to be very conservative for moderate and large significance values of the test.

\vspace{0.5cm}
\section{Summary}
The main focus of this study is the inference of the GMVP weights. After constructing an optimal portfolio, an investor is interested to know whether or not the weights of the portfolio he is holding are still optimal at a fixed time point. For that reason, we investigate several asymptotic and exact statistical procedures for detecting deviations in the weights of the GMVP. One test is based on the sample estimator of the GMVP weights, whereas another uses its shrinkage estimator. To the best of our knowledge, the shrinkage approach, which is very popular in point estimation, is applied in test theory for the first time. The asymptotic distributions of both test statistics are obtained under the null and alternative hypotheses in a high-dimensional setting. This finding is a great advantage with respect to other approaches that appear in the literature which do not elaborate on the distribution under the alternative hypothesis (e.g., \citet*{glombeck}). Finally, we deal with the case of a singular covariance matrix by deriving new testing procedures for the weights of the GMVP that are adopted to the singularity. The distributions of the resulting test statistics are obtained under both the null and alternative hypothesis.

In order to compare the performances of the proposed procedures, the empirical power functions of the derived tests are determined. It is shown that the test based on the shrinkage approach performs uniformly better than the other tests considered in the analysis in terms of both the power function and the ROC curve comparisons when the covariance matrix is singular. The new approach appears to be very promising for testing the portfolio weights in a high-dimensional situation. For a specific scenario, we also have studied a problem how good the power function of the asymptotic test based on the Mahalanobis distance approximates the power of the corresponding test and found good results already for moderate sample size, like $n=500$ with $p=\{50, 250, 350, 450\}$. Surely, these results could not be considered as a general statement to the problem and further investigation in this direction should be done. A similar topic should also be investigated for the test based on a shrinkage estimator, although only asymptotic results are available in the latter case.

\section*{Acknowledgement} The authors would like to thank Professor Pier Luigi Dragotti, Professor Byonghyo Shim, Professor Mathini Sellathurai, and anonymous Reviewers for their helpful suggestions. This research was partly supported by the German Science Foundation (DFG) via the projects BO 3521/3-1 and SCHM 859/13-1 ''Bayesian Estimation of the Multi-Period Optimal Portfolio Weights and Risk Measures'' and by the Swedish Research Council (VR) via the project ''Bayesian Analysis of Optimal Portfolios and Their Risk Measures''.

\bibliography{bibliography}

\section*{Appendix}

In this section, the proofs of Theorems are given.

Let the symbol $\stackrel{d}{=}$ denote equality in distribution. In Lemma \ref{lem0}, we first derive a stochastic representation for $T_n$.

\begin{lemma}\label{lem0}
Under the conditions of Theorem \ref{th1}, the stochastic representation of $T_n$ is expressed as
\begin{equation}\label{stoch_pres_Tn}
T_n \stackrel{d}{=} \frac{n-p}{p-1} \frac{(\sqrt{\lambda_n \xi_3}+\omega_1)^2+\xi_4}{\xi_2},
\end{equation}
where $\omega_1\sim \mathcal{N}(0,1)$, $\xi_2\sim \chi^2_{n-p}$, $\xi_3\sim \chi^2_{n-1}$, and $\xi_4\sim \chi^2_{p-2}$; $\omega_1$, $\xi_2$, $\xi_3$, and $\xi_4$ are independent.
\end{lemma}

\begin{proof}\textit{of Lemma \ref{lem0}:}
Let $\mathbf{L}$ be a $(p-1)\times p$ matrix such that $\mathbf{\hat{w}}^{*}_{n}=\mathbf{L} \mathbf{\hat{w}}_{n}$, i.e., it transforms the vector of the estimated GMVP weights into the vector of its $(p-1)$ first components. We define $\mathbf{M}^\prime = (\mathbf{L}^\prime, \mathbf{1})$ and
$$\mathbf{M} \mathbf{\Sigma}^{-1} \mathbf{M}^\prime = \{\mathbf{H}_{ij}\}_{i,j=1,2}, \quad \mathbf{M} \, \hat{\mathbf{\Sigma}}^{-1}_n \,
\mathbf{M}^{\prime} = \{\hat{\mathbf{H}}_{ij}\}_{i,j=1,2}$$
with
$\mathbf{{H}}_{22} = \mathbf{1}^\prime \mathbf{\Sigma}^{-1}\mathbf{1}$,
$\hat{\mathbf{H}}_{22} = \mathbf{1}^\prime \hat{\mathbf{\Sigma}}^{-1}_n \mathbf{1}$,
$\mathbf{H}_{12}=\mathbf{L} \mathbf{\Sigma}^{-1} \mathbf{1}$, $\hat{\mathbf{H}}_{12}=\mathbf{L}
\hat{\mathbf{\Sigma}}^{-1}_n\mathbf{1}$, $\mathbf{H}_{11} = \mathbf{L} \mathbf{\Sigma}^{-1} \mathbf{L}^\prime$,
and $\mathbf{\hat{H}}_{11} = \mathbf{L} \hat{\mathbf{\Sigma}}^{-1}_n \mathbf{L}^\prime$.

Since $(n-1)\hat{\mathbf{\Sigma}} \sim W_p(n-1, \mathbf{\Sigma})$ ($p$-dimensional Wishart distribution with $n-1$ degrees of freedom and covariance matrix $\mathbf{\Sigma}$) and $rank(\mathbf{M})=p$ we get
with \citet[Theorem 3.2.11]{muirhead2009aspects} that
\[ (n-1) ( \mathbf{M} \hat{\mathbf{\Sigma}}^{-1}_n \mathbf{M}^\prime )^{-1} \sim W_{p}(n-1, ( \mathbf{M} \,
\mathbf{\Sigma}^{-1} \, \mathbf{M}^\prime )^{-1} ), \]
and, consequently (see, \citet[Theorem 3.4.1]{GuptaNagar2000}),
\[ (n-1)^{-1} \mathbf{M} \hat{\mathbf{\Sigma}}^{-1}_n \mathbf{M}^\prime \sim W_{p}^{-1}(n+p, \mathbf{M} \,
\mathbf{\Sigma}^{-1} \, \mathbf{M}^\prime ). \]

Recalling the definition of $\mathbf{\hat{w}}^{*}_{n}$ and $\hat{\mathbf{Q}}^*_n$, we get from Theorem 3 of \citet*{BodnarOkhrin2008} that
\begin{enumerate}
\item $\hat{\mathbf{H}}_{22}= \mathbf{1}^\prime \hat{\mathbf{\Sigma}}^{-1}_n \mathbf{1}$ is independent of $\hat{\mathbf{H}}_{12}/\hat{\mathbf{H}}_{22}=\hat{\mathbf{w}}_n^*$ and
    $\hat{\mathbf{H}}_{11}-\hat{\mathbf{H}}_{12}\hat{\mathbf{H}}_{21}/\hat{\mathbf{H}}_{22}
    =\hat{\mathbf{Q}}^*_n$
\item $(n-1)^{-1}\hat{\mathbf{H}}_{22}= (n-1)^{-1}\mathbf{1}^\prime \hat{\mathbf{\Sigma}}^{-1}_n \mathbf{1}\sim W_1^{-1}(n-p+2,\mathbf{1}^\prime \mathbf{\Sigma}^{-1} \mathbf{1})$
and, consequently,
\begin{equation}\label{th1_eq2}
\xi_2=(n-1) \frac{\mathbf{1}^\prime \mathbf{\Sigma}^{-1}\mathbf{1}}{\mathbf{1}^\prime \hat{\mathbf{\Sigma}}^{-1}_n \mathbf{1}} \sim \chi^2_{n-p},
\end{equation}
\item
\begin{eqnarray*}
&&(n-1)^{-1} \hat{\mathbf{H}}_{12}|(n-1)^{-1}\hat{\mathbf{H}}_{22}, (n-1)^{-1}\hat{\mathbf{Q}}^*_n\\
&\sim&\mathcal{N}\left(\frac{H_{12}}{H_{11}}(n-1)^{-1}\hat{\mathbf{H}}_{22},
(n-1)^{-3}\hat{\mathbf{Q}}^*_n\frac{ \hat{\mathbf{H}}^2_{22}}{ \mathbf{H}_{22}}\right)
\end{eqnarray*}
or, equivalently,
\begin{eqnarray*}
&&\hat{\mathbf{w}}_n^*|(n-1)^{-1}\hat{\mathbf{H}}_{22}, (n-1)^{-1}\hat{\mathbf{Q}}^*_n \nonumber\\
&\sim&\mathcal{N}\left(\frac{H_{12}}{H_{11}},
(n-1)^{-1}\hat{\mathbf{Q}}^*_n\frac{1}{ \mathbf{H}_{22}}\right),
\end{eqnarray*}
where the conditional distribution does not depend on $\hat{\mathbf{H}}_{22}$, i.e., $\hat{\mathbf{w}}_n^*$ and $\hat{\mathbf{Q}}^*_n$ are independent of $\xi_2$. Hence,
\begin{eqnarray}
\hat{\mathbf{w}}_n^*|(n-1)^{-1}\hat{\mathbf{Q}}^*_n
\sim\mathcal{N}\left(\mathbf{w}^*_{GMVP},
\frac{(n-1)^{-1}\hat{\mathbf{Q}}^*_n}{ \mathbf{1}^\prime \mathbf{\Sigma}^{-1}\mathbf{1}}\right). \label{con_norm}
\end{eqnarray}
\end{enumerate}

\vspace{0.1cm}
Let
\begin{eqnarray*}
\xi_1= (n-1) \left(\mathbf{1}^\prime \mathbf{\Sigma}^{-1}\mathbf{1}\right)
\left(\hat{\mathbf{w}}^*_n-\mathbf{r}^*\right)^\prime (\hat{\mathbf{Q}}^*_n)^{-1} \left(\hat{\mathbf{w}}^*_n-\mathbf{r}^*\right).
\end{eqnarray*}
Then, $\xi_1$ and $\xi_2$ are independent, and the application of \eqref{con_norm} leads to
\begin{eqnarray}
\xi_1|(n-1)^{-1}\hat{\mathbf{Q}}^*_n&\sim& \chi^2_{p-1,\lambda_n(\hat{\mathbf{Q}}^*_n)}\nonumber\,,
\end{eqnarray}
with
\begin{eqnarray*}
\lambda_n(\hat{\mathbf{Q}}^*_n)&=&(n-1) \left(\mathbf{1}^\prime \mathbf{\Sigma}^{-1}\mathbf{1}\right)\\
&\times&\left(\mathbf{w}^*_{GMVP}-\mathbf{r}^*\right)^\prime (\hat{\mathbf{Q}}^*_n)^{-1} \left(\mathbf{w}^*_{GMVP}-\mathbf{r}^*\right).
\end{eqnarray*}
\vspace{-0.05cm}
Moreover, in using
$(n-1)(\hat{\mathbf{Q}}^*_n)^{-1} \sim \mathcal{W}_p(n-1,(\mathbf{Q}^*)^{-1} )$ (cf. \citet[Theorems 3.2.10 and 3.2.11]{muirhead2009aspects}), we obtain
{\small
\begin{eqnarray}\label{th1_eq3}
\lambda_n(\hat{\mathbf{Q}}^*_n)&=& \lambda_n
\frac{(n-1)\left(\mathbf{w}^*_{GMVP}-\mathbf{r}^*\right)^\prime (\hat{\mathbf{Q}}^*_n)^{-1} \left(\mathbf{w}^*_{GMVP}-\mathbf{r}^*\right)}
{\left(\mathbf{w}^*_{GMVP}-\mathbf{r}^*\right)^\prime (\mathbf{Q}^*)^{-1} \left(\mathbf{w}^*_{GMVP}-\mathbf{r}^*\right)}\nonumber\\
 &\stackrel{d}{=}& \lambda_n \xi_3\,,
\end{eqnarray}
where $\xi_3\sim \chi^2_{n-1}$.}

The last equality shows that the conditional distribution of $\xi_1$ given $\hat{\mathbf{Q}}^*_n$ depends only on $\hat{\mathbf{Q}}^*_n$ over $\xi_3$, and, consequently, the conditional distribution $\xi_1|\hat{\mathbf{Q}}^*_n$ coincides with $\xi_1|\xi_3$. Using the distributional properties of the non-central $F$-distribution, we obtain the following stochastic representation for $\xi_1$ given by
\[\xi_1 \stackrel{d}{=} (\sqrt{\lambda_n \xi_3}+\omega_1)^2+\xi_4\,,\]
and, hence,
\begin{equation*}
T_n =\frac{n-p}{p-1} \frac{\xi_1}{\xi_2}\stackrel{d}{=} \frac{n-p}{p-1} \frac{(\sqrt{\lambda_n \xi_3}+\omega_1)^2+\xi_4}{\xi_2},
\end{equation*}
where $\omega_1\sim \mathcal{N}(0,1)$, $\xi_2\sim \chi^2_{n-p}$, $\xi_3\sim \chi^2_{n-1}$, and $\xi_4\sim \chi^2_{p-2}$; $\omega_1$, $\xi_2$, $\xi_3$, and $\xi_4$ are independent.
\end{proof}

\begin{proof}\textit{of Theorem \ref{th1}:}
Applying \eqref{stoch_pres_Tn} of Lemma \ref{lem0}, we get
{\footnotesize
\begin{eqnarray*}
&&\frac{n-p}{\xi_2}\sqrt{p-1}\Bigg[\frac{\lambda_n \xi_3+2\sqrt{\lambda_n \xi_3}\omega_1+\omega_1^2+\xi_4}{p-1}\\
&-&\left(1+\lambda_n\frac{n-1}{p-1}\right)\frac{\xi_2}{n-p}\Bigg]\\
&=&\frac{n-p}{\xi_2}\Bigg[\lambda_n \frac{n-1}{p-1}\sqrt{p-1}\left(\frac{\xi_3}{n-1}-1\right)\\
&+&\sqrt{p-1}\left(\frac{\xi_4}{p-1}-1\right)-\left(1+\lambda_n\frac{n-1}{p-1}\right)\sqrt{p-1}\left(\frac{\xi_2}{n-p}-1\right)\\
&+&2\sqrt{\lambda_n}\sqrt{\frac{\xi_3}{p-1}}\omega_1+\frac{\omega_1^2}{\sqrt{p-1}}\Bigg]\,.
\end{eqnarray*}
}
Using the asymptotic properties of a $\chi^2$-distribution with infinite degrees of freedom and the independence of $\omega_1$, $\xi_2$, $\xi_3$, $\xi_4$, the application of Slutsky's lemma (see, for example, Theorem 1.5 in \citet*{dasgupta2008}) leads to
\begin{equation*}
\sqrt{p-1}\left(\frac{T_n-1-\lambda_n\frac{n-1}{p-1}}{C_n}\right) \stackrel{d}{\to} \mathcal{N}\left(0,1\right),
\end{equation*}
where
	\begin{equation*}
C^2_n= 	2+2\frac{\lambda_n^2}{c}+4\frac{\lambda_n}{c}+2\frac{c}{1-c}\left(1+\frac{\lambda_n}{c}\right)^2.
\end{equation*}
\end{proof}

In order to stress the dependence on $n$, we use the notation $\mathbf{\Sigma}_n$ in the proofs of the asymptotic results. For the proof of Theorem \ref{th2} we apply Lemma \ref{lem1}. It must be mentioned that Proposition 3 of \citet*{glombeck} is not fully correct that is why we can not use this result in proving Lemma \ref{lem1}.
\begin{lemma} \label{lem1}
Let
\begin{eqnarray*}
            D_n &=& \frac{\mathbf{b}_n'\mathbf{\hat{\Sigma}}_{n}\mathbf{b}_n}{\mathbf{b}_n'\mathbf{\Sigma}_{n}\mathbf{b}_n} - 1,\\
  E_n& =&  \frac{\mathbf{1}'\mathbf{\hat{\Sigma}}_{n}^{-1}\mathbf{1}}{\mathbf{1}'\mathbf{\Sigma}_{n}^{-1}\mathbf{1}} - \frac{1}{1-c_n}\,.
\end{eqnarray*}
and denote the unit norm vectors
\begin{eqnarray*}
  \mathbf{x}&=& \frac{\mathbf{\Sigma}^{1/2}_{n}\mathbf{b}_n}{\sqrt{\mathbf{b}_n'\mathbf{\Sigma}_{n}\mathbf{b}_n}},~~~  \mathbf{y}= \frac{\mathbf{\Sigma}_{n}^{-1/2}\mathbf{1}}{\sqrt{\mathbf{1}'\mathbf{\Sigma}_{n}^{-1}\mathbf{1}}}.
\end{eqnarray*}

Then, under the assumptions of Theorem \ref{th2} it holds that
	\begin{eqnarray*}\label{glombeck}
         && \sqrt{n}\left( \begin{matrix}
	D_n\\
	E_n
	\end{matrix}\right)\stackrel{d}{\to}\mathcal{N}\left[\left( \begin{matrix}
	0\\0
	\end{matrix} \right) , 2 \left(\begin{matrix}
	1 & -\frac{\lim\limits_{n\to\infty}(\mathbf{x}'\mathbf{y})^2}{1-c}\\
	-\frac{\lim\limits_{n\to\infty}(\mathbf{x}'\mathbf{y})^2}{1-c} & \frac{1}{(1-c)^3}
	\end{matrix} \right) \right]
	\end{eqnarray*}
	for $\dfrac{p}{n} \to c <1 $ as $n \to \infty$.
\end{lemma}

\begin{proof} \textit{of Lemma \ref{lem1}:} \,
Noting that $\mathbf{\hat{\Sigma}}_n\overset{d}{=}\mathbf{\Sigma}_n^{1/2}\mathbf{S}_n\mathbf{\Sigma}^{1/2}_n$ with $\mathbf{S}_n\sim W(n-1, \mathbf{I})$, the result of Lemma \ref{lem1} follows by the direct application of Theorem 3 in \citet*{bai2011}, where it was proven that for $\dfrac{p}{n} \to c <1 $ as $n \to \infty$ the following result holds
  \begin{eqnarray*}
    \sqrt{n}\left( \begin{matrix} \displaystyle
	\mathbf{x}'\mathbf{S}_n\mathbf{x}-1\\
	\mathbf{y}'\mathbf{S}_n^{-1}\mathbf{y}-\frac{1}{1-c_n}
	\end{matrix}\right)\stackrel{d}{\to}\mathcal{N}\left(\mathbf{0} , \frac{2}{c}\boldsymbol{\Theta}_{\mathbf{x, y}}\circ\boldsymbol{\Omega}_c \right),
  \end{eqnarray*}
  where
  {\small
\begin{eqnarray*}
  \boldsymbol{\Theta}_{\mathbf{x, y}}&=&\left(\begin{matrix}
	\lim\limits_{n\to\infty}(\mathbf{x}'\mathbf{x})^2 & \lim\limits_{n\to\infty}(\mathbf{x}'\mathbf{y})^2 \\
	\lim\limits_{n\to\infty}(\mathbf{x}'\mathbf{y})^2 & \lim\limits_{n\to\infty}(\mathbf{y}'\mathbf{y})^2
      \end{matrix} \right)\\
  &=&\left(\begin{matrix}
	1 & \lim\limits_{n\to\infty}(\mathbf{x}'\mathbf{y})^2 \\
	\lim\limits_{n\to\infty}(\mathbf{x}'\mathbf{y})^2 & 1
      \end{matrix} \right),\\
  \boldsymbol{\Omega}_c&=&
\left(
\begin{array}{cc}
\omega_{c,11}& \omega_{c,12}\\
\omega_{c,12}& \omega_{c,22}\\
\end{array}
  \right)
  \end{eqnarray*}
with
\begin{eqnarray*}
\omega_{c,11}&=& \int z^2\text{d}F_c(z)-\left(\int z\text{d}F_c(z) \right)^2 ,\\
\omega_{c,12}&=&  1-\int z\text{d}F_c(z) \int \frac{1}{z}\text{d}F_c(z),\\
\omega_{c,22}&=& \int \frac{1}{z^2}\text{d}F_c(z)-\left(\int \frac{1}{z}\text{d}F_c(z) \right)^2\\
\end{eqnarray*}
}
and the symbol $\circ$ denotes the Hadamard (entrywise) product of matrices. The function $F_c(z)$ denotes the cumulative distribution function of the Marchenko-Pastur law (see, \citet*{bai2010spectral}) for $c<1$ given by
\begin{eqnarray*}
\text{d}  F_c(z) = \frac{1}{2\pi z c}\sqrt{(a_{+}-z)(z-a_{-})}\mathbbm{1}_{[a_{-}, a_{+}]}(z)dz,
\end{eqnarray*}
where $a_{\pm}=(1\pm\sqrt{c})^2$.
 The moments of $F_c(z)$ given in the matrix $\boldsymbol{\Omega}_c$ are already calculated in \citet[Lemma 14]{glombeck} and, thus, it holds
\begin{eqnarray*}
  \boldsymbol{\Omega}_c&=&  \left(\begin{matrix}
	c & -\frac{c}{1-c} \\
	-\frac{c}{1-c} & \frac{c}{(1-c)^3}
	\end{matrix} \right)\,.
\end{eqnarray*}
  At last, after elementary calculus the result follows.
\end{proof}

\begin{proof} \textit{of Theorem \ref{th2}:} \,
First, the asymptotic distribution of $\hat{R}_{\mathbf{b}_n}$ is derived. We rewrite $\hat{R}_{\mathbf{b}_n}$ as
	\begin{eqnarray*}
		\hat{R}_{\mathbf{b}_n} & = & (1-c_n)\frac{\mathbf{b}_n'\mathbf{\hat{\Sigma}}_{n}\mathbf{b}_n}{\mathbf{b}_n'\mathbf{\Sigma}_{n}\mathbf{b}_n} \frac{\mathbf{1}'\mathbf{\hat{\Sigma}}_{n}^{-1}\mathbf{1}}{\mathbf{1}'\mathbf{\Sigma}_{n}^{-1}\mathbf{1}}\mathbf{b}_n'\mathbf{\Sigma}_{n}\mathbf{b}_n\mathbf{1}'\mathbf{\Sigma}_{n}^{-1}\mathbf{1} -1\\
		& = & \Delta_n (1-c_n) (D_n E_n + \frac{D_n}{1-c_n} + E_n) + \Delta_n - 1
	\end{eqnarray*}
	with
          \begin{eqnarray*}
            \Delta_n &=& \mathbf{b}_n'\mathbf{\Sigma}_{n}\mathbf{b}_n \; \mathbf{1}'\mathbf{\Sigma}_{n}^{-1}\mathbf{1}.~~
          \end{eqnarray*}	
          Then, it follows that
          {\small
	\begin{eqnarray*}
          &&	\sqrt{n} \frac{\hat{R}_{\mathbf{b}_n} - \Delta_n +1}{\Delta_n} \\
          &=&  \sqrt{n} (1-c_n) (D_n E_n + D_n/(1-c_n) + E_n)\\
		& = & (1-c_n) \sqrt{n} ( D_n/(1-c_n) + E_n) + o_p(1)\\
          & = & ( 1  \quad 1-c_n ) \; \sqrt{n} \; \left( \begin{matrix}\displaystyle D_n\\\displaystyle E_n \end{matrix} \right) + o_p(1)\\
          &\stackrel{d}{\to}& \mathcal{N}\left( 0, 2 \left(\frac{2-c}{1-c}-\frac{2}{\Delta_n}\right)\right),
	\end{eqnarray*}}
where the last equality follows from Lemma \ref{lem1}.

          Since
        {\small
	\[ \hat{\tilde{\alpha}}_n = \frac{(1-c_n) \left( \frac{\hat{R}_{\mathbf{b}_n} - \Delta_n +1}{\Delta_n} + \frac{\Delta_n - 1}{\Delta_n} \right) }{\frac{c_n}{\Delta_n} + (1-c_n) \left( \frac{\hat{R}_{\mathbf{b}_n} - \Delta_n +1}{\Delta_n} + \frac{\Delta_n - 1}{\Delta_n} \right) } \,,\]	
	it follows that
		\[ \sqrt{n} \left( \hat{\tilde{\alpha}}_n - \frac{ (1-c_n) (\Delta_n - 1)}{c_n + (1-c_n) (\Delta_n - 1)} \right) = I_n + II_n \]
	with}
        {\small
	\begin{eqnarray*}
		I_n & = &   \sqrt{n} \; \frac{(1-c_n) \left( \frac{\hat{R}_{\mathbf{b}_n} - \Delta_n +1}{\Delta_n}  \right) }{\frac{c_n}{\Delta_n} + (1-c_n) \left( \frac{\hat{R}_{\mathbf{b}_n} - \Delta_n +1}{\Delta_n} + \frac{\Delta_n - 1}{\Delta_n} \right) } \\
		& = & \sqrt{n} \; \frac{(1-c_n) \left( \frac{\hat{R}_{\mathbf{b}_n} - \Delta_n +1}{\Delta_n}  \right) }{\frac{c_n}{\Delta_n} + (1-c_n)  \frac{\Delta_n - 1}{\Delta_n}  }  \\
&\times& \frac{\frac{c_n}{\Delta_n} + (1-c_n)  \frac{\Delta_n - 1}{\Delta_n}  }{\frac{c_n}{\Delta_n} + (1-c_n) \left( \frac{\hat{R}_{\mathbf{b}_n} - \Delta_n +1}{\Delta_n} + \frac{\Delta_n - 1}{\Delta_n} \right) }\\
		& = & \sqrt{n} \; \frac{(1-c_n) \left( \frac{\hat{R}_{\mathbf{b}_n} - \Delta_n +1}{\Delta_n}  \right) }{1 - c_n - \frac{1-2 c_n}{\Delta_n}} \\
&\times& \frac{1}{1 + \frac{1-c_n}{\sqrt{n} (1 - c_n - \frac{1-2 c_n}{\Delta_n})} \sqrt{n} \frac{\hat{R}_{\mathbf{b}_n} - \Delta_n +1}{\Delta_n}} .
\end{eqnarray*}
	}
	
Furthermore,{\small
	\begin{eqnarray*}
          II_n & = & \sqrt{n} (1 - c_n) (1 - \frac{1}{\Delta_n}) \\
          &\times&\left( \frac{1}{\frac{c_n}{\Delta_n} + (1-c_n) \left( \frac{\hat{R}_{\mathbf{b}_n} - \Delta_n +1}{\Delta_n} + \frac{\Delta_n - 1}{\Delta_n} \right) }\right.\\
          &-&\left. \frac{1}{\frac{c_n}{\Delta_n} + (1-c_n)  \frac{\Delta_n - 1}{\Delta_n}  } \right)\\
		& = &  \sqrt{n} \frac{1 - c_n}{1 - c_n - \frac{1-2 c_n}{\Delta_n}} \; (1 - \frac{1}{\Delta_n}) \\
&\times& \left( \frac{1}{1 + \frac{1-c_n}{1 - c_n - \frac{1-2c_n}{\Delta_n}}  \frac{\hat{R}_{\mathbf{b}_n} - \Delta_n +1 }{\Delta_n} }
		- 1 \right) \\
		& = & - \frac{(1 - c_n)^2}{(1 - c_n - \frac{1-2 c_n}{\Delta_n})^2} \; \left(1 - \frac{1}{\Delta_n}\right) \\
&\times&
		\frac{\sqrt{n} \frac{ \hat{R}_{\mathbf{b}_n} - \Delta_n +1 }{\Delta_n} }{1 + \frac{1-c_n}{\sqrt{n} (1 - c_n - \frac{1-2 c_n}{\Delta_n} ) } \sqrt{n} \frac{ \hat{R}_{\mathbf{b}_n} - \Delta_n +1 }{\Delta_n} }  .
	\end{eqnarray*}
	}

Consequently, if $\sqrt{n} (1 - c_n - \frac{1-2 c_n}{\Delta_n} ) \rightarrow \infty$ as $n \to \infty$ then {\small
	\begin{eqnarray*}
		I_n + II_n & = & \sqrt{n} \frac{ \hat{R}_{\mathbf{b}_n} - \Delta_n +1 }{\Delta_n}
                                 \frac{c_n(1 - c_n)}{(1 - c_n - \frac{1-2 c_n}{\Delta_n})^2} \; \frac{1}{\Delta_n} \;\\
          &\times&	\frac{1}{1 + \frac{1-c_n}{\sqrt{n} (1 - c_n - \frac{1-2 c_n}{\Delta_n} ) } \sqrt{n} \frac{ \hat{R}_{\mathbf{b}_n} - \Delta_n +1 }{\Delta_n} } \\
		& = & \sqrt{n} \frac{ \hat{R}_{\mathbf{b}_n} - \Delta_n +1 }{\Delta_n}
                      \frac{c_n(1 - c_n) \Delta_n}{(c_n + (\Delta_n - 1) (1 - c_n))^2} \;\\
          &\times& (1 + o_p(1))\\
        \end{eqnarray*}

        \vspace{-1cm}

        \begin{equation*}
        \stackrel{d}{\approx}  {\cal N}\left( 0, 2 \frac{c_n^2 (1-c_n)(2-c_n) \Delta_n}{(c_n + (\Delta_n - 1) (1 - c_n))^4}\left(\Delta_n+\frac{2(c_n-1)}{2-c_n} \right)\right),
        \end{equation*}}
where we have used the equality $$\frac{2-c_n}{1-c_n}-\frac{2}{\Delta_n}=\frac{2-c_n}{1-c_n}\frac{1}{\Delta_n}\left(\Delta_n+\frac{2(c_n-1)}{2-c_n}\right).$$
	Since $\mathbf{b}_n'\mathbf{\Sigma}_{n}\mathbf{b}_n \, \geq \, \min\limits_{\substack{\bf{w}}} \mathbf{w}'\mathbf{\Sigma}_{n}\mathbf{w} =\displaystyle \frac{1}{\mathbf{1}'_n\mathbf{\Sigma}_n^{-1}\mathbf{1}_n} $, it holds that $\Delta_n \geq 1$, and, thus, the condition $\lim\limits_{n\to\infty}\sqrt{n} (1 - c_n - \frac{1-2 c_n}{\Delta_n} ) \rightarrow \infty$ is always fulfilled. Taking into account the relation $\Delta_n=R_{\mathbf{b}_n}+1$ the proof of Theorem \ref{th2} is finished.
\end{proof}

In the proof of Theorem \ref{th3} we use the following two lemmas. Lemma \ref{lem2} extends the results of Theorem 1 in \citet*{bodnar2016singular} to the case $n>p$, while Lemma \ref{lem3} presents a stochastic representation of $\tilde{T}_n$ similarly to the statement of Lemma \ref{lem0} in the non-singular case.

\begin{lemma}\label{lem2}
Let $\mathbf{V} \sim W_p(N,\mathbf{\Sigma})$ with $rank(\mathbf{\Sigma})=q \le N$ and let $\mathbf{L}: k \times p$ be a matrix of constants of rank $k\le q$. Then it holds that
$$
(\mathbf{L}\mathbf{V}^+\mathbf{L}')^{-1} \sim W_{k}\left(n-q+k,(\mathbf{L}\mathbf{\Sigma}^+\mathbf{L}')^{-1}\right).
$$
\end{lemma}

\begin{proof}\textit{of Lemma \ref{lem2}:}
The stochastic representation of $\mathbf{V}$ is expressed as
\begin{equation}\label{pth1_eq1}
\mathbf{V} \stackrel{d}{=} \mathbf{Y}\mathbf{Y}' ~~ \text{with} ~~ \mathbf{Y} \sim \mathcal{N}_{p,N}(\mathbf{0},\mathbf{\Sigma} \otimes \mathbf{I}_n)\,.
\end{equation}

Let $\mathbf{\Sigma}=\mathbf{Q} \mathbf{\Lambda} \mathbf{Q}'$ be the singular value decomposition of $\mathbf{\Sigma}$ where $\mathbf{\Lambda}: q\times q$ is the matrix of non-zero eigenvalues and $\mathbf{Q}:p\times q$ is the semi-orthogonal matrix of the corresponding eigenvectors, i.e., $\mathbf{Q}'\mathbf{Q}=\bI_q$. Then the stochastic representation of $\mathbf{Y}$ is given by
\begin{equation}\label{pth1_eq2}
\mathbf{Y} \stackrel{d}{=} \mathbf{Q} \mathbf{\Lambda}^{1/2}\mathbf{Z} ~~ \text{with} ~~ \mathbf{Z} \sim \mathcal{N}_{q,N}(\mathbf{0},\mathbf{I}_q \otimes \mathbf{I}_n)\,,
\end{equation}
and, hence,
\begin{equation}\label{stoch_singular}
\mathbf{V} \stackrel{d}{=} \mathbf{Q} \mathbf{\Lambda}^{1/2}\mathbf{Z}\mathbf{Z}' \mathbf{\Lambda}^{1/2} \mathbf{Q}' \,,
\end{equation}
where $\mathbf{Z}\mathbf{Z}' \sim \mathcal{W}_q(N,\mathbf{I}_q)$.

Since $\mathbf{Q} \mathbf{\Lambda}^{1/2}$ is a full column-rank matrix and $\mathbf{\Lambda}^{1/2} \mathbf{Q}'$ is a full row-rank matrix, we get
\begin{eqnarray}\label{pth1_eq3}
\mathbf{L}\mathbf{V}^+\mathbf{L}'&\stackrel{d}{=}& \mathbf{L}\left(\mathbf{Q} \mathbf{\Lambda}^{1/2}\mathbf{Z}\mathbf{Z}' \mathbf{\Lambda}^{1/2} \mathbf{Q}'\right)^+\mathbf{L}'\nonumber\\
&=& \mathbf{L}\mathbf{Q} \mathbf{\Lambda}^{-1/2}\left(\mathbf{Z}\mathbf{Z}'\right)^+ \mathbf{\Lambda}^{-1/2} \mathbf{Q}'\mathbf{L}'\nonumber\\
&=& \mathbf{L}\mathbf{Q} \mathbf{\Lambda}^{-1/2}\left(\mathbf{Z}\mathbf{Z}'\right)^{-1} \mathbf{\Lambda}^{-1/2} \mathbf{Q}'\mathbf{L}'\,,
\end{eqnarray}
because $\mathbf{Z}\mathbf{Z}'$ is non-singular (cf., \citet*{Greville1966}). Finally, the application of Theorem 3.2.11 in \citet*{muirhead2009aspects} leads to
\begin{eqnarray*}
(\mathbf{L}\mathbf{V}^+\mathbf{L}')^{-1}&\sim& \mathcal{W}_k\left(N-q+k,\left(\mathbf{L}\mathbf{Q} \mathbf{\Lambda}^{-1/2}\mathbf{I}_q \mathbf{\Lambda}^{-1/2} \mathbf{Q}'\mathbf{L}'\right)^{-1}\right) \\
&=&\mathcal{W}_{k}\left(N-q+k,(\mathbf{L}\mathbf{\Sigma}^+\mathbf{L}^\prime)^{-1}\right)\,.
\end{eqnarray*}
\end{proof}

\begin{lemma}\label{lem3}
Under the conditions of Theorem \ref{th3}, the stochastic representation of $\tilde{T}_n$ is expressed as
\begin{equation}\label{stoch_pres_Tn_sing}
\tilde{T}_n \stackrel{d}{=} \frac{n-q}{k} \frac{(\sqrt{\tilde{\lambda}_n \xi_3}+\omega_1)^2+\xi_4}{\xi_2},
\end{equation}
where $\omega_1\sim \mathcal{N}(0,1)$, $\xi_2\sim \chi^2_{n-q}$, $\xi_3\sim \chi^2_{n-q+k}$, and $\xi_4\sim \chi^2_{k-1}$; $\omega_1$, $\xi_2$, $\xi_3$, and $\xi_4$ are independent.
\end{lemma}

\begin{proof}\textit{of Lemma \ref{lem3}:}
Let $\mathbf{M}^\prime = (\mathbf{L}^\prime, \mathbf{1})$ and define
$$\mathbf{M} \mathbf{\Sigma}^{+} \mathbf{M}^\prime = \{\mathbf{H}_{ij}\}_{i,j=1,2}, \quad \mathbf{M} \, \hat{\mathbf{\Sigma}}^{+}_n \,
\mathbf{M}^{\prime} = \{\hat{\mathbf{H}}_{ij}\}_{i,j=1,2}$$
with
$\mathbf{{H}}_{22} = \mathbf{1}^\prime \mathbf{\Sigma}^{+}\mathbf{1}$,
$\hat{\mathbf{H}}_{22} = \mathbf{1}^\prime \hat{\mathbf{\Sigma}}^{+}_n \mathbf{1}$,
$\mathbf{H}_{12}=\mathbf{L} \mathbf{\Sigma}^{+} \mathbf{1}$, $\hat{\mathbf{H}}_{12}=\mathbf{L}
\hat{\mathbf{\Sigma}}^{+}_n\mathbf{1}$, $\mathbf{H}_{11} = \mathbf{L} \mathbf{\Sigma}^{+} \mathbf{L}^\prime$,
and $\mathbf{\hat{H}}_{11} = \mathbf{L} \hat{\mathbf{\Sigma}}^{+}_n \mathbf{L}^\prime$.

Since $(n-1)\hat{\mathbf{\Sigma}} \sim W_p(n-1, \mathbf{\Sigma})$ and $rank(\mathbf{M})=k+1\le q$, the application of Lemma \ref{lem2} leads to
\[ (n-1) ( \mathbf{M} \hat{\mathbf{\Sigma}}^{+}_n \mathbf{M}^\prime )^{-1} \sim W_{k+1}(n-q+k, ( \mathbf{M} \,
\mathbf{\Sigma}^{+} \, \mathbf{M}^\prime )^{-1} ), \]
and, consequently, $(n-1)^{-1} \mathbf{M} \hat{\mathbf{\Sigma}}^{+}_n \mathbf{M}^\prime$ has a non-singular Wishart distribution given by
\[ (n-1)^{-1} \mathbf{M} \hat{\mathbf{\Sigma}}^{+}_n \mathbf{M}^\prime \sim W_{k+1}^{-1}(n-q+2k+2, \mathbf{M} \,
\mathbf{\Sigma}^{+} \, \mathbf{M}^\prime ). \]

\vspace{0.15cm}
Let
\begin{eqnarray*}
\xi_1 &=& (n-1) \left(\mathbf{1}^\prime \mathbf{\Sigma}^{+}\mathbf{1}\right)
\left(\hat{\tilde{\mathbf{w}}}^*_n-\tilde{\mathbf{r}}^*\right)^\prime (\hat{\tilde{\mathbf{Q}}}^*_n)^{-1} \left(\hat{\tilde{\mathbf{w}}}^*_n-\tilde{\mathbf{r}}^*\right),\\
\xi_2 &=& (n-1) \frac{\mathbf{1}^\prime \mathbf{\Sigma}^{+}\mathbf{1}}{\mathbf{1}^\prime \hat{\mathbf{\Sigma}}^{+}_n\mathbf{1}},
\end{eqnarray*}
where $\hat{\tilde{\mathbf{w}}}^*_n$ and $\hat{\tilde{\mathbf{Q}}}^*_n$ are defined in Section III.C.

Since $(n-1)^{-1} \mathbf{M} \hat{\mathbf{\Sigma}}^{+}_n \mathbf{M}^\prime$ has a non-singular Wishart distribution, following the proof of Lemma \ref{lem0}, we get that $\xi_1$ and $\xi_2$ are independent, $\xi_2\sim \chi^2_{n-q}$, and
\begin{eqnarray}
\xi_1|(n-1)^{-1}\hat{\tilde{\mathbf{Q}}}^*_n&\sim& \chi^2_{k,\tilde{\lambda}_n(\hat{\tilde{\mathbf{Q}}}^*_n)}\nonumber\,,
\end{eqnarray}
with
\begin{eqnarray*}
\tilde{\lambda}_n(\hat{\tilde{\mathbf{Q}}}^*_n)&=&(n-1) \left(\mathbf{1}^\prime \mathbf{\Sigma}^{+}\mathbf{1}\right)\\
&\times&\left(\tilde{\mathbf{w}}^*_{GMVP}-\tilde{\mathbf{r}}^*\right)^\prime (\hat{\tilde{\mathbf{Q}}}^*_n)^{-1} \left(\tilde{\mathbf{w}}^*_{GMVP}-\tilde{\mathbf{r}}^*\right)
\end{eqnarray*}
and $\tilde{\lambda}_n(\hat{\tilde{\mathbf{Q}}}^*_n)\stackrel{d}{=} \lambda_n \xi_3$ where $\xi_3\sim \chi^2_{n-q+k}$ and
\begin{eqnarray*}
\tilde{\lambda}_n=\left(\mathbf{1}^\prime \mathbf{\Sigma}^{+}\mathbf{1}\right)
\left(\tilde{\mathbf{w}}^*_{GMVP}-\tilde{\mathbf{r}}^*\right)^\prime (\tilde{\mathbf{Q}}^*)^{-1} \left(\tilde{\mathbf{w}}^*_{GMVP}-\tilde{\mathbf{r}}^*\right).
\end{eqnarray*}

Hence,
\begin{equation*}
\tilde{T}_n =\frac{n-q}{k} \frac{\xi_1}{\xi_2}\stackrel{d}{=} \frac{n-q}{k} \frac{(\sqrt{\tilde{\lambda}_n \xi_3}+\omega_1)^2+\xi_4}{\xi_2},
\end{equation*}
where $\omega_1\sim \mathcal{N}(0,1)$, $\xi_2\sim \chi^2_{n-q}$, $\xi_3\sim \chi^2_{n-q+k}$, and $\xi_4\sim \chi^2_{k-1}$; $\omega_1$, $\xi_2$, $\xi_3$, and $\xi_4$ are independent.
\end{proof}

\begin{proof}\textit{of Theorem \ref{th3}:}
Applying \eqref{stoch_pres_Tn_sing} of Lemma \ref{lem3}, we get
{\footnotesize
\begin{eqnarray*}
&&\frac{n-q}{\xi_2}\sqrt{k}\Bigg[\frac{\tilde{\lambda}_n \xi_3+2\sqrt{\tilde{\lambda}_n \xi_3}\omega_1+\omega_1^2+\xi_4}{k}\\
&-&\left(1+\tilde{\lambda}_n\frac{n-q+k}{k}\right)\frac{\xi_2}{n-q}\Bigg]\\
&=&\frac{n-q}{\xi_2}\Bigg[\tilde{\lambda}_n \frac{n-q+k}{k}\sqrt{k}\left(\frac{\xi_3}{n-q+k}-1\right)+\sqrt{k}\left(\frac{\xi_4}{k}-1\right)\\
&-&\left(1+\tilde{\lambda}_n\frac{n-q+k}{k}\right)\sqrt{k}\left(\frac{\xi_2}{n-q}-1\right)\\
&+&2\sqrt{\tilde{\lambda}_n}\sqrt{\frac{\xi_3}{k}}\omega_1+\frac{\omega_1^2}{\sqrt{k}}\Bigg]\,.
\end{eqnarray*}
}
Using the asymptotic properties of a $\chi^2$-distribution with infinite degrees of freedom and the independence of $\omega_1$, $\xi_2$, $\xi_3$, $\xi_4$, the application of Slutsky's lemma (see, for example, Theorem 1.5 in \citet*{dasgupta2008}) leads to
\begin{equation*}
\sqrt{k}\left(\frac{\tilde{T}_n-1-\tilde{\lambda}_n\frac{n-q+k}{k}}{\tilde{C}_n}\right) \stackrel{d}{\to} \mathcal{N}\left(0,1\right),
\end{equation*}
where
	\begin{eqnarray*}
\tilde{C}^2_n&=&2+2\frac{(1-\tilde{c}+\tilde{b})\tilde{\lambda}_n^2}{\tilde{b}}+4\frac{(1-\tilde{c}+\tilde{b})\tilde{\lambda}_n}{\tilde{b}}\\
&+&2\frac{\tilde{b}}{1-\tilde{c}}\left(1+\frac{(1-\tilde{c}+\tilde{b})\tilde{\lambda}_n}{\tilde{b}}\right)^2.
\end{eqnarray*}
\end{proof}

In order to proof Proposition \ref{alpha+} we need the following lemma, which is a special case of \citet[Theorem 1]{rubio2011}.
\begin{lemma}\label{lemmaRM}
Let a nonrandom $q\times q$-dimensional matrix $\mathbf{\Theta}_q$ possesses a uniformly bounded trace norm (sum of singular values) and let $\mathbf{S}_N\sim W(N, \bI_q)$. Then it holds that
\begin{equation*}
\left|\text{tr}\left(\mathbf{\Theta}_q(\mathbf{S}_N-z\bI_q)^{-1}\right)-(x(z)-z)^{-1}\text{tr}\left(\mathbf{\Theta}_q\right)\right|\stackrel{a.s.}{\longrightarrow}0\,
\end{equation*}
for $q/N\longrightarrow \tilde{c} \in (0, +\infty)$ as $N\rightarrow\infty$, where
\begin{equation}\label{RM2011_id_xz}
x(z)=\dfrac{1}{2}\left(1-\tilde{c}+z+\sqrt{(1-\tilde{c}+z)^2-4z}\right)\,.
\end{equation}
\end{lemma}

\begin{proof}\textit{of Proposition \ref{alpha+}:}
  The proof is similar to the proof of Theorem 2.1 by \citet*{bodnar2014estimation} with a few important modifications due to the singularity of $\bSigma_n$. Indeed, taking into account the equality \eqref{stoch_singular} we have
  \begin{eqnarray*}
    (n-1)\hat{\mathbf{\Sigma}}_n\overset{d}{=}\mathbf{Q} \mathbf{\Lambda}^{1/2}\mathbf{Z}_n\mathbf{Z}_n' \mathbf{\Lambda}^{1/2} \mathbf{Q}' \,,
  \end{eqnarray*}
  where the Wishart matrix $\mathbf{Z}_n\mathbf{Z}_n'\sim \mathcal{W}_q(n-1,\mathbf{I}_q)$ is nonsingular. Using now the properties of the Moore-Penrose inverse and \eqref{pth1_eq3} we get
  \begin{eqnarray}\label{S+}
    \hat{\mathbf{\Sigma}}^+_n&\overset{d}{=}&\mathbf{Q} \mathbf{\Lambda}^{-1/2}\left(\frac{1}{n-1}\mathbf{Z}_n\mathbf{Z}_n'\right)^{-1} \mathbf{\Lambda}^{-1/2} \mathbf{Q}'\nonumber\\
    &=&\tilde{\mathbf{\Sigma}}^{-1/2}_n\left(\frac{1}{n-1}\mathbf{Z}_n\mathbf{Z}_n'\right)^{-1}\tilde{\mathbf{\Sigma}}^{-1/2\;'}_n\,.
  \end{eqnarray}
  Moreover,  note the following identities
  \begin{equation}
    \label{eq:MP}
\bSigma_n=\tilde{\mathbf{\Sigma}}^{1/2}_n\tilde{\mathbf{\Sigma}}^{1/2\;'}_n~\text{and}~  \mathbf{\Sigma}^+_n=\tilde{\mathbf{\Sigma}}^{-1/2}_n\tilde{\mathbf{\Sigma}}^{-1/2\;'}_n.
  \end{equation}
 Recall the optimal shrinkage intensity expressed as
 \begin{eqnarray}\label{alfa_app}
  \hat{\alpha}_n^+&=& \dfrac{\mathbf{b}_n^{\prime}\bSigma_n\mathbf{b}_n-\dfrac{\bi^\prime\bS_n^{+}\bSigma_n\mathbf{b}_n}{\bi^\prime\bS_n^{+}\bi}}
 {\dfrac{\bi^\prime\bS_n^{+}\bSigma_n\bS_n^{+}\bi}{(\bi^\prime\bS_n^{+}\bi)^2}-2\dfrac{\bi^\prime\bS_n^{+}\bSigma_n\mathbf{b}_n}
 {\bi^\prime\bS_n^{+}\bi}+\mathbf{b}_n^{\prime}\bSigma_n\mathbf{b}_n}\,.
\end{eqnarray}
 Due to \eqref{S+} and \eqref{eq:MP} it holds that for all $z\in\mathbbm{C}^+=\{\tilde{z}\in\mathbbm{C}: \Im(\tilde{z})>0\}$
{\small
 \begin{eqnarray}
\bi^\prime\bS_n^{+}\bi&=&\left.\text{tr}\left[\left(\dfrac{1}{n-1}\bx_n\bx^\prime_n-z\bI\right)^{-1}
\mathbf{\Theta}_{\xi}\right]\right|_{z=0}
\label{th31_eq11}\\
\bi^\prime\bS_n^{+}\bSigma_n\mathbf{b}_n
&=&\left.\text{tr}\left[\left(\dfrac{1}{n-1}\bx_n\bx^\prime_n-z\bI\right)^{-1}
\mathbf{\Theta}_{\zeta}\right]\right|_{z=0}
\label{th31_eq21}\\
\bi^\prime\bS_n^{+}\bSigma_n\bS_n^{+}\bi
&=&\left.\dfrac{\partial}{\partial z}\text{tr}\left[\left(\frac{1}{n-1}\bx_n\bx^\prime_n-z\bI\right)^{-1} \mathbf{\Theta}_{\xi}\right]\right|_{z=0}
\label{th31_eq31}\,,
\end{eqnarray}}
with $\mathbf{\Theta}_{\xi}=\tilde{\bSigma}_n^{-\frac{1}{2}\;'}\bi\bi^\prime\tilde{\bSigma}_n^{-\frac{1}{2}}$ and $\mathbf{\Theta}_{\zeta}=\tilde{\bSigma}_n^{\frac{1}{2}\;'}\mathbf{b}_n\bi^\prime\tilde{\bSigma}_n^{-\frac{1}{2}}$. The symbol $\left.\cdot\right|_{z=0}$ stays for the limit $z\to0$.

Let
\begin{eqnarray*}
\xi_n(z)&=&\text{tr}\left[\left(\dfrac{1}{n-1}\bx_n\bx^\prime_n-z\bI\right)^{-1}\mathbf{\Theta}_{\xi}\right],\\
\zeta_n(z)&=&\text{tr}\left[\left(\dfrac{1}{n-1}\bx_n\bx^\prime_n-z\bI\right)^{-1}\mathbf{\Theta}_{\zeta}\right]\,.
\end{eqnarray*}
where both matrices $\mathbf{\Theta}_{\xi} $ and $\mathbf{\Theta}_{\zeta}$ possess a bounded trace norm since
\begin{eqnarray*}
\|\mathbf{\Theta}_{\xi}\|_{tr}
&=&\bi^\prime \bSigma^{+}_n\bi\le M_l^{-1}~~\text{and}\\
\|\mathbf{\Theta}_{\zeta}\|_{tr}
&=&\sqrt{\bi^\prime \bSigma^{+}_n\bi}\sqrt{\mathbf{b}_n^\prime \bSigma_n\mathbf{b}_n}\le \sqrt{\dfrac{M_u}{M_l}} \,.
\end{eqnarray*}
Then, for all $z\in\mathbbm{C}^+$,  we get from Lemma \ref{lemmaRM}
\begin{eqnarray}\label{th31_eq41}
  &&|\xi_n(z)-(x(z)-z)^{-1}\text{tr}\left[\mathbf{\Theta}_{\xi}\right]|\nonumber\\
  &=&|\xi_n(z)-(x(z)-z)^{-1}\bi^\prime \bSigma^{+}\bi| \stackrel{a.s.}{\longrightarrow}0\\[0.3cm]
  &&|\zeta_n(z)-(x(z)-z)^{-1}\text{tr}\left[\mathbf{\Theta}_{\zeta}\right]|\nonumber\\
  &=&\left|\zeta_n(z)-(x(z)-z)^{-1}\right|\stackrel{a.s.}{\longrightarrow}0\label{th31_eq5}\,
\end{eqnarray}
for  $q/n\rightarrow \tilde{c} \in (0, 1)$ as $n\rightarrow\infty$,
where $x(z)$ is given in (\ref{RM2011_id_xz}). Using that $\lim\limits_{z\rightarrow0^+}(x(z)-z)^{-1}=(1-\tilde{c})^{-1}$ and combining (\ref{th31_eq41}) and (\ref{th31_eq5}) with (\ref{th31_eq11}) and (\ref{th31_eq21}) leads to
\begin{equation}\label{th31_eq6}
|\bi^\prime\bS_n^{+}\bi-(1-\tilde{c})^{-1} \bi^\prime \bSigma^{+}_n\bi| \stackrel{a.s.}{\longrightarrow}0,
\end{equation}
\begin{equation}\label{th31_eq7}
\left|\bi^\prime\bS_n^{+}\bSigma_n\mathbf{b}_n-(1-\tilde{c})^{-1}\right|\stackrel{a.s.}{\longrightarrow}0
\end{equation}
for $q/n\rightarrow \tilde{c}\in (0, 1)$ as $n\rightarrow\infty$.
Finally, using the equality
\begin{equation*}
\left.\dfrac{\partial}{\partial z}\dfrac{1}{ x(z)-z}\right|_{z=0}=-\left.\dfrac{x^\prime(z)-1}{ (x(z)-z)^2}\right|_{z=0}
= \dfrac{1}{(1-\tilde{c})^3}\,,
\end{equation*}
we get
\begin{eqnarray*}
  &&  \left|\xi_n^\prime(0)-\left.\dfrac{\partial}{\partial z}(x(z)-z)^{-1}\right|_{z=0}\text{tr}\left[\mathbf{\Theta}_{\xi}\right]\right|\nonumber\\
  &=&|\xi^\prime_n(0)-(1-\tilde{c})^{-3}\bi^\prime \bSigma^{+}_n\bi| \stackrel{a.s.}{\longrightarrow}0
\end{eqnarray*}
for $q/n\rightarrow \tilde{c}\in (0, 1)$ as $n\rightarrow\infty$. As a result,
\begin{equation}\label{th31_eq9}
|\bi^\prime\bS_n^{+}\bSigma_n\bS_n^{-1}\bi-(1-\tilde{c})^{-3}\bi^\prime \bSigma^{+}_n\bi| \stackrel{a.s.}{\longrightarrow}0\,
\end{equation}
for $q/n\rightarrow \tilde{c} \in(0, 1)$ as $n\rightarrow\infty$.
At last, the application of (\ref{th31_eq6}), \eqref{th31_eq7} and (\ref{th31_eq9}) to \eqref{alfa_app} implies the result of Proposition \ref{alpha+}.
\end{proof}

\begin{proof}\textit{of Theorem \ref{th4}:}
  Using the proof of Proposition \ref{alpha+} we can immediately deduce that Lemma \ref{lem1} also holds in the case of singular covariance matrix $\bSigma_n$ with the only exception that the usual matrix inverse must be replaced by the Moore-Penrose inverse and $p$ must be replaced by $q$, i.e., $c_n$ becomes $\tilde{c}_n$. That is why the proof of Theorem \ref{th2} can be applied step by step again without any further changes.
\end{proof}

\end{document}